\documentclass[a4paper,english,cleveref,autoref]{lipics-v2019}
\usepackage{amsmath}
\usepackage{amssymb}
\usepackage{latexsym}
\usepackage{graphicx}
\usepackage{color}
\usepackage{amsfonts}
\usepackage{mathrsfs}
\usepackage{stmaryrd}
\usepackage{url}
\usepackage{ulem}
\usepackage{mathpartir}
\usepackage{enumerate}
\usepackage{prooftree}
\usepackage{multicol}
\usepackage{multirow}
\usepackage{framed}
\usepackage{xspace}

\long\def\ignore#1{\relax}

\newcommand{\fubla}[1]{\red{blablabla}} 
\newcommand{\furef}[1]{\red{ref}} 
\newcommand{\fucite}[1]{\red{cite}} 

\newcommand{\bcen}{\begin{center}}
\newcommand{\ecen}{\end{center}}

\newcommand{\compF}[2]{\mathtt{#1}\left(#2\right)}

\newcommand{\cmin}[2]{#1{\setminus\!\!\setminus}\, #2}

\newcommand{\fv}[1]{\mathtt{fv}(#1)}
\newcommand{\bv}[1]{\mathtt{bv}(#1)}

\def\l{\lambda}
\def\Gam{\Gamma}

\def\Del{\Delta}

\def\sig{\sigma}

\newcommand{\sep}{\hspace*{0.5cm}}

\newcommand{\Rewplus}[1]{\rightarrow^{+}_{#1}}
\newcommand{\Rewparam}[2]{\rightarrow^{#1}_{#2}}
\newcommand{\notRew}[1]{{\not\!\!\!\!\Rew{#1}}}

\newcommand{\emm}{ [\,] }
\newcommand{\es}{ \emptyset}

\newcommand{\A}{\mathcal{A}}
\newcommand{\B}{\mathcal{B}}

\newcommand{\R}{\mathcal{R}}

\newcommand{\mult}[1]{[#1]}
\newcommand{\umult}[1]{\langle #1 \rangle}

\newcommand{\deft}[1]{{\bf #1}}

\def\l{\lambda}

\newcommand{\slist}{\mathtt{L}}

\newcommand{\betas}{{\tt b}}
\newcommand{\exps}{{\tt e}}
\newcommand{\matchs}{{\tt m}}

\newcommand{\tri}{\triangleright}

\newcommand{\dom}[1]{\mathtt{dom}(#1)}

\newcommand{\iI}{{i \in I}}

\newcommand{\kK}{{k \in K}}

\newcommand{\sz}[1]{\compF{sz}{#1}}

\newcommand{\app}{{\tt app}}

\newcommand{\appn}{{\tt app_{p}}}
\newcommand{\pattn}{{\tt pat_p}}

\newcommand{\rew}{\rightarrow}

\newcommand{\ttabs}[1]{\mathtt{abs}(#1)}

\newcommand{\seq}[2]{#1 \vdash #2}
\newcommand{\Seq}[2]{#1 \Vdash #2}

\newcommand{\amuju}[6]{#5 \vdash^{(#1,#2,#3,#4)} #6}

\newcommand{\amuJu}[6]{#5 \vdash^{(#1,#2,#3,#4)} #6}

\newcommand{\typend}{\bullet}
\newcommand{\typendm}{\bullet_{\M}}
\newcommand{\typendn}{\bullet_{\N}}

\newcommand{\tight}{\mathtt{t}}

\newcommand{\isnotabs}[1]{\neg \ttabs{#1}}

\newcommand{\istight}[1]{{\tt tight}(#1)}

\newcommand{\sigk}{\sigma_k}                   

\newcommand{\inter}{\wedge}
\newcommand{\interpret}[1]{[ \! [ #1 ] \! ]}

\renewcommand{\bot}{\umult{\,}}

\newcommand{\emul}{\mult{\, }}

\renewcommand{\emph}[1]{{\it  #1}}

\newcommand{\Gamk}{\Gam_k}            

\newcommand{\Delk}{\Del_k}

\newcommand{\many}{\mathtt{many}} 

\newcommand{\PT}{ \mathcal{P} }

\newcommand{\K}{ \mathtt{K} }

\newcommand{\ruletight}{{\tt abs_p}}

\usepackage{ifthen}

\usepackage{mathtools}
\usepackage{yfonts}
\usepackage{appendix}
\usepackage{xcolor}

\renewcommand\l{\lambda}

\newcommand\emptymultiset{\ems}
\newcommand\ems{[\,]}
\newcommand\msunion{\sqcup}

\newcommand\vars{\mathcal{V}}

\newcommand\Deribbase[5]{{#3}\ {\pmb\vdash}_{#2}^{#1} {#4}\  {:}\  {#5}}

\makeatletter

\newcommand{\Deribase}[1]{%
  \def\DeribW[##1]{\Deribbase{##1}{#1}}%
  \def\DeribWO{\Deribbase{}{#1}}%
  \@ifnextchar[\DeribW\DeribWO%
  }

  \newcommand{\Deri}{%
  \def\DeriW_##1{\Deribase{##1}}%
  \def\DeriWO{\Deribase{}}%
  \@ifnextchar_\DeriW\DeriWO%
  }
\makeatother

\newcommand{\result}{r}

\makeatletter
\newcommand{\exder}{%
  \def\exderW[##1]{\triangleright_{##1}\ }%
  \def\exderWO{\triangleright\ }%
  \@ifnextchar[\exderW\exderWO%
  }
\makeatother

\makeatletter
\newcommand{\appresult}{%
  \def\appresultW<##1>{\app_\result^{##1}}%
  \def\appresultWO{\app_\result}%
  \@ifnextchar<\appresultW\appresultWO%
  }
\makeatother

\newcommand{\introarrow}{{\tt abs}}

\newcommand{\tderiv}{\Phi}

\newcommand{\id}{{\tt I}}

\newcommand{\var}[1]{{\tt var}(#1)}

\makeatletter

\newcommand{\Rewbase}{%
  \def\RewbaseW[##1]##2##3{\ {\xrightarrow{##1}}{}_{##2}^{##3}\, }%
  \def\RewbaseWO##1##2{\ {\xrightarrow{}}{}_{##1}^{##2}\, }%
  \@ifnextchar[\RewbaseW\RewbaseWO%
  }

\newcommand{\Rewbasebis}{%
  \def\RewbaseW[##1]##2##3{\ {\xrightarrow{##1}}{}_{##2}^{##3}}%
  \def\RewbaseWO##1##2{\ {\xrightarrow{}}{}_{##1}^{##2}}%
  \@ifnextchar[\RewbaseW\RewbaseWO%
  }

\newcommand{\Rew}[1]{%
  \def\RewW[##1]{\Rewbase[##1]{#1}{}}%
  \def\RewWO{\Rewbase{#1}{}}%
  \@ifnextchar[\RewW\RewWO%
}

\newcommand{\Rewbis}[1]{%
  \def\RewW[##1]{\Rewbasebis[##1]{#1}{}}%
  \def\RewWO{\Rewbasebis{#1}{}}%
  \@ifnextchar[\RewW\RewWO%
  }

\newcommand{\Rewn}[2][*]{%
  \def\RewnW[##1]{\Rewbase[##1]{#2}{#1}}%
  \def\RewnWO{\Rewbase{#2}{#1}}%
  \@ifnextchar[\RewnW\RewnWO%
  }

\makeatother

\newcommand{\head}{{\tt h}}

\newcommand{\nbvctxtwo}[1]{\nbvctxtwo{#1}}

\newcommand{\isub}[2]{\{#1/#2\}}
\newcommand{\esub}[2]{[#1 \backslash #2]}
\renewcommand{\isub}[2]{\{#1\backslash#2\}}

\newcommand{\llbrace}{\{ \kern -0.27em \vert}
\newcommand{\rrbrace}{\vert \kern -0.27em \}}

\renewcommand{\l}{\lambda}
\newcommand{\cf}{\emph{cf.}\xspace} 
\newcommand{\ie}{\emph{i.e.}\xspace}

\newcommand{\ih}{\emph{i.h.}\xspace}

\newcommand{\red}[1]{{\color{red} {#1}}}

\newcommand{\myinput}[1]{\ifthenelse{\boolean{withimages}}{\input{#1}}{}}

\newcommand{\ax}{\mathsf{ax}}

\newcommand{\size}[1]{|#1|}

\newcommand{\pair}[2]{\langle #1,#2 \rangle}

\newcommand{\ov}[1]{\overline{#1}}

\newcommand{\pder}{\Vdash}
\newcommand{\prodt}[2]{\times(#1,#2)}

\newcommand{\produ}[1]{\times_1(#1)}
\newcommand{\prodd}[1]{\times_2(#1)}

\newcommand{\trempty}{\mathtt{emptypair}}
\newcommand{\tnormalpair}{\mathtt{pair_p}}

\newcommand{\trpair}{\mathtt{pair}}

\newcommand{\trsub}{\mathtt{match}}

\newcommand{\trvarpat}{\mathtt{pat_v}}
\newcommand{\trpairpat}{\mathtt{pat_{\times}}}

\newcommand{\oprod}{o}

\newcommand{\Upper}{$\mathscr{U}$}
\newcommand{\Exact}{$\mathscr{E}$}

\newcommand{\M}{\mathcal{M}}
\newcommand{\N}{\mathcal{N}}
\newcommand{\appctx}[2]{#1 {[ \! [}  #2 {] \! ]} }
\newcommand{\pattern}{{\tt p}}
\newcommand{\ctxt}[1]{{\tt #1}}

\usepackage{appendix}
\capturecounter{theorem}
\capturecounter{lemma}
\capturecounter{proposition}

\title{A Quantitative Understanding of Pattern Matching} 

\author{Sandra Alves}{DCC-FCUP \& CRACS, Univ. Porto, Portugal}{}{}{}

\author{Delia Kesner}{IRIF, Univ. Paris, CNRS, France \and Institut Universitaire de France (IUF), France}{}{}{}

\author{Daniel Ventura}{INF, Univ. Federal de Goi\'as,
  Brazil}{}{}{}

\authorrunning{S. Alves and D. Kesner and D. Ventura}

\Copyright{Sandra Alves and Delia Kesner and Daniel Ventura}

\ccsdesc[100]{Theory of computation~Type theory}
\ccsdesc[100]{Theory of computation~Models of computation}

\keywords{Intersection Types, Pattern Matching, Exact Bounds}

\category{}

\supplement{}

\funding{Supported by the Brazilian Research Council
  (CNPq) grant Universal 430667/2016-7.}

\nolinenumbers 

\hideLIPIcs  

\EventEditors{John Q. Open and Joan R. Access}
\EventNoEds{2}
\EventLongTitle{42nd Conference on Very Important Topics (CVIT 2016)}
\EventShortTitle{CVIT 2016}
\EventAcronym{CVIT}
\EventYear{2016}
\EventDate{December 24--27, 2016}
\EventLocation{Little Whinging, United Kingdom}
\EventLogo{}
\SeriesVolume{42}
\ArticleNo{23}

\begin{document}

\maketitle
\begin{abstract} This paper shows that the recent approach to
    quantitative typing systems for programming languages can be
    extended to pattern matching features. Indeed, we define
    two     resource-aware type systems, named \Upper\ and \Exact,
    for a $\l$-calculus equipped with pairs
    for both patterns and terms.  Our typing systems
    borrow some basic  ideas from~\cite{BKRDR15}, which characterises (head)
    normalisation in a {\it qualitative} way, in the sense that
    typability and normalisation coincide. But, in contrast
    to~\cite{BKRDR15}, our ({\it static}) systems also provide {\it
      quantitative} information about the {\it dynamics} of the
    calculus. Indeed,
    system \Upper\  provides {\it upper bounds}
    for the
    {\it length} of (head) normalisation sequences {\it plus}  the {\it size} of
    their corresponding normal forms, while
    system \Exact, which can be seen as a refinement of
    system \Upper, 
    produces {\it exact bounds} for {\it each of} them.
    This is achieved by means of a
    non-idempotent intersection type system equipped with different
    technical tools. First of all,  we use product types to type pairs,
    instead of the disjoint unions in~\cite{BKRDR15},
    thus avoiding an overlap between
      ``being a pair'' and ``being duplicable'',
     resulting in  an essential tool to reason
      about quantitativity.
    Secondly, typing sequents in system \Exact\ are decorated with
    tuples of integers, which provide quantitative information about
    normalisation sequences, notably {\it time} (\cf length) and {\it
      space} (\cf size).  Another key tool
    of  system \Exact\ is that
    the type system distinguishes
    between {\it consuming} (contributing to time) and {\it
      persistent} (contributing to space) constructors.  Moreover, the
    time resource information is remarkably refined, because
    it     discriminates between different kinds of reduction steps performed
    during evaluation, so that beta, substitution and
    matching steps are counted separately.
\end{abstract}

\section{Introduction}
\label{s:introduction}
This paper gives a quantitative understanding of pattern matching, by
introducing two typing systems, named \Upper\ and \Exact, that
respectively provide upper an exact bounds for the length of
(head) normalisation sequences, as well as the size of the reached normal forms.

\noindent \deft{Pattern Calculi:}
Modern (functional) programming languages, like OCaml and Haskell, and proof
assistants, such as Coq and Isabelle, are equipped with rich pattern
matching mechanisms, that allow data to be processed and decomposed
according to some criteria (typically, its
syntactical structure). However, the theory and semantics of
programming usually focuses on $\lambda$-calculi, where abstraction
carries on variables instead of patterns, causing a
conceptual gap, since some theoretical properties of the
$\lambda$-calculus do not translate directly to pattern calculi. For
example, solvability for pattern calculi (called observability) has
been shown to differ from typability, which suffices to characterise
solvability in the $\lambda$-calculus (see~\cite{BKRDR15} for
details). Other interesting examples are standardisation for
pattern calculi~\cite{KLR11-hor}, and neededness~\cite{BonelliKLR12}.
It is then crucial to study the semantics of programming languages with
pattern matching features by means of theoretical calculi equipped
with built-in patterns --for instance those in~\cite{rhoCalIGLP-I-2001,Kahl-2004a,JK09,KLR11-hor,AlvesDFK18}-- instead of considering pattern matching as an encoded mechanism.

A natural approach to model these kind of languages is to generalise
$\lambda$-abstractions $\lambda x. t$ to functions of the form
$\lambda p. t$, where $p$ is a {\it pattern} specifying the expected
structure of their arguments.
In this work we focus on
a $\lambda$-calculus using pair constructors for both terms and
patterns. This can be seen as a simplified form of algebraic pattern
matching, but still powerful enough to
reason about the most interesting features of syntactical matching
mechanisms.

\noindent \deft{Quantitative Types:} Quantitative types are related to
the consumption of time and space by programs, and the resource-aware
semantics of {\it non-idempotent intersection} has been employed in
such considerations.  {\it Non-idempotent intersection} type systems
(also known as {\it quantitative} type systems) have been independently
introduced in the framework of the $\lambda$-calculus to identify
needed redexes~\cite{Gardner94} and to study a linearisation of the
$\l$-calculus~\cite{Kfoury2000} (first appeared in
\cite{kfoury96}). Although widely unnoticed, the {\it quantitative} aspect of
such typing systems is crucial in both investigations, where how (and
how many times) a resource is consumed through a normalising reduction
sequence plays a central role. However, only after~\cite{BoudolCL99}
this quantitative feature of non-idempotent intersection was
highlighted, and since De Carvalho's thesis in 2007
(see also~\cite{Carvalho07}) its relation with linear
logic~\cite{Girard87} and quantitative relational models has been
deeply explored.

As its idempotent counterpart, non-idempotent
intersection type systems may characterise different
  notions of normalisation (such as head, weak and strong)~\cite{Carvalho18,BernadetL13,BucciarelliKV17} but, instead of using some
semantic argument such as reducibility, simple combinatorial arguments
are enough to guarantee termination of typable terms. Quantitative
types have since been applied on the $\lambda$-calculus for the
characterisation of termination with respect to a variety of
other evaluation strategies, for instance 
call-by-value~\cite{Ehrhard12,AccattoliG18}, call-by-need~\cite{Kesner16,BBBK17,AccattoliGL19} and (linear) head
reduction~\cite{Gardner94,AccattoliGK18}. The quantitative reasoning tool has been
well-adapted also to some explicit resource
calculi~\cite{BernadetL13,KV14,KV15}, as well as to pattern 
calculi~\cite{BernadetTh,BKRDR15,BBM18}, proof-nets~\cite{CarvalhoF16},
classical logic~\cite{KV17} and call-by-push-value~\cite{EhrhardG16,abs-1904-06845}.

\noindent \deft{Upper-bounds versus Exact-bounds:} The size of typing
derivations has often been used as an upper bound
to the lenght of different evaluation
strategies~\cite{BBBK17,BucciarelliKV17,Ehrhard12,KV14}, ever since
 its definition to deal with the
well-known notions of head and leftmost
evaluation in Krivine's abstract machine
(KAM)~\cite{Carvalho07,Carvalho18,krivine93book}.
A key notion behind all these works is that
when $t$ evaluates to $t'$, then 
the size of the type derivation of $t'$ is smaller than
the one of $t$, thus size of type derivations 
provides an {\it upper bound} for the
{\it length} of the reduction to a normal form  as well as for the {\it size}
of this normal form.

A crucial point to obtain {\it exact bounds}, instead of upper bounds, is to
consider only minimal typing derivations, such as principal derivations,
which give the notion of {\it all and only information} for typings
(\cf\ \cite{Wells02} for an abstract definition).   A syntactic notion of minimal typings, bearing
similarities with principal typings, was applied to the
KAM~\cite{Carvalho07}, then to the maximal evaluation
strategy~\cite{BernadetL13}. The technique was further developed
in~\cite{AccattoliGK18} with the introduction of {\it tightness} as an
appropriate notion for minimal typings, systematically broadening the
definition of exact bounds for other evaluation strategies. Exact bounds were
also obtained for both call-by-value~\cite{AccattoliG18} and
call-by-need~\cite{AccattoliGL19} strategies, as well as for classical
logic~\cite{KV19}.

\noindent \deft{Related work and Contributions:}
Non-idempotent intersection types have been used to characterise 
  strong normalisation, in  a calculus with fix-point operators and pattern matching on constructors~\cite{BernadetTh}.
Similarly, a strong call-by-need strategy for a pattern matching
language was defined in~\cite{BBM18}, and completeness of the
strategy was shown by means of non-idempotent intersection types
by extending the technique introduced in~\cite{Kesner16,BBBK17}.
In both cases, despite the use of non-idempotent types, no quantitative
results were obtained by means of the typing system.
 
Closer to our work, \cite{BKRDR15} studied the solvability property in
a pattern calculus with pairs by means of intersection types. One of
the contributions of~\cite{BKRDR15} is a characterisation of
(non-deterministic) head-normalisation by means of typability. Our
presentation of the pattern calculus goes a step further, by providing
a suitable {\it deterministic} head-reduction strategy (called
head-evaluation) which is complete w.r.t. the notion of
head-normalisation defined in~\cite{BKRDR15}, in the sense that
head-normalisation and termination of the head-strategy are
equivalent, thus answering one of the open questions
in~\cite{BKRDR15}.

Moreover, the mentioned characterisation result in~\cite{BKRDR15} is
merely {\it qualitative}, as it does not give any {\it upper/exact
  bound} for head-evaluation, while our work
fully exploits quantitativity, first by providing a typing system
being able to compute upper bounds for
head-evaluation, then by refining the
above  type system as well as the bounds to be
\emph{exact}.

Indeed, the first contribution of this paper is to go beyong the
qualitative characterisation of head-normalisation for pattern calculi
by providing a typing system \Upper\ being able to statically compute
{\it upper bounds} for 
head-evaluation.  The main reason why the type system
  in~\cite{BKRDR15} fails to provide \emph{upper} bounds to
  head-normalisation sequences is because of the use of commuting
  conversions in the reduction relation
  associated to the pattern calculus. We solve this problem by
  integrating these commuting conversions into the basic reduction
  rules, so that the resulting system
  implements reduction {\it at a distance}, like in~\cite{AK12}.
The new specification of the reduction system at a distance is
  now well-behaved w.r.t. our first type system \Upper: if $t$ is well
  typed in \Upper, then the size of its type derivation gives an upper
  bound to the (deterministic) head-reduction sequence from $t$ to its
  (head) normal form. Our system \Upper\ can be seen as a form of quantitative (relational) {\it model} for the pair pattern calculus (Sec.~\ref{s:model}),
 following the lines of~\cite{BE01,BucciarelliEM12,PPR17}.

Concerning {\it exact} bounds, the main reason why the type system in~\cite{BKRDR15} fails to provide \emph{exact bounds} to head-normalisation
sequences is because of a double usage of non-idempotent
intersection types: they are not strictly reserved to specify
duplicability power, as is usual in non-idempotent
types~\cite{BucciarelliKV17}, but are also used to described different
paths in the syntactical structure of patterns and pairs. More
precisely,~\cite{BKRDR15} models product types by disjoint unions, so
that any pair $\pair{t}{u}$ of typed terms $t$ and $u$ has necessarily
{\it at least} two types, one of the form $\produ{\sig}$ where $\sig$
is the type of $t$, and one of the form $\prodd{\tau}$, where $\tau$
is the type of $u$.  This has an important undesirable consequence,
because multiple types carry two completely different meanings: being
a pair (but not necessarily a pair to be duplicated), or being a
duplicable term (but not necessarily a pair).
To address this issue we
use standard product types for pairs, and multi-sets that are {\it
  singletons of products} for pair-patterns.  In other words, we make
use of the exponential isomorphism $! ( A \bindnasrepma B ) \equiv ! A
\otimes ! B $ of multiplicative exponential linear
logic~\cite{Girard87}. Our (simpler) formulation
of product types turns out to be more appropriate to reason about
quantitativity.  

We can thus say that the second contribution of this paper is to
  go beyond upper bounds by providing
  a typing system \Exact\ being able to statically provide
  {\it exact bounds} for
  head-evaluation. This is
  done by using several key tools. 
  First of all, we adopt the specification of
    product types  explained in the
    previous paragraph. Moreover, we do  not use any kind of
  (idempotent) {\it structural} type, as done in~\cite{BKRDR15}.
Besides that, an important  notion
used in system
\Exact\ is the clear distinction between {\it consuming} and {\it
  persistent} constructors.  This has some
intuition coming from the theory of residuals~\cite{barendregt84nh},
where any symbol occurring in a normal form can be traced back to the
original term. A constructor is consuming (resp. persistent) if it is
consumed (resp. not consumed) during head-reduction. For instance, in
$(\l z. z)(\l z. z)$ the first abstraction is consuming while the
second one is persistent.  This dichotomy between consuming/persistent
constructors has been used in~\cite{AccattoliGK18} (resp.~\cite{KV19})
for the $\lambda$-calculus (resp. $\lambda\mu$-calculus), and adapted
here for the pattern calculus. Indeed, patterns and terms are consumed
when the pair constructor is destroyed during the execution of the
pattern matching rule. Otherwise, patterns and pairs are persistent,
and they do appear in the normal form of the original term.
Last, but not least, another major ingredient of our approach is the use of {\it
    tight} types, inspired from~\cite{AccattoliGK18}, and the corresponding notion of tight (\cf\ minimal)
  derivations. This is combined with the introduction of {\it
    counters} in the typing judgements,  which are used
  to discriminate between different sorts of
  reduction steps performed during evaluation, so that beta,
  substitution and matching steps are statically (exactly) counted,
  each one with its own counter.
The following list summarises our contributions:
\begin{itemize}
\item A deterministic head-strategy
  for the patttern calculus which is complete w.r.t. the notion of head-normalisation.
\item A  type system \Upper, which  provides
  upper bounds for the lenght of head-normalisation
  sequences plus the size of
  its corresponding normal forms.
\item Refinement of system \Upper\ to system \Exact, being able
  to provide exact bounds for both the lenght of head-normalisation
  sequences and the size of
  its corresponding normal forms.
\end{itemize}

\noindent \deft{Road-map:} Sec.~\ref{s:pattern-calculus} introduces
the pattern calculus. Sec.~\ref{s:typing-system} presents
the typing system \Upper, together with some
of its quantitive properties, and Sec.~\ref{s:model} suggests a relational model
for our pattern calculus
based on the type system. In Sec.\ref{s:tight-typing-system}, we refine \Upper\ to
extract exact bounds, which leads to the definition of our second
typing system \Exact. The correctness (resp. completeness) proof for
\Exact\ is given in Sec.~\ref{s:soundness}
(resp. Sec.~\ref{s:completeness}).  Conclusions and future work are
discussed in Sec.~\ref{s:conclusion}.


\section{The Pattern Calculus}
\label{s:pattern-calculus}

In this section we introduce the pattern calculus,
which is an   extension of the $\lambda$-calculus, where
abstraction is extended to {\it pair patterns}, and terms are extended to
{\it pairs}. We start by introducing the syntax of the calculus.

\deft{Terms and contexts} of the pattern calculus are defined by means of the following grammars:
\[ \hfill \begin{array}{llll}
   \mbox{(\deft{Patterns})} &  p, q   & ::= & x \mid  \pair{p}{q}  \\
   \mbox{(\deft{Terms})}    &  t,u,v & ::= & x \mid \l p.t \mid \pair{t}{u} \mid tu
                                               \mid t\esub{p}{u} \\
   \mbox{(\deft{List Contexts})}    &  \slist  & ::= & \Box   \mid \slist\esub{p}{u} \\
   \mbox{(\deft{Contexts})}    &  \ctxt{C}  & ::= & \Box   \mid \l p. \ctxt{C}  \mid \pair{\ctxt{C}}{t} \mid \pair{t}{\ctxt{C}} \mid  \ctxt{C}t\mid t \ctxt{C}
                                               \mid  \ctxt{C} \esub{p}{t}  \mid t\esub{p}{ \ctxt{C} }
   \end{array} \hfill \]
where $x,y,z,w \ldots $ range over a countable set of variables, and every
pattern $p$ is assumed to be {\it linear}, \ie\ every variable
appears at most once in $p$. The term $x$ is called a \deft{variable},
$\l p. t$ is an \deft{abstraction}, $\pair{t}{u}$ is a \deft{pair},
$tu$ is an \deft{application} and $t\esub{p}{u}$ is a \deft{closure},
where $\esub{p}{u}$ is an \deft{explicit matching} operator.
Special terms are $\id:= \l z. z$,
$\Delta := \l z. zz$ and
$\Omega := \Delta \Delta$. As usual we use the abbreviation $\l p_1 \ldots p_n. t_1 \ldots t_m  $  for 
$\l p_1 (\ldots (\l p_n. ((t_1 t_2) \ldots t_m))\ldots)$, $n\geq 0$,  $m\geq 1$. 

We write $\var{p}$ to denote the variables in
the 
  pattern $p$. \deft{Free} and \deft{bound variables} of terms
  and contexts are defined as
expected, in particular $\fv{\l p.t} := \fv{t} \setminus \var{p}$
and $\fv{t[p/u]} := (\fv{t} \setminus \var{p}) \cup \fv{u}$. We
write $p \# q$ if  $\var{p}$ and $\var{q}$ are disjoint.
As usual, terms are considered modulo $\alpha$-conversion, so that
for example $\l \pair{x}{y}. xz =_{\alpha} \l \pair{x'}{y'}. x'z$
and $x\esub{\pair{x}{y}}{z} =_{\alpha} x'\esub{\pair{x'}{y'}}{z}$.  Given a
list context $\slist$ and a term $t$, $\appctx{\slist}{t}$ denotes the term
obtained by replacing the unique occurrence of $\square$ in
$\slist$ by $t$, possibly allowing the capture of free variables of
$t$. We use the predicate $\ttabs{t}$ when $t$ is  of the form $\appctx{\slist}{\l p. u}$.
The \deft{reduction relation $\Rew{\pattern}$} on terms
is given by the closure of {\it all} contexts of the
following rewriting rules. 
  \[\hfill \begin{array}{lllll}
     \appctx{\slist}{\l p. t} u & \Rew{} &
                                           \appctx{\slist}{t\esub{p}{u}}
     & \bv{\slist}\cap \fv{u} = \es\\
     t\esub{\pair{p_1}{p_2}}{\appctx{\slist}{\pair{u_1}{u_2}}} & \Rew{} &
     \appctx{\slist}{t\esub{p_1}{u_1}\esub{p_2}{u_2}} &
                                                         \bv{\slist}
                                                         \cap \fv{t} = \es
                                                         \\
     t\esub{x}{u}  & \Rew{} & t\isub{x}{u}  & \\
           \end{array} \hfill  \]

  The reduction relation $\Rew{\pattern}$ defined above is
    related to that in~\cite{BKRDR15}, called
    $\Rew{\Lambda_{\pattern}}$, in the following sense:
    $\Rew{\Lambda_{\pattern}}$ contains two subsystem relations, one
    to deal with {\it clashes}, which are not handled in the
    present calculus since we consider typable terms only (\cf Lem. \ref{l:relevance+clash-free-upper}.\ref{clash-free-upper}), and another one containing the following five
    rules:
         \[\hfill \begin{array}{lllll}
     (\l p. t) u & \Rew{} &
           t\esub{p}{u}\\ t\esub{\pair{p_1}{p_2}}{\pair{u_1}{u_2}} &
           \Rew{} & t\esub{p_1}{u_1}\esub{p_2}{u_2} \\ t\esub{x}{u} &
           \Rew{} & t\isub{x}{u} & \\ t\esub{p}{v}u & \Rew{} &
           (tu)\esub{p}{v} & \fv{u} \cap \var{p} = \es
           \\ t\esub{\pair{p_1}{p_2}}{u\esub{q}{v}} & \Rew{} &
           t\esub{\pair{p_1}{p_2}}{u} \esub{q}{v} & \fv{t} \cap
           \var{q} = \es \\ \end{array} \hfill \] The two last rules
         can be seen as commuting conversions, and so they can be
         integrated in the first (two) rules by using the {\it
        substitution at a distance paradigm}~\cite{AK12}, which is usually specified by
         means of a context $\slist$, as in  the first two rules of our
         reduction relation $\Rew{\pattern}$. It is worth noticing that
         $t \Rew{\pattern} t'$ can be simulated by $t
         \Rewplus{\Lambda_{\pattern}} t'$.
         For instance,
           $(\l p. t)\esub{p_1}{u_1}\esub{p_2}{u_2} u \Rew{\pattern}
             t\esub{p}{u}\esub{p_1}{u_1}\esub{p_2}{u_2}$ can be simulated by:
               \begin{align*}
                  (\l p. t)\esub{p_1}{u_1}\esub{p_2}{u_2} u \Rew{\Lambda_{\pattern}}
             ((\l p. t)\esub{p_1}{u_1} u )\esub{p_2}{u_2} &
                                                            \Rew{\Lambda_{\pattern}}  ((\l p. t) u )\esub{p_1}{u_1}\esub{p_2}{u_2} \\
                  &\Rew{\Lambda_{\pattern}} t\esub{p}{u}\esub{p_1}{u_1}\esub{p_2}{u_2}
                 \end{align*} 
         This is equivalent to
         what is done in call-by-name $\lambda$-calculus at a distance.  Our
         formulation of the pattern calculus at a
         distance, given by the relation $\Rew{\pattern}$,
         as well as the corresponding head strategy that we present below, are essential to get quantitative results about
         head-normalisation (\cf\ Sec.~\ref{s:typing-system} and
         Sec.~\ref{s:tight-typing-system}).

         Although
         the reduction relation $\Rew{\pattern}$ is {\it
           non-deterministic}, it can easily been shown to be {\it
           confluent}, for example using the same technique as in~\cite{BKRDR15}. However, in order to study exact bounds of
         evaluation, we need to define a {\it deterministic} strategy
         for the pattern calculus, \ie\ a subrelation of
         $\Rew{\pattern}$ that is able to compute the same normal
         forms.  Fig.~\ref{f:head-strategy} gives an operational
         semantics for the pattern calculus, which turns out to be an
         extension of the well-known notion of {\it head-reduction}
         for $\l$-calculus, then also named \deft{head-reduction}, and
         denoted by the relation $\Rewbis{\head}$. In the following
         inductive definition $t \Rew{\head} u$ means that $t$
         head-reduces to $u$, and $t \notRew{\head}$ means that $t$ is
         a \deft{head normal-form}, \ie\ there is no $u$ such that
         $t \Rew{\head} u$. 
\begin{figure}[!h]
 \[\hspace{-.6cm}  \begin{array}{c}
      \infer[\betas]{\bv{\slist} \cap \fv{u} = \es}
          {\appctx{\slist}{\l p. t} u \Rew{\head} \appctx{\slist}{t\esub{p}{u}} } \sep
      \infer[\matchs]{t \notRew{\head} \sep \sep \bv{\slist} \cap \fv{t} = \es}{t\esub{\pair{p_1}{p_2}}{\appctx{\slist}{\pair{u_1}{u_2}}}  \Rew{\head} \appctx{\slist}{t\esub{p_1}{u_1}\esub{p_2}{u_2}}} \sep
      \infer[\exps]{t \notRew{\head}\ }
           {t\esub{x}{u}   \Rew{\head}  t\isub{x}{u}}  \\ \\ 
       \infer{t \Rew{\head} t'}{\l p. t \Rew{\head} \l p. t'} \sep
    \infer{t \Rew{\head} t'\ \sep\ \isnotabs{t}}{tu \Rew{\head} t'u}  \sep
    \infer{t \Rew{\head} t'}{t\esub{p}{u} \Rew{\head} t'\esub{p}{u}}  \sep 
 \infer{t \notRew{\head}\ \sep\ p\neq x\ \sep\ u \Rew{\head} u' }
           {t\esub{p}{u} \Rew{\head} t\esub{p}{u'}} \sep
     \end{array} \]
   \caption{The head-reduction strategy for the pattern calculus}
   \label{f:head-strategy}
   \end{figure}

   Rule $\betas$ fires the computation of terms by transforming an
   application of a function to an argument into a closure term.
   Decomposition of patterns and terms is performed by means
   of rule $\matchs$, when a pair pattern is matched against a pair
   term.  Substitution is performed by rule $\exps$, \ie\ an explicit
   (simple) matching of the form $\esub{x}{u}$ is executed.  This form of
   syntactic pattern matching is very simple, and does not consider
   any kind of failure result, but is expressive enough to specify the
   well-known mechanism of matching. Context closure is {\it similar} to the
   $\l$-calculus case, but not exactly the same. Indeed, reduction is
   performed on the left-hand side of applications and closures
   whenever possible.  Otherwise, arguments of explicit matching
   operators must be reduced in order to unblock these operators,
   \ie\ in order to decompose $\esub{p}{u}$ when $p$ is a pair pattern
   but $u$ is still not a pair. Notice however that when $u$ is
   already a pair, reduction inside $u$ cannot take place at all, thus
   implementing a kind of {\it lazy} strategy for pattern matching.

   The head-reduction strategy so far defined can be understood as a
   combination of {\it call-by-name} and {\it call-by-value}: the
   call-by-name flavour comes from the fact that arguments for
   variables are never evaluated before function application, while in
   order to match a term against a pattern, this one needs to be
   (partially) evaluated (to a pair) before performing a matching
   decomposition, a phenomenon which also gives a {\it call-by-value}
   parfum to the strategy.  Standardisation of calculi as the one in
   this paper has been studied in~\cite{KLR11-hor}.

   Given a (one-step) reduction relation $\Rew{\R}$,
  we use $\Rewn{\R}$, or more precisely 
  $\Rewparam{k}{\R}\ (k \geq0)$ to denote the reflexive-transitive closure
   of $\Rewbis{\R}$, \ie\ $k$ $\R$-steps. In the case of $\head$ reduction, we may  use the
 alternative notation
 $\Rewparam{(b,e,m)}{\head}$ to emphasize the number  of reduction steps
 in a given reduction sequence, \ie\ if $\rho: t \Rewparam{(b,e,m)}{\head} u$,
 then there are exactly $b$ $\betas$-steps, $e$ $\exps$-steps and
 $m$ $\matchs$-steps in the reduction sequence $\rho$. We will often use the notation $\Rew{\betas}$ to explicitly refer to a $\betas$-step (resp. $\Rew{\exps}$ and $\Rew{\matchs}$ for $\exps$ and $\matchs$ steps). 
The reduction relation $\Rew{\head}$ is in fact a function:

 \begin{proposition}
   The relation $\Rew{\head}$ is deterministic.
 \end{proposition}

 \begin{example}
     \label{e:first}
     Let us consider the combinators $\id = \l z. z$ and $\K := \l x_1. \l y_1. x_1$. Then we have $ (\l \pair{x}{y}. x(\id y))\esub{z}{\id} (\id \pair{\K}{w})
   \Rewparam{(4,6,1)}{\head} \l y_1. w$:
   \[\begin{array}{llll}
       & (\l \pair{x}{y}. x(\id y))\esub{z}{\id} (\id \pair{\K}{w}) & \Rew{\betas} &
         (x (\id y)) \esub{\pair{x}{y}}{ \id \pair{\K}{w}}\esub{z}{\id}  \\
 \Rew{\betas} & 
                (x (\id y)) \esub{\pair{x}{y}}{z\esub{z}{\pair{\K}{w}}}\esub{z}{\id}  & \Rew{\exps} & (x (\id y)) \esub{\pair{x}{y}}{\pair{\K}{w}}\esub{z}{\id}  \\
        \Rew{\matchs} & 
      (x (\id  y)) \esub{x}{\K} \esub{y}{w}\esub{z}{\id}   & \Rew{\exps} & 
                                                                           (\K (\id y))  \esub{y}{w} \esub{z}{\id}  \\
       \Rew{\betas} & 
      (\l y_1. x_1)\esub{x_1}{\id y}  \esub{y}{w} \esub{z}{\id}  & \Rew{\exps} & 
       (\l y_1. \id y)  \esub{y}{w} \esub{z}{\id}\\
       \Rew{\betas} & (\l y_1. z\esub{z}{y})\esub{y}{w}\esub{z}{\id}&
       \Rew{\exps} & (\l y_1. y)\esub{y}{w}\esub{z}{\id} \\
       \Rew{\exps} & (\l y_1. w) \esub{z}{\id} & \Rew{\exps} & \l y_1. w \\
    \end{array} \]
\end{example}

 Head normal-forms may contain {\it
    (head)} clashes, \ie\ (head) terms not syntactically well-formed. For example, a
(list context containing a) pair applied to
 another term $\appctx{\slist}{\pair{u_1}{u_2}}v$, or a matching
 between a pair pattern and a
(list context containing a) function $t\esub{\pair{p_1}{p_2}}{\appctx{\slist}{\l p. u}}$, or its expanded form $\appctx{\slist_1}{\l \pair{p_1}{p_2}. t}\appctx{\slist_2}{\l p. u}$ are considered to be  (head) clashes.
 A rewriting system raising a warning (\ie\ a failure)
 when detecting a (head) clash  has been  defined  in~\cite{BKRDR15},
 allowing to restrict the attention to a smaller set of
 terms,  called
 {\it canonical} terms, that are
 intended to be the (head) clash-free terms that are not reducible by
 the  relation $\Rew{\head}$.
 Canonical terms  can be characterised inductively:
 \[ \begin{array}{llll}
      (\mbox{\deft{canonical forms}}) & \M & ::= & \l p. \M  \mid \pair{t}{t} \mid \M\esub{\pair{p_1}{p_2}}{\N}\mid \N \\ 
      (\mbox{\deft{pure canonical forms}}) & \N & ::= & x \mid \N t \mid \N\esub{\pair{p_1}{p_2}}{\N}\\ 
    \end{array} \]
  
 In summary, canonical terms and irreducible terms are related as follows:
  
  \begin{proposition}
    \label{p:M-not-reducible}
     $t \in \M$ if and only if $t$ is (head) clash-free and $t \notRew{\head}$.
  \end{proposition}

  \deft{Size} of canonical terms is defined as follows:
  $\size{x} ::= 0$, $\size{\pair{t}{u} } ::= 1$,
  $\size{\N t} ::= \size{\N} + 1$,
  $\size{\l p. \M } ::= \size{\M} +1$, and
  $\size{\M\esub{\pair{p_1}{p_2}}{\N}} ::= \size{\M} + \size{\N} + 1$.
  As an example, the terms $\l \pair{x}{y}. \pair{x}{\id}$
  and $\l x. y (\pair{x}{z} \id)$ are
  canonical forms of size
    $2$ while $x \Omega$ and $z \esub{\pair{z}{w}}{x
    \Omega}$ are pure canonical terms of size $1$ and $2$
    respectively. The term $\pair{x}{\id}
  w$ is none of them, and the term $\id
  x$ can reduce to the canonical term $x$.

          Finally, we define a term $t$ to be \deft{head-normalisable}
          if there exists a canonical form $u \in \M$ such that
          $t \Rewn{\pattern} u$. Moreover, $t$ is said to be
          \deft{head-terminating} if there exists a canonical form
          $u \in \M$ and an integer $k\geq 0$ such that
          $t \Rewparam{k}{\head} u$. The relation between the
            non-deterministic reduction relation $\Rew{\pattern}$ and
            the deterministic strategy $\Rew{\head}$ will be
            established latter, but we can already say that, while $t$
          head-terminating immediately implies $t$ head-normalisable,
          the completeness of the head-strategy
          w.r.t. to head-normalisation is not trivial
          (Thm.~\ref{t:characterisation-upper}).

\section{The \Upper\ Typing System}
\label{s:typing-system}
In this section we introduce our first typing system \Upper for the pattern calculus. 
We  start by defining the sets of \deft{types} and \deft{multi-set
  types}, given by means of the following grammars:
\[ \hfill \begin{array}{llll}
     \mbox{(Product Types)} & \PT & :: = & \prodt{\A_1}{\A_2} \\
     \mbox{(Types)} &  \sig & ::= & \typend  \mid \PT \mid \A \rew \sig \\
     \mbox{(Multi-set Types)} & \A & :: = & \mult{\sig_k}_{\kK} \\
   \end{array} \hfill     \]
 where $K$ is a (possibly empty) finite set of indexes, and a multi-set type is an unordered list of (not necessarily different) elements.
 We write $\size{\A}$ to denote the number of elements of the
 multi-set $\A$. For example
 $\mult{\typend, \ems \rew \typend, \typend}$ is a multi-set type, where  $\ems$ denotes the {\it empty} multi-set.  We write $\sqcup$ to denote multi-set union.
 Multi-set types are used to specify how programs consume terms:
 intuitively, the
 empty multi-set is assigned to  terms that are  erased during (head)
 reduction, while duplicable terms are necessarily typed with non-empty multi-sets. As usual the arrow type is right-associative.

 A product type, representing the type of a pair, is defined as the product of two (possibly empty) multi-sets of types. This
 formulation  of product types turns out to
 be a key tool to reason about quantitativity,
 and constitutes an essential difference with the
product types proposed in~\cite{BKRDR15}. Indeed,
in  op. cit.,   multi-set types carry two completely different meanings: being
a pair (but not necessarily a pair to be duplicated), or being a
duplicable term (but not necessarily a pair).  Our
specification of products can then be
interpreted as the use of 
the exponential isomorphism $! ( A \bindnasrepma B ) \equiv ! A
\otimes ! B $ of multiplicative exponential linear
logic~\cite{Girard87}.
 
A \deft{typing context} $\Gam$ is a map from variables to
multi-set types, such that only finitely many variables are not mapped to the empty multi-set $\emptymultiset$. We write $\dom{\Gam}$ to denote the domain of $\Gam$, which is the set $\{x \mid \Gam(x) \neq \emptymultiset\}$. We may write $\Gam \# \Del$ if and only if $\dom{\Gam}$ and $\dom{\Del}$ are disjoint. Given typing contexts $\{\Gam_i\}_\iI$ we write $\inter_\iI \Gam_i$ for the context that maps $x$ to $\msunion_\iI \Gam_i(x)$. One particular case is $\Gam \inter \Del$. We sometimes write $\Gam;\Del$ instead of $\Gam \inter \Del$, when $\Gam \#\Del$, and we do not distinguish $\Gam;x:\ems$ from $\Gam$. The typing context $\Gam|_{p}$ is such that $\Gam|_{p} (x) = \Gam(x)$, if $x\in \var{p}$ and $\ems$ otherwise. The typing context $\cmin{\Gam}{\vars}$ is defined by  $(\cmin{\Gam}{\vars})(x) = \Gam(x)$ if $x\notin \vars$ and $\ems$ otherwise.  

The \deft{type assignment system} \Upper is given in
Fig.~\ref{f:typing-rules-unc} and can be
  seen as a natural extension of Gardner's system~\cite{Gardner94} to
  explicit operators, pairs and product types. It assigns types (resp. multi-set
types) to terms, using an auxiliary (sub)system that assigns multi-set
types to patterns. We use $\Phi \tri \seq{\Gam}{t:\sig}$ (resp. $\Phi
\tri \seq{\Gam}{t:\A}$) to denote \deft{term type derivations} ending
with the sequent $\seq{\Gam}{t:\sig}$ (resp. $\seq{\Gam}{t:\A}$), and
$\Pi \tri \Gam \pder p : \A$ to denote \deft{pattern type derivations}
ending with the sequent $\Gam \pder p : \A$.      The size of a
derivation $\tderiv$, denoted by $\sz{\tderiv}$, is the number of all the typing rules used in
$\tderiv$ except $\many$\footnote{An equivalent type system
    can be presented without the $\many$ rule, as the system of
    \cite{BKRDR15}. However, the current presentation allows
    one to have more elegant inductive proofs.} (this is particularly appropriate
to show the substitution lemma).

\begin{figure}[h]
\begin{framed}
\begin{center}
\[ \hfill \begin{array}{c}
     \infer[(\trvarpat)]{\phantom{AAA}}{x:\A  \pder x:\A}      \sep  \sep
     \infer[(\trpairpat)]{\Gam \pder p: \A \sep \Delta \pder q:\B \sep p \# q }
       { \Gam \inter  \Delta \pder \pair{p}{q} :
       \mult{\prodt{\A}{\B}} }\  \\\\
    \hline \\
       \infer[(\ax)]
      {  } 
     {\seq{x: \mult{\sig}}{  x: \sig} }    \sep \sep \sep
     \infer[(\many)]{(\seq{\Gamk}{t: \sigk})_{\kK}}
       {\seq{\inter_{\kK}\Gamk}{t: \mult{\sigk}_{\kK}}} \\ \\ 
       \infer[(\introarrow)]{ \seq{\Gam}{t : \sig}\sep
       \Gam|_{p} \pder  p:\A}  
      {\seq{\cmin{\Gam}{\var{p}}}{\l p.t: \A \rew  \sig}} \sep \sep       
\infer[(\app)]{ \seq{\Gam}{t: \A \rew \sig} \sep\sep 
       \seq{\Del}{u : \A }}
       {\seq{\Gam \inter \Del}{ t\,u: \sig}}      \\ \\ 
\infer[(\trpair)]{\seq{\Gam}{t:\A} \sep
       \seq{\Del}{u:\B}}
            {\seq{\Gam \inter\Del}{\pair{t}{u} : \prodt{\A}{\B}}} \sep\sep 
       \infer[(\trsub)]{\seq{\Gam}{t:\sig} \sep
       \Gam|_{p} \pder p:\A \sep
       \seq{\Del}{u:\A}}
       {\seq{(\cmin{\Gam}{\var{p}}) \inter  \Del}{t[p/u]:\sig}}
\end{array} \hfill \]    
\end{center}
\caption{Typing System \Upper}
\label{f:typing-rules-unc}
\end{framed}
\end{figure}

Note that, when assigning
types (multi-set types) to terms, we only allow the introduction of
multi-set types on the right through the $\many$ rule.

Most of the rules for terms are straightforward. Rule $\trsub$ is used
to type the explicit matching operator $t\esub{p}{u}$ and can be seen
as a combination of rules $\app$ and $\introarrow$. 
Rule $\trvarpat$ is used when the pattern is a
variable $x$. Its multi-set type is
the type declared for $x$ in the typing
context. Rule $\trpairpat$  is used when the pattern has a product
type, which means that the pattern will be matched with a pair. The
condition $p \# q$ ensures linearity of patterns. Note that any pair
term can be typed, in particular, with $\prodt{\emul}{\emul}$.

The system enjoys the following key properties, easily  proved by
induction:
\begin{toappendix}
\begin{lemma}
    \label{l:relevance+clash-free-upper}
    Let $\Phi \tri \seq{\Gam}{t:\sig}$. Then, 
    \begin{enumerate}
    \item \begin{toappendix}
        ({\bf Relevance}) \label{relevance-upper}
        $\dom{\Gam} \subseteq \fv{t}$.
        \end{toappendix}
    \item \begin{toappendix}
      ({\bf Clash-Free}) \label{clash-free-upper}
      $t$ is (head) clash-free.
      \end{toappendix}
    \item \begin{toappendix}
        \label{i:qsr} ({\bf Upper Subject Reduction})
      $t \Rew{\pattern} t'$ implies there is $\Phi' \tri \seq{\Gam}{t':\sig}$ s.t. $\sz{\Phi} \geq
      \sz{\Phi'}$, and $t \Rew{\head} t'$ implies
      there is $\Phi' \tri \seq{\Gam}{t':\sig}$ s.t.  $\sz{\Phi} >
      \sz{\Phi'}$.
      \end{toappendix}
    \item \begin{toappendix}
        \label{i:se} ({\bf Upper Subject Expansion})
      $t' \Rew{\pattern} t$ implies  there is  $\Phi' \tri \seq{\Gam}{t':\sig}$
      such that $\sz{\Phi'} \geq
      \sz{\Phi}$ and 
      $t' \Rew{\head} t$ implies  there is  $\Phi' \tri \seq{\Gam}{t':\sig}$
      such that $\sz{\Phi'} >
      \sz{\Phi}$.
      \end{toappendix}
  \end{enumerate}
\end{lemma}
\end{toappendix}

\begin{proof} All proofs are by induction on $\Phi$,
    item~\ref{i:qsr} (resp. item~\ref{i:se}) also uses a substitution
    (resp. anti-substitution) lemma (see~App.~\ref{a:typing-system} for
    details).
  \end{proof}

Although the system in~\cite{BKRDR15} already characterises
  head-normalisation in the pattern calculus, it does not provide
  upper bounds for the number of steps of the head strategy
  until its canonical form. This is mainly due to the fact that
  the reduction system in~\cite{BKRDR15} does not always
  decrease the measure of the typed terms, even when reduction is performed in the so-called {\it typed} occurrences. We can recover
  this situation by the following soundness and completeness result:

  \begin{theorem}[Characterisation of Head-Normalisation and Upper Bounds]
    \label{t:characterisation-upper}
    Let $t$ be a term in the pattern calculus. Then (1) $t$ is
    typable in system \Upper iff (2) $t$ is head-normalisable iff (3)  $t$ is
    head-terminating. Moreover, if $\Phi \tri \seq{\Gam}{t:\sig}$,
    then the head-strategy terminates on $t$ in at most $\sz{\Phi}$ steps.
\end{theorem}

\begin{proof} The statement  (1) $\Rightarrow$ (3) holds by Upper Subject Reduction for $\Rew{\head}$. The statement (3) $\Rightarrow$ (2) is
    straightforward since $\Rew{\head}$ is included in $\Rew{\pattern}$.
    Finally,  the statement  (2) $\Rightarrow$ (1) holds
    by the fact that canonical terms are typable
    (easy), and by using Upper Subject
  Expansion for $\Rew{\pattern}$. \end{proof}

The previous upper bound result is especially possible thanks to the
Upper Subject Reduction Property, stating in particular that reduction
$\Rew{\head}$ {\it strictly} decreases the size of typing
derivations. It is worth noticing that the reduction relation
in~\cite{BKRDR15} does not enjoy this property, particularly in the
case of the rule $t\esub{p}{v}u \Rew{} (tu)\esub{p}{v} $, which is a
permuting conversion rule, (slightly) changing the structure of the type derivation, but not its size.

\section{Towards a Relational Model for the Pattern Calculus}
\label{s:model}
Denotational and operational semantics have tended to abstract
quantitative information such as time and space as computational
resource consumption.  Since the invention of Girard's linear
logic~\cite{Girard87}, where formulas are interpreted as resources,
quantitative interpretation of programs, such as relational
models~\cite{BE01,BucciarelliEM12,Carvalho07}, have been naturally
defined and studied, by following the simple idea that multi-sets are
used to record the number of times a resource is consumed.  Thus,
relational models for the $\l$-calculus use multi-sets to keep track
of how many times a resource is used during a computation.

In this brief section we emphasize a semantical result that is
implicit in the previous section. Since relational models are often
presented by means of typing systems~\cite{PPR17,Ong17},
our system \Upper\ suggests a
quantitative model for our pair pattern calculus  in the following way.
Indeed, consider a term $t$
  such that 
  $\fv{t} \subseteq \{x_1, \ldots, x_n\}$, in which case we say that
  the list $\vec{x}= (x_1, \ldots, x_n)$ is \deft{suitable} for $t$. Then,
given $\vec{x} = (x_1, \ldots, x_n)$ suitable for $t$, define the interpretation
of a term $t$ for $\vec{x}$ as
\[ \hfill \interpret{t}_{\vec{x}} = \{
  ((\A_1, \ldots, \A_n), \sig) \mid\
  \mbox{there exists } \Phi \tri x_1:\A_1, \ldots,
    x_n:\A_n \vdash t:\sig \} \hfill \] A straightforward corollary
    of items \ref{i:qsr} and \ref{i:se} of Lem.~\ref{l:relevance+clash-free-upper} is
that $t =_{\pattern} u$ implies
$\interpret{t}_{\vec{x}} = \interpret{u}_{\vec{x}}$, where
$=_{\pattern}$ is the equational theory generated by the reduction
relation $\Rew{\pattern}$. 

\section{The \Exact\ Typing System}
\label{s:tight-typing-system}
In this section we introduce our second typing system
\Exact\ for the  pattern calculus, which is obtained
by refining system \Upper\ presented in Sec.~\ref{s:typing-system}.
\[ \hfill \begin{array}{llll}
     \mbox{(Product Types)} & \PT & :: = & \prodt{\A_1}{\A_2} \\
     \mbox{(Tight Types)} &  \tight & ::= & \typendn \mid \typendm  \\
     \mbox{(Types)} &  \sig & ::= & \tight  \mid \PT \mid \A \rew \sig \\
     \mbox{(Multi-set Types)} & \A & :: = & \mult{\sig_k}_{\kK} \\
   \end{array} \hfill     \]
Types in $\tight$, which can be seen as a refinement of the base
  type $\typend$, denote the so-called tight types. The constant
$\typendm$ denotes the type of any term reducing to a canonical form, while $\typendn$
  denotes the type of any term reducing to a pure canonical
  form.  We write $\istight{\sig}$, if $\sig$ is of the form
$\typendm$ or $\typendn$ (we use  $\typend$ to denote either form). We extend this notion to multi-sets of types and typing contexts as expected (that is, $\istight{\mult{\sig_i}_\iI}$ if $\istight{\sig_i}$ for all $\iI$, and $\istight{\Gam}$ if $\istight{\Gam(x)}$, for all $x \in \dom{\Gam}$. 
   
The crucial idea behind the grammar of types is to distinguish between
{\it consuming} constructors typed with standard types, and {\it
  persistent} constructors typed with tight types, as hinted in the
introduction. A constructor is consuming (resp. persistent) if it is
consumed (resp. not consumed) during head-reduction.  Indeed, the pair
constructor is consumed (on the pattern side as well as on the term
side) during the execution of the pattern matching rule
$\matchs$. Otherwise, patterns and pairs are persistent, and they do
appear in the normal form of the original term.  This dichotomy
between consuming and persistent constructors is reflected in the
typing system by using different typing rules to type them,
notably for the abstraction, the application, the pair
terms and the pair patterns.

The \deft{type assignment system} \Exact, given in
Fig.~\ref{f:typing-rules}, introduces
  tuples/counters in the typing judgments, where the intended meaning
  is explained below. We use
$\Phi \tri \amuju{b}{e}{m}{f}{\Gam}{t:\sig}$ (resp.
$\Phi \tri \amuju{b}{e}{m}{f}{\Gam}{t:\A}$) to denote \deft{term type
  derivations} ending with the sequent
$\amuju{b}{e}{m}{f}{\Gam}{t:\sig}$
(resp. $\amuju{b}{e}{m}{f}{\Gam}{t:\A}$), and
$\Pi \tri \Gam \pder^{(e,m,f)} p : \A$ to denote \deft{pattern type
  derivations} ending with the sequent $\Gam \pder^{(e,m,f)} p : \A$.
Often in examples, we will use the notation $\Phi^{(b,e,m,f)}$
(resp. $\Pi^{(e,m,f)}$) to refer to a term derivation (resp. pattern
derivation) ending with a sequent annotated with indexes $(b,e,m,f)$
(resp. $(e,m,f)$).

As mentioned in the introduction, exact bounds
  can only be   extractable from {\it minimal} derivations. In our
  framework the notion of minimality is implemented by
  means of tightness~\cite{AccattoliGK18}.
  We say that a derivation
$\Phi \tri \amuju{b}{e}{m}{f}{\Gam}{t:\sig}$ (resp.
$\Phi \tri \amuju{b}{e}{m}{f}{\Gam}{t:\A}$) is \deft{tight}, denoted
by $\istight{\Phi}$, if and only if $\istight{\Gam}$ and
$\istight{\sig}$ (resp.  $\istight{\A}$).  The size of a derivation
$\tderiv$ is defined as in System \Upper.
  
The tuple $(b,e,m,f)$ in type (multi-set type) derivations is used to
compute bounds for head-reductions sequences reaching a head normal
form. That is, we will show that if $\Gam \vdash^{(b,e,m,f)} t:\sig$
is tight, then $t \Rewparam{(b,e,m)}{\head} v$, where $b$ is the
number of $\betas$-steps, $e$ the number of $\exps$-steps, $m$ the
number of $\matchs$-steps and $f$ is the size of the head
normal-form $v$.  The tuple $(e,m,f)$ in pattern derivations gives the
potentiality of pattern $p$. That is, if $\Gam
\pder^{(e,m,f)} p:\A$ then the pattern $p$ represents the generation of $e$
substitution steps, $m$ matching steps and $f$ blocked patterns
during a head evaluation.

\begin{figure}[!h]
\begin{framed}
\begin{center}
\[ \begin{array}{c}
 \infer[(\trvarpat)]{\phantom{AAA}}{x:\A  \pder^{(1,0,0)} x:\A} 
\\ \\
\infer[(\trpairpat)]{\Gam \pder^{(e_p,m_p,n_p)} p: \A \sep \Delta \pder^{(e_q,m_q,n_q)} q:\B \sep p \# q }
       { \Gam \inter  \Delta \pder^{(e_p+e_q,1+m_p+m_q,n_p+n_q)} \pair{p}{q} :
       \mult{\prodt{\A}{\B}} }\  \\\\
       \infer[(\pattn)]{\Gam \subseteq \var{\pair{p}{q}} \sep
       \istight{\Gam}}{\Gam  \pder^{(0,0,1)}  \pair{p}{q}: \mult{\typendn} }  \\ \\
    \hline \\
       \infer[(\ax)]
      {  } 
      {\amuju{0}{0}{0}{0}{x: \mult{\sig}}{  x: \sig} } \\ \\ 
       \infer[(\introarrow)]{ \amuju{b_t}{e_t}{m_t}{f_t}{\Gam}{t : \sig}\sep
       \Gam|_{p} \pder^{(e_p,m_p,f_p)}  p:\A}  
      {\amuju{b_t+1}{e_t+e_p}{m_t + m_p}{f_t + f_p}{\cmin{\Gam}{\var{p}}}{\l p.t: \A \rew  \sig}}\\ \\ 
       \infer[(\ruletight)]{\amuju{b}{e}{m}{f}{\Gam}{t:\tight} \sep \istight{\Gam|_{p}}}
     {\amuju{b}{e}{m}{f+1}{\cmin{\Gam}{\var{p}}}{\l p.t:\typendm}}  \\ \\
\infer[(\many)]{(\amuju{b_k}{e_k}{m_k}{f_k}{\Gamk}{t: \sigk})_{\kK}}
       {\amuJu{+_{\kK}b_k}{+_{\kK}e_k}{+_{\kK} m_k}{+_{\kK}f_k}{\inter_{\kK}\Gamk}{t: \mult{\sigk}_{\kK}}} \\ \\
       
\infer[(\app)]{ \amuju{b_t}{e_t}{m_t}{f_t}{\Gam}{t: \A \rew \sig} \sep\sep 
       \amuJu{b_u}{e_u}{m_u}{f_u}{\Del}{u : \A }}
       {\amuju{b_t+b_u}{e_t+e_u}{m_t+m_u}{f_t+f_u}{\Gam \inter \Del}{ t\,u: \sig}}  \\ \\
\infer[(\appn)]{ \amuju{b_t}{e_t}{m_t}{f_t}{\Gam}{t: \typendn } }
       {\amuju{b_t}{e_t}{m_t}{f_t+1}{\Gam  }{ t\,u: \typendn}}  \\ \\

       \infer[(\trpair)]{\amuju{b_t}{e_t}{m_t}{f_t}{\Gam}{t:\A} \sep
       \amuju{b_u}{e_u}{m_u}{f_u}{\Del}{u:\B}}
       {\amuju{b_t+b_u}{e_t+b_u}{m_t+m_u}{f_t+f_u}{\Gam \inter\Del}{\pair{t}{u} : \prodt{\A}{\B}}} \\ \\
       \infer[(\tnormalpair)]{\phantom{AA}}
       {\amuju{0}{0}{0}{1}{}{\pair{t}{u} : \typendm}} \\\\ 
       \infer[(\trsub)]{\amuju{b_t}{e_t}{m_t}{f_t}{\Gam}{t:\sig} \sep
       \Gam|_{p} \pder^{(e_p,m_p,f_p)} p:\A \sep
       \amuJu{b_u}{e_u}{m_u}{f_u}{\Del}{u:\A}}
       {\amuju{b_t+b_u}{e_t+e_u+e_p}{m_t+m_u+ m_p}{f_t+f_u
       + f_p  }{(\cmin{\Gam}{\var{p}}) \inter  \Del}{t[p/u]:\sig}}
\end{array} \]    
\end{center}
\caption{Typing System \Exact}
\label{f:typing-rules}
\end{framed}
\end{figure}

We now give some
  intuition behind the typing rules in
  Fig.~\ref{f:typing-rules}, by addressing
 in particular the
    consumable/persistent paradigm.

\noindent $\bullet$ Rule $\ax$: Since $x$ is itself
a head normal-form, it will not generate any $\betas$, $\exps$ or $\matchs$ steps,
  and its size is $0$.
  
\noindent $\bullet$ Rule $\introarrow$: Used to type abstractions
$\lambda p.t$ to be applied (\ie\ consumed), therefore it has a
functional type $\A \rew \sig$. Final indexes of the abstraction
are obtained from the ones of the body and the pattern, and $1$ is
added to the first index since the abstraction will be consumed by a $\betas$-reduction step.
  
  \noindent $\bullet$ Rule $\ruletight$: Used to type abstractions
  $\lambda p.t$ that are not going to be applied/consumed (they are persistent).
  Only the last index (size of the normal form)
  is incremented by one since the abstraction remains in the normal
  form (the abstraction is persistent). Note that both body $t$ and
  variables in $p$ should be typed with a tight type. 
  
\noindent $\bullet$ Rule $\app$: Types applications $tu$ where $t$ will eventually
  become an abstraction, and thus the application constructor
    will be consumed. Indexes for $tu$ are exactly
  the sum of the indexes for $t$ and $u$. Note that we do not need to
  increment the counter for $\betas$ steps, since this was already
  taken into account in the $\introarrow$ rule.
  
\noindent $\bullet$ Rule $\appn$: Types applications $tu$ where $t$ is neutral,
  therefore will never become an abstraction, and the application
  constructor becomes persistent. Indexes are
  the ones for $t$, adding one to the (normal term) size to count
  for the (persistent) application.
  
  \noindent $\bullet$ Rule $\trpair$: Types pairs consumed
    during some matching step. We add the indexes for the
two components of the pair without incrementing the number of
$\matchs$ steps, since it is incremented when
 typing a consuming abstraction, with rule $\introarrow$.
  
 \noindent $\bullet$ Rule $\tnormalpair$: Used to type pairs that are
not consumed in a matching step (they are
persistent), therefore appear in the head normal-form.
Since the pair is already a head normal-form its indexes are zero except for the
size, which counts the pair itself.
  
\noindent $\bullet$ Rule $\trsub$: Note that we do not
need separate cases for consuming and persistent explicit matchings, since in both cases
typable occurrences of $u$ represent potential reduction steps for
$u$, which need to be taken into account in the final counter of the
term.

\noindent $\bullet$ Rule $\trvarpat$: Typed variables always
generate one $\exps$ and zero $\matchs$ steps, even when erased.

\noindent $\bullet$ Rule $\trpairpat$: Used when the pattern has a product type, which means that the pattern will be matched with a pair. We add the counters for the two components of the pair and increment the counter for the $\matchs$ steps.

\noindent $\bullet$ Rule $\pattn$: Used when the pattern has a tight
type, which means that it will not be matched with a pair and
therefore will be blocked (it is persistent). This kind of pairs generate
zero $\exps$ and $\matchs$ steps, and will contribute with one blocked
pattern to the size of the normal form.

The system enjoys the following key properties, easily  proved by induction:
\begin{lemma}
    \label{l:relevance+clash-free}
    Let $\Phi \tri \amuju{b}{e}{m}{f}{\Gam}{t:\sig}$. Then, 
    \begin{enumerate}
    \item $(${\bf Relevance}$)$ \label{relevance}
      $\dom{\Gam} \subseteq \fv{t}$.
    \item $(${\bf Clash-Free}$)$ \label{clash-free}
    $t$ is (head) clash-free. 
  \end{enumerate}
\end{lemma}

We now discuss two examples.
\begin{example}\label{e:second}
  Let us consider $t_0=(\l \pair{x}{y}. (\l \pair{w}{z}. wyz)x) \pair{\pair{\K}{a}}{b}$, with the following head-reduction sequence:
 \[\begin{array}{llll}
 & (\l \pair{x}{y}. (\l \pair{w}{z}. wyz)x) \pair{\pair{\K}{a}}{b} & \Rew{\betas} &((\l \pair{w}{z}. wyz)x) \esub{ \pair{x}{y}}{\pair{\pair{\K}{a}}{b}} \\ \Rew{\betas} & (wyz)\esub{\pair{w}{z}}{x}  \esub{ \pair{x}{y}}{\pair{\pair{\K}{a}}{b}} & \Rew{\matchs} & (wyz)\esub{\pair{w}{z}}{x}  \esub{x}{\pair{\K}{a}}\esub{y}{b}\\
\Rew{\exps} & (wyz)\esub{\pair{w}{z}}{\pair{\K}{a}}  \esub{y}{b} & \Rew{\matchs} & (wyz)\esub{w}{\K}\esub{z}{a}  \esub{y}{b} \\
 \Rew{\exps} & (\K yz) \esub{z}{a}  \esub{y}{b} & \Rew{\betas} & ((\l y_1. x_1)\esub{x_1}{y}z) \esub{z}{a}  \esub{y}{b} \\
 \Rew{\betas} & x_1\esub{y_1}{z} \esub{x_1}{y} \esub{z}{a}  \esub{y}{b} & \Rew{\exps} & x_1 \esub{x_1}{y} \esub{z}{a}  \esub{y}{b} \\
     \Rew{\exps} & y \esub{z}{a}  \esub{y}{b}   & \Rew{\exps} &
   y   \esub{y}{b}  \\
 \Rew{\exps} & b
 \end{array}
 \] 
 Note that, there are two matching steps in the head-reduction
 sequence, but the second step is only created after the substitution
 of $x$ by $\pair{\K}{a}$. Our method allows us to extract this information
 from the typing derivations because of the corresponding types
 for $\pair{x}{y}$ and $\pair{w}{z}$. Indeed, both patterns are typed with a product type (\cf\ the forthcoming tight typing derivations),  and therefore the corresponding pairs are consumed and not persistent.

Since $t_0= (\l \pair{x}{y}. (\l \pair{w}{z}. wyz)x) \pair{\pair{\K}{a}}{b} \Rewparam{(4,6,2)}{\head}b$,  the term $t_0$ should be tightly typable
 with counter $(4,6,2,0)$, where $0$ is  the size of $b$.
 In the construction of such  tight derivation we proceed by pieces.
 Let $T_\K = \mult{\typendn} \rew \emul \rew \typendn$. We first
 construct the following pattern derivation for $\pair{w}{z}$:
\[ \Pi_{\pair{w}{z}}\tri \ \inferrule{ w: \mult{T_\K} \pder^{(1,0,0)} w: \mult{T_\K} \sep
            \pder^{(1,0,0)} z: \emm}
      { w: \mult{T_\K} \pder^{(2,1,0)} \pair{w}{z}: \mult{\prodt{\mult{T_\K}}{\emm}}}\]
         In the following  $T_{\pair{w}{z}} = \mult{\prodt{\mult{T_\K}}{\emm}}$. We construct a similar pattern derivation for $\pair{x}{y}$:
         \[\Pi_{\pair{x}{y}}\tri\ \inferrule{ x: T_{\pair{w}{z}} \pder^{(1,0,0)} x:
             T_{\pair{w}{z}}  \sep
            y: \mult{\typendn}\pder^{(1,0,0)} y: \mult{\typendn}}
     { x: T_{\pair{w}{z}}; y: \mult{\typendn} \pder^{(2,1,0)} \pair{x}{y}: \mult{ \prodt{T_{\pair{w}{z}}}{\mult{\typendn}} }}\]

     In the rest of the example $T_{\pair{x}{y}} = \mult{\prodt{T_{\pair{w}{z}}}{\mult{\typendn}}}$. We build a type derivation for $\l \pair{x}{y}. (\l \pair{w}{z}. wyz)x$, where $\Gam_w = w: \mult{T_\K}$, $\Gam_y =y: \mult{\typendn}$, $\Gam = \Gam_w ; \Gam_y$, and $\Gam_x = x: T_{\pair{w}{z}}$. Furthermore, in this example and throughout the paper, we will use $(\ov{0})$ to denote the tuple $(0,0,0,0)$.
\[ 
\Phi_1 \tri  \inferrule{\inferrule*{\inferrule*{\inferrule*{\inferrule*{\Gam_w \vdash^{(\ov{0})} w :
            T_\K \sep \Gam_y \vdash^{(\ov{0})}
            y:\mult{\typendn} }{\Gam\vdash^{(\ov{0})} wy: \emul \rew \typendn}
           \vdash^{(\ov{0})} z:\emul}{\Gam\vdash^{(\ov{0})}
          wyz: \typendn  }  \Pi_{\pair{w}{z}}^{(2,1,0)} }{\Gam_y \vdash^{(1,2,1,0)} \l \pair{w}{z}. wyz: T_{\pair{w}{z}}  \rew \typendn} \Gam_x \vdash^{(\ov{0})} x:T_{\pair{w}{z}} }{\Gam_y; \Gam_x\vdash^{(1,2,1,0)} (\l \pair{w}{z}. wyz)x: \typendn } \Pi_{\pair{x}{y}}^{(2,1,0)}}{ \vdash^{(2,4,2,0)} \l \pair{x}{y}. (\l \pair{w}{z}. wyz)x: T_{\pair{x}{y}}\rew \typendn } \]
\[\Phi_{\K} \tri \inferrule{\inferrule*{x_1: \mult{\typendn} \vdash^{(\ov{0})} x_1: \typendn \sep  \pder^{(1,0,0)} y_1: \emm }{x_1: \mult{\typendn} \vdash^{(1,1,0,0)} \l y_1. x_1:  \emm \rew \typendn} \sep x_1: \mult{\typendn} \pder^{(1,0,0)} x_1: \mult{\typendn} }{ \vdash^{(2,2,0,0)} \K: T_\K}
\]

From $\Phi_1$ and $\Phi_{\K}$ we build the following tight derivation for $t_0= (\l \pair{x}{y}. (\l \pair{w}{z}. wyz)x) \pair{\pair{\K}{a}}{b}$:
\[ \Phi \tri \inferrule{\Phi_1^{(2,4,2,0)} 
    \inferrule*{\inferrule*{\inferrule*{\inferrule*{\inferrule*{\Phi_{\K}^{(2,2,0,0)}}{\vdash^{(2,2,0,0)} \K : \mult{T_\K}} \sep \inferrule*{}{\vdash^{(\ov{0})} a:\emm}}{\vdash^{(2,2,0,0)} \pair{\K}{a}: \prodt{\mult{T_\K}}{\emm} }}{\vdash^{(2,2,0,0)}\pair{\K}{a}: T_{\pair{w}{z}}} \sep b: \mult{\typendn} \vdash^{(\ov{0})} b: \mult{\typendn}}{b: \mult{\typendn} \vdash^{(2,2,0,0)}\pair{\pair{\K}{a}}{b}: \prodt{T_{\pair{w}{z}}}{\mult{\typendn}}}}{b: \mult{\typendn} \vdash^{(2,2,0,0)} \pair{\pair{\K}{a}}{b}: \mult{ \prodt{T_{\pair{w}{z}}}{\mult{\typendn}} } }}
   {b: \mult{\typendn} \vdash^{(4,6,2,0)}(\l \pair{x}{y}. (\l \pair{w}{z}. wyz)x) \pair{\pair{\K}{a}}{b}:\typendn}\]
Therefore, $\Phi^{(4,6,2,0)}$ gives the expected exact bounds.
It is worth noticing that the pair $\pair{\pair{\K}{a}}{b}$ is typed here
 with a singleton multi-set, while it would be typable with a multi-set
 having at least two elements in the typing system proposed in~\cite{BKRDR15},
 even if the term is not going to be duplicated. 

\end{example}

\begin{example}\label{e:third}
  We now consider the  term $t_1= (\l z. \l \pair{x}{y}. \id) z z) \pair{u}{v}$,
  having the following  head-reduction sequence to head normal-form:
\[ \begin{array}{llll}
    & (\l z. (\l \pair{x}{y}. \id) z z) \pair{u}{v} & \Rew{\betas} & 
  ((\l \pair{x}{y}. \id) z z) \esub{z}{\pair{u}{v}} \\
  \Rew{\betas} & (\id\esub{\pair{x}{y}}{z} z) \esub{z}{\pair{u}{v}} & \Rew{\betas} &
  w\esub{w}{z}\esub{\pair{x}{y}}{z} \esub{z}{\pair{u}{v}} \\
  \Rew{\exps} & z\esub{\pair{x}{y}}{z} \esub{z}{\pair{u}{v}} & \Rew{\exps} & \pair{u}{v}\esub{\pair{x}{y}}{\pair{u}{v}}\\
  \Rew{\matchs} & \pair{u}{v}\esub{x}{u}\esub{y}{v} & \Rew{\exps} &\pair{u}{v}\esub{y}{v}\\
  \Rew{\exps} & \pair{u}{v} 
       \end{array} \] 

Note that the pair $\pair{u}{v}$ is copied twice during the reduction, but only one of the copies is consumed by a matching. The copy of the pair that is not consumed will persist in the term, therefore it will be typed with $\typendm$. The other copy will be consumed in a matching step, however its components are not going to be used, therefore we will type it with $\mult{\oprod}$, where $\oprod$ denotes $\prodt{\ems}{\ems}$.

We need to derive a tight derivation with counter $(3,4,1,1)$
for $t_1$. We first consider the following derivation:
\[\Phi_1 \tri \inferrule{\inferrule*{\inferrule*{}{w: \mult{\typendm} \vdash^{(\ov{0})} w: \typendm} \sep w: \mult{\typendm} \pder^{(1,0,0)} w: \mult{\typendm}}{\vdash^{(1,1,0,0)} \id: \mult{\typendm} \rew \typendm }\sep \inferrule*{\pder^{(1,0,0)} x : \emm \sep \pder^{(1,0,0)} y : \emm}{\pder^{(2,0,0)} \pair{x}{y} : \mult{\oprod}}}{\vdash^{(2,3,1,0)} \l \pair{x}{y}. \id: \mult{\oprod} \rew \mult{\typendm} \rew \typendm}
\]
From $\Phi_1$ we obtain the following derivation:

\[
\Phi \tri\inferrule{\inferrule*{\inferrule*{ \Phi_1^{(2,3,1,0)} \sep \inferrule*{}{z: \mult{\oprod} \vdash^{(\ov{0})} z: \mult{\oprod} }}{z: \mult{\oprod} \vdash^{(2,3,1,0)} (\l \pair{x}{y}. \id) z: \mult{\typendm} \rew \typendm } \sep \inferrule*{}{z: \mult{\typendm} \vdash^{(\ov{0})} z: \mult{\typendm}}}{z: \mult{\oprod, \typendm} \vdash^{(2,3,1,0)} (\l \pair{x}{y}. \id) z z: \typendm}  z: \mult{\oprod, \typendm} \pder^{(1,0,0)} z: \mult{\oprod, \typendm}}{\vdash^{(3,4,1,0)} \l z. (\l \pair{x}{y}. \id) z z: \mult{\oprod, \typendm} \rew \typendm}
\]

Using $\Phi$ we obtain the following tight derivation, and its expected counter: 
\[ \inferrule*{\Phi^{(3,4,1,0)}  \sep \inferrule*{\inferrule*{\vdots}{\vdash^{(\ov{0})} \pair{u}{v}: \oprod} \sep
    \inferrule*{}{\vdash^{(0,0,0,1)} \pair{u}{v}: \typendm}}{
    \vdash^{(0,0,0,1)} \pair{u}{v}: \mult{ \oprod,\typendm }  }}{\vdash^{(3,4,1,1)} (\l z. (\l \pair{x}{y}. \id) z z) \pair{u}{v}: \typendm}\]

\end{example}

\section{Correctness for System \Exact}
\label{s:soundness}
This section studies the implication ``tight typability implies
head-normalisable''. The two key properties used to show this
implication are {\it minimal counters for canonical forms}
(Lem.~\ref{l:M-tight}) and the {\it exact subject reduction
  property} (Lem.~\ref{l:subject-reduction}).  Indeed,
Lem.~\ref{l:M-tight} guarantees that a tight derivation for a
canonical form $t$ holds the right counter of the form
$(0,0,0,\size{t})$.  Lem.~\ref{l:subject-reduction}
gives in fact an {\it
  (exact) weighted subject reduction} property, weighted because
head-reduction strictly decreases the counters of typed terms, and
exact because only {\it one} counter is decreased by $1$ for each
head-reduction step. Subject reduction is based on a {\it substitution}
  property (Lem.~\ref{l:substitution}). We start with
a key auxiliary lemma.

\begin{toappendix}
\begin{lemma}[Tight Spreading]
  \label{l:tight-spreading}
  Let $t \in \N$. Let $\tderiv \rhd \amuju{b}{e}{m}{f}{\Gam}{t:\sig}$
  be a typing derivation such that $\istight{\Gam}$. Then $\sig$ is tight
  and the last rule of $\tderiv$ does not belong to $\{\app, \introarrow,
  \ruletight, \trpair, \tnormalpair\}$.
\end{lemma}
\end{toappendix}

  \begin{proof}
  By induction on $t\in \N$, taking into account the fact that $t$ is
  not an abstraction nor a pair (\cf\ App.~\ref{a:correctness}).
\end{proof}

\begin{toappendix}
\begin{lemma}[Canonical Forms and Minimal Counters]
  \label{l:M-tight}
  Let $\tderiv \rhd \amuju{b}{e}{m}{f}{\Gam}{t: \sig}$ be
  a tight derivation. Then
  $t \in \M$ if and only if $b = e = m = 0$.
  \end{lemma}
\end{toappendix}
\begin{proof}
  The left-to-right implication  is by induction on the definition of
  the set $\M$,  using Lem.~\ref{l:tight-spreading} for the cases of
  application and explicit matching.
  The right-to-left implication  is by induction on $\tderiv$
  and uses Lem.~\ref{l:tight-spreading} (\cf\ App.~\ref{a:correctness}). 
\end{proof}

\begin{toappendix}
\begin{lemma}[Substitution for System \Exact]
  \label{l:substitution}
  If $\Phi_t \tri  \amuju{b_t}{e_t}{m_t}{f_t}{\Gam; x:\A}{t:\sig}$, and
  $\Phi_u \tri  \amuju{b_u}{e_u}{m_u}{f_u}{\Del}{u:\A}$, then there
  exists $\Phi_{t\isub{x}{u}} \tri  \amuju{b_t+b_u}{e_t+e_u}{m_t+m_u}{f_t+f_u}{\Gam \inter \Del}{t\isub{x}{u}:\sig}$. 
\end{lemma}
\end{toappendix}

\begin{proof}
  By induction on $\Phi_t$ (\cf\ App.~\ref{a:correctness}). 
\end{proof}

  \begin{toappendix}
\begin{lemma}[Exact Subject Reduction]
  \label{l:subject-reduction}
  If $\Phi \tri  \amuju{b}{e}{m}{f}{\Gam}{t:\sig}$, and
  $t \Rew{\head} t'$, then $\Phi' \tri  \amuju{b'}{e'}{m'}{f}{\Gam}{t':\sig}$, where
  \begin{itemize}
  \item If $\R = \betas$, then $b' = b-1$, $e'=e$, $m'=m$.
  \item If $\R = \exps$, then $b' = b$, $e'=e-1$, $m'=m$.
  \item If $\R = \matchs$, then $b' = b$, $e'=e$, $m'=m-1$.
  \end{itemize}  
\end{lemma}
\end{toappendix}

\begin{proof} By induction on $\Rew{\head}$, using
  the substitution property (Lem.~\ref{l:substitution})  (\cf\ App.~\ref{a:correctness}). 
\end{proof}

  Lem.~\ref{l:subject-reduction} provides a simple argument to obtain the
  implication ``tigthly typable implies head-normalisable'': if $t$ is
  tightly typable, and reduction decreases the counters, then
  head-reduction necessarily terminates. But the soundness implication
  is in fact more precise than that. Indeed: 

   \begin{theorem}[Correctness]
     \label{t:correctness}
     Let $\Phi \tri  \amuju{b}{e}{m}{f}{\Gam}{t:\sig}$ be a tight derivation.
     Then there exists $u \in \M$ and a reduction sequence $\rho$ such that
      $\rho: t \Rewparam{(b,e,m)}{\head} u$
     and $\size{u} = f$. 
   \end{theorem}

        \begin{proof}
  By induction on $b + e + m$.

  If $b + e + m=0$ (\ie\ $b=e=m=0$), then Lem.~\ref{l:M-tight} gives
  $t \in \M$, so that $t \notRew{\head}$ holds by
  Prop.~\ref{p:M-not-reducible}.  We let $u:= t$ and thus
  $t \Rewparam{(0,0,0)}{\head} t$. It is easy to show
    that tight derivations
    $\Phi \tri  \amuju{0}{0}{0}{f}{\Gam}{t:\sig}$
    for terms in $\M$ verify
    $\size{t} = f$.

  If $b + e + m>0$, we know by Lem.~\ref{l:M-tight}
  that $t \notin \M$, and we know by Lem.~\ref{l:relevance+clash-free}:\ref{clash-free}
  that $t$ is (head) clash-free. Then, $t$ turns to be head-reducible
  by Prop.~\ref{p:M-not-reducible}, \ie\ 
  there exists $t'$ such that $t \Rew{\head} t'$.
By Lem.~\ref{l:subject-reduction} there is a derivation
$\tderiv' \tri \amuju{b'}{e'}{m'}{f}{\Gam}{t':\sig}$ such that
$b'+e'+m' +1 =  b + e + m$. The \ih\ applied to $\tderiv'$
then gives $t' \Rewparam{(b',e',m')}{\head} u$ and
$\size{u} = f$. We conclude with the sequence
$t \Rewparam{}{\head} t' \Rewparam{(b',e',m')}{\head} u$,
with the counters as expected.

\end{proof}

\section{Completeness for System \Exact}
\label{s:completeness}
In this section we study the reverse implication ``head-normalisable
implies tight typability''. In this case the key properties are the
existence of {\it tight derivations for canonical forms}
(Lem.~\ref{l:M-typable-tight}) and the {\it subject expansion property}
(Lem.~\ref{l:subject-expansion}).  As in the previous section these
properties are {\it (exact) weighted} in the sense that
Lem.~\ref{l:M-typable-tight} guarantees that a canonical form $t$ has a
tight derivation with the right counter, and Lem.~\ref{l:subject-expansion} shows that each
step of head-expansion strictly increases exactly one of the counters
of tightly typed terms.  Subject expansion relies on an {\it
  anti-substitution} property (Lem.~\ref{l:anti-substitution}).

\begin{toappendix}
\begin{lemma}[Canonical Forms and Tight Derivations]
     \label{l:M-typable-tight}
Let $t \in \M$. There exists a tight derivation $\tderiv \rhd \amuju{0}{0}{0}{\size{t}}{\Gam}{t: \tight}$.
\end{lemma}
\end{toappendix}

\begin{proof} We generalise the property to the two following
  statements: 
  \begin{itemize}
  \item If $t \in \N$, then there exists a tight derivation $\tderiv
    \rhd \amuju{0}{0}{0}{\size{t}}{\Gam}{t: \typendn}$. 
  \item If $t \in \M$, then there exists a tight derivation $\tderiv \rhd \amuju{0}{0}{0}{\size{t}}{\Gam}{t: \tight}$.
  \end{itemize}
The proof then proceeds by induction on $\N, \M$, using relevance
(Lem.~\ref{l:relevance+clash-free}:\ref{relevance}) (\cf\ App.~\ref{a:completeness}).

\end{proof}

  \begin{toappendix}
 \begin{lemma}[Anti-Substitution for System \Exact]
  \label{l:anti-substitution}
  Let $\Phi \tri  \amuju{b}{e}{m}{f}{\Gam}{t\isub{x}{u}:\sig}$.
  Then, there exist derivations $\Phi_t$, $\Phi_u$,
  integers $b_t, b_u, e_t, e_u, m_t, m_u, f_t, f_u$,
  contexts $\Gam_t, \Gam_u$, and multi-type $\A$ such that
  $\Phi_t \tri  \amuju{b_t}{e_t}{m_t}{f_t}{\Gam_t; x:\A}{t:\sig}$, 
  $\Phi_u \tri  \amuju{b_u}{e_u}{m_u}{f_u}{\Gam_u}{u:\A}$,
  $b=b_t+ b_u$,
  $e = e_t+ e_u$,
  $m= m_t+ m_u$,
  $f=f_t+ f_u$, and 
   $\Gam = \Gam_t \inter \Gam_u$.
     \end{lemma}
\end{toappendix}

\begin{proof}
  By induction on $\Phi$ (\cf\ App.~\ref{a:completeness}).
\end{proof}

  \begin{toappendix}
\begin{lemma}[Exact Subject Expansion]
\label{l:subject-expansion}
If $\Phi' \tri  \amuju{b'}{e'}{m'}{f'}{\Gam}{t':\sig}$, and
    $t \Rew{\head} t'$, then $\Phi \tri  \amuju{b}{e}{m}{f}{\Gam}{t:\sig}$, where
    \begin{itemize}
    \item If $\R = \betas$, then $b = b'+1$, $e'=e$, $m'=m$.
    \item If $\R = \exps$, then $b' = b$, $e=e'+1$, $m'=m$.
    \item If $\R = \matchs$, then $b' = b$, $e'=e$, $m=m'+1$.
   \end{itemize}  
\end{lemma}
\end{toappendix}

\begin{proof} By induction on $\Rewbis{\head}$, using
  the anti-substitution property (Lem.~\ref{l:anti-substitution})
  (\cf\ App.~\ref{a:completeness}).
\end{proof}

The previous lemma provides a simple argument to obtain the
implication ``head-normalisable implies tigthly typable'', which
can in fact be stated in a more precise way: 

\begin{theorem}[Completeness]
  \label{t:completeness}
  Let $t$ be a head-normalising term such that
  $t \Rewparam{(b,e,m)}{\head} u$, $u \in \M$. 
  Then there exists a tight derivation
  $\Phi \tri  \amuju{b}{e}{m}{\size{u}}{\Gam}{t:\tight}$.
\end{theorem}

\begin{proof}
  By induction on  $b+e+m$.
  \begin{itemize}
  \item If $(b+e+m) = 0$ then $t=u \in \M$, therefore $\amuju{0}{0}{0}{\size{t}}{\Gam}{t:\tight}$, by Lem.~\ref{l:M-typable-tight}.
  
  \item If $(b+e+m) > 0$, then $t  \Rew{\head} t' \Rewparam{(b',e',m')}{\head} u$, where $b'+e'+m'+1 = b+e+m$. By the \ih $\amuju{b'}{e'}{m'}{\size{u}}{\Gam}{t':\tight}$. Then from Lem.~\ref{l:subject-expansion}, it follows that $\amuju{b}{e}{m}{\size{u}}{\Gam}{t:\tight}$. 
  \end{itemize}
\end{proof}

In summary, soundness and completeness do not only
establish an equivalence between
tight typability and head-normalization, but they provide a
much refined equivalence property  stated as follows:

\begin{corollary}
  Given a term $t$, the following statements are equivalent
  \begin{itemize}
  \item There is a tight derivation $\Phi \tri  \amuju{b}{e}{m}{s}{\Gam}{t:\tight}$. 
  \item There exists a canonical form $u \in \M$  such that
     $t \Rewparam{(b,e,m)}{\head} u$ and $\size{u}=s$.
  \end{itemize}
    
\end{corollary}  

\section{Conclusion}
\label{s:conclusion}

This paper provides a quantitative insight of pattern matching by using
{\it statical} tools to study {\it dynamical} properties.  Indeed, our
(static) typing system \Upper (resp. \Exact) provides
upper (resp. exact) bounds  about time and space
properties related to (dynamic) computation. More precisely, the
tuples of integers in the conclusion of a {\it tight} \Exact-derivation for a
term $t$ provides the exact {\it length} of the head-normalisation sequence
of $t$ and the {\it size} of its normal form. Moreover, the length of
the normalisation sequence is {\it discriminated} according to
different kind of steps performed to evaluate $t$.

Future work includes generalisations to more powerful notions of
(dynamic) patterns, and to other reduction strategies
for pattern calculi, as well to programs with recursive
schemes. Inhabitation for our typing system is conjectured to be
decidable, as the one in~\cite{BKRDR15}, but this still needs to be 
formally proved, in which case the result
"solvability = typing+ inhabitation" in opt. cit. would be
restated
  in a simpler framework.
The quest of a clear notion of model for
pattern calculi also remains open.

Last, but not least, time cost analysis of a language with
constructors and pattern matching is studied in~\cite{AccattoliB17},
where it is shown that evaluation matching rules other than
$\beta$-reduction may be negligible, depending on the reduction
strategy and the specific notion  of value.  The quantitative
technical tools that we provide in this paper should be able to prove
the same result by means of a tight type system.

\renewcommand{\em}{\it}
\bibliographystyle{plainurl}

\newpage

\appendix
\section{The \Upper\ Typing System}
\label{a:typing-system}

{\bf Lemma \ref{l:relevance+clash-free-upper}}
    Let $\Phi \tri \seq{\Gam}{t:\sig}$. Then,\smallskip
\begin{enumerate}
\item \gettoappendix {relevance-upper}
\item \gettoappendix {clash-free-upper}
\item \gettoappendix {i:qsr}
\item \gettoappendix {i:se}
\end{enumerate}

\begin{proof}
  Let $\Phi \tri \seq{\Gam}{t:\sig}$.
  \begin{enumerate}
    \item By straightforward induction on $\Phi$. Note that $\Gam = (\cmin{\Gam}{\var{p}});\Gam|_{p}$ in both $(\introarrow)$ and $(\trsub)$ rules.
    \item By induction on $\Phi$ and the syntax-directed aspect of  system \Upper.

The base case for rule $(\ax)$ is trivial and the cases for
rules $(\many), (\introarrow), (\trpair)$ are straightfoward from the
\ih\ We present the case for rule $(\app)$, the case for rule
$(\trsub)$ being similar.

Let $\Phi$ be of the form
\[\infer{\Phi_u \tri \seq{\Gam_u}{u: \A \rew \sig} \sep\sep 
       \Phi_v \tri \seq{\Gam_v}{v : \A }}
     {\seq{\Gam_u \inter \Gam_v}{ u\,v: \sig}}\]
   By \ih\ both $u$ and $v$ are (head) clash-free. Note that $u$ cannot be
   of the form $\appctx{\slist}{\pair{u_1}{u_2}}$ thus is either
   $\appctx{\slist}{x}$ or $\appctx{\slist}{\l p. u'}$.

If $u$ is $\appctx{\slist}{x}$ then $t = u\,v$ is (head) clash-free.

If $u$ is $\appctx{\slist}{\l p. u'}$ then
in order to ensure $t = \appctx{\slist}{\l p. u'}v$ is (head) clash-free we need to
guarantee that either $p$ is not a pair or $v$ is not of the form $\appctx{\slist'}{\l q. v'}$. We analyse the two possibilities for $\A$:
\begin{itemize}
\item If $\A = \mult{\prodt{\A_1}{\A_2}}$ then $v$
cannot be of the form $\appctx{\slist'}{\l
    q. v'}$, which can only be typed with functional types. Therefore, $t$ is (head)  clash-free.
  
\item If $\A \neq \mult{\prodt{\A_1}{\A_2}}$ then $p$ is a variable thus $t$ is (head) clash-free.
\end{itemize}

    \item By induction on $\Rew{\pattern}$ (resp. $\Rew{\head}$) and
(Substitution) Lemma~\ref{l:substitution-upper} below. The
  first three cases represent the base cases for both
  reductions, where
the size relation is strict.
\begin{itemize}
    \item $t= \appctx{\slist}{\l p. v}u \Rew{\pattern/\head} \appctx{\slist}{v \esub{p}{u}} =t'$.
      The proof is by induction on the list $\slist$. We only show
      the case of the empty list as the other one is straightforward.
      The typing derivation $\Phi$ is necessarily of the form
      \[ \begin{prooftree}
          \begin{prooftree}
            \seq{\Gam_v}{v:\sigma} \sep
            \Gam_v|_{p} \pder p: \A
          \justifies{\seq{\cmin{\Gam_v}{\var{p}}}{\l p.v:\A \rew \sigma} }
        \end{prooftree} \sep
        \seq{\Gam_u}{u: \A}     
          \justifies{\seq{\cmin{\Gam_v}{\var{p} \inter  \Gam_u}}{(\l p. v)  u :\sigma}}
         \end{prooftree} \] 
       We then construct the following derivation $\Phi'$:
\[ \begin{prooftree}
          \seq{\Gam_v}{v:\sigma} \sep
          \Gam_v|_{p} \pder p: \A \sep
          \seq{\Gam_u}{u: \A}      
          \justifies{\seq{\cmin{\Gam_v}{\var{p} \inter  \Gam_u}}{v\esub{p}{u} :\sigma}}
        \end{prooftree} \]
     Moreover, $\sz{\Phi} = \sz{\Phi'} + 1$.

   \item $t= v\esub{x}{u}  \Rew{\pattern}
    v\isub{x}{u}=t'$. Then $\Phi$ has
     premisses
      $\Phi_v \tri \seq{\Gam_v; x:\A}{v:\sigma}$, $\Pi_x \tri x:\A \pder x:\A$ and $\Phi_u \tri
      \seq{\Gam_u}{u: \A}$, where $\Gam = \Gam_v \inter  \Gam_u$ and
      $\sz{\Phi} = \sz{\Phi_v} + \sz{\Phi_u} + \sz{\Pi_x} + 1$. Lemma~\ref{l:substitution-upper} then gives a derivation
      $\Phi'$ ending with $\seq{\Gam_v \inter \Gam_u}{v\isub{x}{u}:\sigma}$,
      where $\sz{\Phi'} = \sz{\Phi_v} + \sz{\Phi_u} - \size{\A} <
      \sz{\Phi_v} + \sz{\Phi_u} + \sz{\Pi_x} < \sz{\Phi}$, since $\size{\A} \geq
        0$ and $\sz{\Pi_x} = 1$.

      When  $t= v\esub{x}{u}   \Rew{\head}  v\isub{x}{u}=t'$, where $v
      \notRew{\head}$, the same results hold.
    
    \item
      $t= v\esub{\pair{p_1}{p_2}}{\appctx{\slist}{\pair{u_1}{u_2}}}
      \Rew{\pattern} \appctx{\slist}{
        v\esub{p_1}{u_1}\esub{p_2}{u_2}}=t'$. Let abbreviate
        $p = \pair{p_1}{p_2}$ and $u = \pair{u_1}{u_2}$. The typing
      derivation $\Phi$ is necessarily of the form
      \[ \begin{prooftree}
        \Phi_v \tri\seq{\Gam_v}{v:\sigma} \sep
        \Pi_{p} \tri \Gam_v|_{p} \pder \pair{p_1}{p_2}:\A \sep
        \Phi_u \tri \seq{\Gam_u}{\appctx{\slist}{\pair{u_1}{u_2}}:\A} 
         \justifies{\seq{\cmin{\Gam_v}{\var{\pair{p_1}{p_2}}} \inter  \Gam_u}{v\esub{\pair{p_1}{p_2}}{\appctx{\slist}{\pair{u_1}{u_2}}}}:\sigma }
       \end{prooftree} \]
     Moreover, $\A = \mult{\prodt{\A_1}{\A_2}}$ and $\sz{\Phi}
     = \sz{\Phi_v} + \sz{\Pi_{p}} + \sz{\Phi_u} + 1$.
     
    Then $\Pi_p$ is of the form: 
     \[ \begin{prooftree}
          \Pi_{p_1} \tri \Gam_v|_{p_1} \pder p_1:\A_1 \sep
          \Pi_{p_2} \tri \Gam_v|_{p_2} \pder p_2:\A_2 \sep p_1 \# p_2
          \justifies{\Gam_v|_{p} \pder \pair{p_1}{p_2}:\mult{\prodt{\A_1}{\A_2}}}
        \end{prooftree} \]
      and $\sz{\Pi_p} = \sz{\Pi_{p_1}} + \sz{\Pi_{p_2}} + 1$
      
      The proof is
       then by induction on the list $\slist$. 

    \begin{itemize}
      \item For
       $\slist = \Box$ we have $\Phi_u$ of the form:
 \[ \begin{prooftree}\begin{prooftree}
           \Phi_{u_1}
          \tri \seq{\Gam_{u_1}}{u_1: \A_1}   \sep 
          \Phi_{u_2}
          \tri \seq{\Gam_{u_2}}{u_2: \A_2}
          \justifies{\seq{\Gam_u}{\pair{u_1}{u_2}:\prodt{\A_1}{\A_2}}}
            \end{prooftree}
          \justifies{ \seq{\Gam_u}{\pair{u_1}{u_2}:\A}}
        \end{prooftree} \]
      where $\Gam_u = \Gam_{u_1} \inter \Gam_{u_2}$ and
      $\sz{\Phi_u} = \sz{\Phi_{u_1}} + \sz{\Phi_{u_2}} + 1$.
  We first construct the following derivation:
     \[  \begin{prooftree}
           \Phi_v \tri \seq{\Gam_v}{v:\sigma} \sep
           \Pi_{p_1} \tri \Gam_v|_{p_1} \pder p_1:\A_1 \sep
           \Phi_{u_1} \tri \seq{\Gam_{u_1}}{u_1: \A_1} 
          \justifies{\seq{\cmin{\Gam_v}{\var{p_1}} \inter \Gam_{u_1}}{v\esub{p_1}{u_1}: \sigma}}
        \end{prooftree} \]

      By using the relevance property and $\alpha$-conversion to assume freshness of bound variables,
      we can construct a derivation $\Phi'$ with conclusion
      $\seq{
        \cmin{\cmin{\Gam_v}{\var{p_1}}}{\var{p_2}} \inter \Gam_{u_1} \inter \Gam_{u_2}}{v\esub{p_1}{u_1}\esub{p_2}{u_2}:
        \sigma}$. Note that $\cmin{\Gam_v}{\var{\pair{p_1}{p_2}}}
      = \cmin{\cmin{\Gam_v}{\var{p_1}}}{\var{p_2}}$ since $p_1
        \# p_2$. Moreover, $\sz{\Phi'} = \sz{\Phi_v} +
        \sz{\Pi_{p}} + \sz{\Phi_u} < \sz{\Phi}$.

      \item Let $\slist = \slist'\esub{q}{s}$. Then $\Phi_u$ is necessarily
        of the following form:
        \[ \begin{prooftree}
            \begin{prooftree}
              \Phi_{\slist'} \tri \seq{\Del_u}{\appctx{\slist'}{\pair{u_1}{u_2}}:\prodt{\A_1}{\A_2} } \sep
              \Pi_q \tri \Seq{\Del_u|_{q}}{q:\B} \sep
             \Phi_s \tri  \seq{\Del_s}{s: \B}
              \justifies{ \seq{\Gam_u}{\appctx{\slist'}{\pair{u_1}{u_2}}\esub{q}{s}:\prodt{\A_1}{\A_2}}}
             \end{prooftree}
            \justifies{\seq{\Gam_u}{\appctx{\slist'}{\pair{u_1}{u_2}}\esub{q}{s}:\A}}
            \end{prooftree} \]
          where $\Gam_u = \cmin{\Del_u}{\var{q}} \inter
          \Del_s$. 

          We will apply the \ih\ on the reduction step
          $v\esub{p}{\appctx{\slist'}{u}}
      \Rew{\pattern} \appctx{\slist'}{
        v\esub{p_1}{u_1}\esub{p_2}{u_2}}$, in particular we type
      the left-hand side term with the following derivation $\Psi_1$:
      \[ \begin{prooftree}
          \Phi_v \sep
          \Pi_p \sep
          \begin{prooftree}
            \Phi_{\slist'}
            \tri\seq{\Del_u}{\appctx{\slist'}{\pair{u_1}{u_2}}:\prodt{\A_1}{\A_2}}
            \justifies{\seq{\Del_u}{\appctx{\slist'}{\pair{u_1}{u_2}}:\A}} 
          \end{prooftree}
          \justifies{\seq{\cmin{\Gam_v}{\var{p} \inter \Del_u}}{v\esub{\pair{p_1}{p_2}}{\appctx{\slist'}{\pair{u_1}{u_2}}}}: \sig}
          \end{prooftree}\] 
        The \ih\ gives a derivation $\Psi_2 \tri
        \seq{\cmin{\Gam_v}{\var{p} \inter \Del_u}}{\appctx{\slist'}{v\esub{p_1}{u_1}\esub{p_2}{u_2}}}: \sig$ verifying $\sz{\Psi_2} < \sz{\Psi_1}$. Let $\Lambda = \cmin{\Gam_v}{\var{p} \inter \Del_u}$.
        We conclude with the following derivation
        $\Phi'$:
         \[ \begin{prooftree}
          \Psi_2 \sep
          \Pi_q \tri \Seq{\Del_u|_{q}}{q:\B} \sep
          \Phi_s \tri     \seq{\Del_s}{s: \B}
          \justifies{\seq{\cmin{\Lambda}{\var{q}} \inter \Del_s}{\appctx{\slist'}{v\esub{p_1}{u_1}\esub{p_2}{u_2}}\esub{q}{s}: \sig}}
          \end{prooftree}\] 
        Indeed, we first remark that $\Lambda|_{q} = \Del_u|_{q}$ holds by
        relevance and $\alpha$-conversion. Secondly,
        $\cmin{\Gam_v}{\var{p}} \inter \Gam_u =\cmin{\Gam_v}{\var{p}} \inter
            (\cmin{\Del_u}{\var{q}}) \inter \Del_s =
                      \cmin{\Lambda}{\var{q}} \inter \Del_s$
       also holds  by
       relevance and $\alpha$-conversion. Last, we have
       $\sz{\Phi'} = \sz{\Psi_2}+ \sz{\Pi_q} + \sz{\Phi_s} + 1
       < \sz{\Psi_1}+ \sz{\Pi_q} + \sz{\Phi_s} + 1 =
       \sz{\Phi_v} + \sz{\Pi_p} + \sz{\Phi_{\slist'}} + 1 + \sz{\Pi_q} + \sz{\Phi_s} + 1 = 
       \sz{\Phi_v} + \sz{\Pi_p} + \sz{\Phi_u} + 1 = \sz{\Phi}$.\\
     \end{itemize}
   When $t= v\esub{\pair{p_1}{p_2}}{\appctx{\slist}{\pair{u_1}{u_2}}}
      \Rew{\head} \appctx{\slist}{
        v\esub{p_1}{u_1}\esub{p_2}{u_2}}=t'$, where $v
    \notRew{\head}$, the same results hold.

    \item $t= \l p. u \Rew{\pattern/\head} \l p. u' =t'$, where
      $u \Rew{\pattern/\head} u'$. This case is straightforward by the \ih\
    \item $t= vu \Rew{\pattern} v'u =t'$, where
      $v \Rew{\pattern} v'$
      and  $t= vu \Rew{\head} v'u =t'$, where
      $v \Rew{\head} v'$ and $\isnotabs{v}$. This case is
      straightforward by the \ih\
      \item $t= vu \Rew{\pattern} vu' =t'$, where
      $u \Rew{\pattern} u'$. This case is
      straightforward by the \ih. In particular, when $v$ has type
      $\emm \rew \sig$, $u$ (and then $u'$) would be typed by a
      $(\many)$ rule with no premiss, holding $\sz{\Phi} = \sz{\Phi'}$.
    \item $t= v\esub{p}{u} \Rew{\pattern/\head} v'\esub{p}{u} =t'$, where
      $v \Rew{\pattern/\head} v'$. This case is straightforward by the \ih\
      
    \item $t= v\esub{p}{u} \Rew{\pattern} v\esub{p}{u'} =t'$, where
      $u \Rew{\pattern} u'$. The proof also holds here by the \ih.
      In particular, when $p=x$ and $x \notin \fv{v}$, then
      by relevance we have $x$ of type $\emul$ as well as $u$ of type $\emul$.
      This means that $u$,$u'$ are typed by a $(\many)$ rule with no premiss,
      and in that case we get $\sz{\Phi} = \sz{\Phi'}$.

      Now, let us consider $t= v\esub{p}{u} \Rew{\head} v\esub{p}{u'} =t'$, where
      $v \notRew{\head}$ and $p\neq x$ and $u \Rew{\head} u'$.
      By construction there are typing subderivations 
      $\Phi_v \tri \seq{\Gam_v}{v:\sig}$,
      $ \Pi_p \tri \Gam_v|_{p} \pder p:\A$ and
      $\Phi_u \tri \seq{\Gam_u}{u:\A}$
      such that $\Gam = \cmin{\Gam_v}{\var{p}} \inter \Gam_u$.
      Since $p$ is not a variable then $\Pi_p$ ends with
      rule $\trpairpat$. In which case $\A$ contains
      exactly one type, let us say $\A = \mult{\sig_u}$. Then $\Phi_u$ has the following form
      \[ \infer{\seq{\Gam_u}{u:\sig_u}}
               {\Phi_u \tri \seq{\Gam_u}{u:\mult{\sig_u}}}\]
             The \ih\ applied to the premise of $\Phi_u$ gives
             a derivation 
             $\seq{\Gam_u}{u':\sig_u}$
             and having the expected size relation. 
             To conclude we build a type derivation  $\Phi'$ for
             $v\esub{p}{u'}$ having the expected size relation. 
\end{itemize}

    \item By induction on $\Rew{\pattern}$ (resp. $\Rew{\head}$) and (Anti-Substitution) Lemma~\ref{l:anti-substitution-upper} below. 
\begin{itemize}
\item $t'=\appctx{\slist}{\l p. v}u \Rew{\pattern/\head} \appctx{\slist}{v \esub{p}{u}}=t$. The proof is by induction on the list $\slist$. We consider the case $\slist = \Box$, since the other case follows straightforward by \ih. 
      The typing derivation $\Phi$ is necessarily of the form:
\[ \begin{prooftree}
          \seq{\Gam_v}{v:\sigma} \sep
          \Gam_v|_{p} \pder p: \A \sep
          \seq{\Gam_u}{u: \A}      
          \justifies{\seq{\cmin{\Gam_v}{\var{p} \inter  \Gam_u}}{v\esub{p}{u} :\sigma}}
        \end{prooftree} \]
       We then construct the following derivation $\Phi'$:
      \[ \begin{prooftree}
          \begin{prooftree}
            \seq{\Gam_v}{v:\sigma} \sep
            \Gam_v|_{p} \pder p: \A
          \justifies{\seq{\cmin{\Gam_v}{\var{p}}}{\l p.v:\A \rew \sigma} }
        \end{prooftree} \sep
        \seq{\Gam_u}{u: \A}     
          \justifies{\seq{\cmin{\Gam_v}{\var{p} \inter  \Gam_u}}{(\l p. v)  u :\sigma}}
         \end{prooftree} \] 
      Moreover, $\sz{\Phi'} = \sz{\Phi} + 1$.
   \item $t'=v\esub{x}{u}  \Rew{\pattern} v\isub{x}{u}=t$. Then by Lemma~\ref{l:anti-substitution-upper}, there exist derivations $\Phi_v, \Phi_u$, contexts $\Gam_v,\Gam_u$ and a multi-type $\A$, such that $\Phi_v \tri \seq{\Gam_v; x:\A}{v:\sigma}$, $\Phi_u \tri \seq{\Gam_u}{u: \A}$, $\Gam = \Gam_v \inter  \Gam_u$, and $\sz{\Phi} = \sz{\Phi_v} + \sz{\Phi_u}  - \size{\A}$. Furthermore, one has $\Pi_x \tri x:\A \pder x: \A$
 Then one can construct the following derivation $\Phi'$.
 \[ \begin{prooftree}
            \seq{\Gam_v;x:\A}{v:\sigma} \sep
           x:\A \pder x: \A  \sep
        \seq{\Gam_u}{u: \A}     
          \justifies{\seq{\Gam_v \inter \Gam_u}{  v\esub{x}{u}  :\sigma}}
 \end{prooftree} \]
 Furthermore, $\sz{\Phi'} = \sz{\Phi_v} + \sz{\Pi_x} + \sz{\Phi_u} > \sz{\Phi_v} + \sz{\Phi_u} - \size{A}$, since $\size{\A} \geq 0$ and $\sz{\Pi_x} = 1$.
The same result holds for $t=v\esub{x}{u}   \Rew{\head}
v\isub{x}{u}=t'$, where $v\notRew{\head}$.

\item  $t'=v\esub{\pair{p_1}{p_2}}{\appctx{\slist}{\pair{u_1}{u_2}}}
  \Rew{\pattern} \appctx{\slist}{v\esub{p_1}{u_1}\esub{p_2}{u_2}}
  =t$. Let abbreviate $p = \pair{p_1}{p_2}$ and $u = \pair{u_1}{u_2}$. 
  The proof is by induction on the list $\slist$.
  \begin{itemize}
    \item $\slist = \Box$, then the typing derivation $\Phi$ is necessarily of the form:
{\small 
\[ \begin{prooftree}
          \begin{prooftree}
            \seq{\Gam_v}{v:\sigma} \ \
            \Gam_v|_{p_1} \pder p_1: \A_1 \ \ \seq{\Gam_1}{u_1:\A_1}  
          \justifies{\seq{(\cmin{\Gam_v}{\var{p_1}})\inter \Gam_1}{v\esub{p_1}{u_1} :\sigma} }
        \end{prooftree} ((\cmin{\Gam_v}{\var{p_1}})\inter \Gam_1)|_{p_2} \pder p_2:\A_2 \ \  \seq{\Gam_2}{u_2: \A_2}     
          \justifies{\seq{((\cmin{(\cmin{\Gam_v}{\var{p_1}}) \inter  \Gam_1)}{\var{p_2}) \inter \Gam_2}}{v\esub{p_1}{u_1}\esub{p_2}{u_2} :\sigma}}
         \end{prooftree} \]} 
 $\Gam = (\cmin{((\cmin{\Gam_v}{\var{p_1}}) \inter  \Gam_1)}{\var{p_2}}) \inter
 \Gam_2$. Moreover, $(\cmin{((\cmin{\Gam_v}{\var{p_1}}) \inter  \Gam_1)}{\var{p_2}}) = (\cmin{(\cmin{\Gam_v}{\var{p_1}})}{\var{p_2}} \inter  \cmin{\Gam_1}{\var{p_2}} )$, where $\cmin{(\cmin{\Gam_v}{\var{p_1}})}{\var{p_2}} = \cmin{\Gam_v}{\var{p}}$ and $\cmin{\Gam_1}{\var{p_2}} =_{Lem. \ref{l:relevance+clash-free}:\ref{relevance}} \Gam_1$. Similarly, $ ((\cmin{\Gam_v}{\var{p_1}}) \inter  \Gam_1)|_{p_2} =_{Lem. \ref{l:relevance+clash-free}:\ref{relevance}}
 (\cmin{\Gam_v}{\var{p_1}})|_{p_2}$ and, by linearity of patterns, $
 (\cmin{\Gam_v}{\var{p_1}})|_{p_2} =\Gam_v|_{p_2}$. Hence, we
   conclude with the following derivation $\Phi'$:
 
 \[
 \begin{prooftree}
   \seq{\Gam_v}{v:\sigma} \sep
   \begin{prooftree} \Gam_v|_{p_1} \pder p_1:\A_1 \sep \Gam_v|_{p_2} \pder p_2:\A_2
     \justifies {\Gam_v|_{p} \pder p :
       \mult{\prodt{\A_1}{\A_2}} }
   \end{prooftree}
       \sep
   \begin{prooftree}
     \begin{prooftree}
     \seq{\Gam_1}{u_1:\A_1} \sep \seq{\Gam_2}{u_2:\A_2}
     \justifies{\seq{\Gam_1 \inter \Gam_2}{u :\prodt{\A_1}{\A_2}}}  
     \end{prooftree}
      \justifies{\seq{\Gam_1 \inter \Gam_2}{u : \mult{\prodt{\A_1}{\A_2}}}}  
   \end{prooftree}
   \justifies{\seq{(\cmin{\Gam_v}{\var{p}}) \inter (\Gam_1 \inter \Gam_2)}{v\esub{p}{u}:\sig}}
 \end{prooftree}
 \]
 Furthermore, $\sz{\Phi} = \sz{\Phi_v} +  \sz{\Pi_{p_1}} + \sz{\Phi_{u_1}} + 1 + \sz{\Pi_{p_2}} + \sz{\Phi_{u_2}} +1 $ and   $\sz{\Phi'} = \sz{\Phi_v} + \underbrace{\sz{\Pi_{p_1}} + \sz{\Pi_{p_2}} + 1}_{\sz{\Pi_{p}}} + \underbrace{\sz{\Phi_{u_1}} + \sz{\Phi_{u_2}} + 1}_{\sz{\Phi_{u}}} +1$.
 
\item If $\slist = \slist'\esub{q}{s}$, then $t'=v\esub{\pair{p_1}{p_2}}{\appctx{\slist'\esub{q}{s}}{\pair{u_1}{u_2}}} \Rew{\pattern} \appctx{\slist'\esub{q}{s}}{v\esub{p_1}{u_1}\esub{p_2}{u_2}}= \appctx{\slist'}{v\esub{p_1}{u_1}\esub{p_2}{u_2}}\esub{q}{s} = t$, and $\Phi$ is of the form: 
  \[
  \begin{prooftree}
    \Phi_{\slist'} \tri \seq{\Gam_{\slist'}}{\appctx{\slist'}{v\esub{p_1}{u_1}\esub{p_2}{u_2}}: \sig} \sep \Pi_q \tri \Gam_{\slist'}|_{q} \pder q :\A \sep \Phi_s \tri \seq{\Gam_s}{s:\A}  
    \justifies{\seq{\cmin{\Gam_{\slist'}}{\var{q}} \inter \Gam_s}{\appctx{\slist'}{v\esub{p_1}{u_1}\esub{p_2}{u_2}}\esub{q}{s} : \sig} }
  \end{prooftree}
  \]
  From $v\esub{\pair{p_1}{p_2}}{\appctx{\slist'}{\pair{u_1}{u_2}}}
  \Rew{\pattern} \appctx{\slist'}{v\esub{p_1}{u_1}\esub{p_2}{u_2}}$
  and $\Phi_{\slist'}$ by the \ih one gets $\Phi'_{\slist'} \tri
  \seq{\Gam_{\slist'}}{v\esub{\pair{p_1}{p_2}}{\appctx{\slist'}{\pair{u_1}{u_2}}}:
    \sig}$ with $\sz{\Phi'_{\slist'}} > \sz{\Phi_{\slist'}}$. Furthermore $\Phi'_{\slist'}$ is necessarily of the form:
  \[
  \begin{prooftree}
    \Phi_v \tri \seq{\Gam_v}{v:\sig} \sep \Pi_{p} \tri \Gam_v|_{p} \pder p : \mult{\prodt{\A_1}{\A_2}}\sep
    \begin{prooftree}
      \Phi_u \tri \seq{\Gam_u}{\appctx{\slist'}{u}: \prodt{\A_1}{\A_2}}
      \justifies{ \seq{\Gam_u}{\appctx{\slist'}{u}: \mult{\prodt{\A_1}{\A_2}}}}
    \end{prooftree}
    \justifies{\seq{\cmin{\Gam_{v}}{\var{p}} \inter \Gam_u}{v\esub{p}{\appctx{\slist'}{u}}: \sig}}
  \end{prooftree}
  \]
  Then one can construct the following derivation $\Phi'_u$:

  \[
  \begin{prooftree}
    \begin{prooftree}
      \seq{\Gam_u}{\appctx{\slist'}{u}: \prodt{\A_1}{\A_2}} \sep
      \Gam_u|_{q} \pder q: \A \sep \Phi_s
      \justifies{
        \seq{\cmin{\Gam_u}{\var{q}} \inter \Gam_s}{\appctx{\slist'\esub{q}{s}}{u} : \prodt{\A_1}{\A_2}}}
    \end{prooftree}
      \justifies{\seq{\cmin{\Gam_u}{\var{q}} \inter \Gam_s}{\appctx{\slist'\esub{q}{s}}{u} : \mult{\prodt{\A_1}{\A_2}}}}
    \end{prooftree}
  \]
  From which we build $\Phi'$:
  \[
  \begin{prooftree}
    \seq{\Gam_v}{v:\sig} \sep \Gam_v|_{p} \pder p : \mult{\prodt{\A_1}{\A_2}}\sep
    \seq{\cmin{\Gam_u}{\var{q}} \inter \Gam_s }{\appctx{\slist'\esub{q}{s}}{u} : \mult{\prodt{\A_1}{\A_2}}} 
    \justifies{\seq{\cmin{\Gam_v}{\var{p}} \inter \cmin{\Gam_u}{\var{q}} \inter \Gam_s }{v\esub{p}{\appctx{\slist'\esub{q}{s}}{u}}:\sig}}
  \end{prooftree}
  \]
  With $\cmin{\Gam_v}{\var{p}} \inter \cmin{\Gam_u}{\var{q}} \inter \Gam_s = \cmin{(\cmin{\Gam_v}{\var{p}} \inter \Gam_u)}{\var{q}} \inter \Gam_s = \Gam$.
  Furthermore $\sz{\Phi} = \sz{\Phi_{\slist'}} + \sz{\Pi_q} + \sz{\Phi_s} +1 <  \sz{\Phi'_{\slist'}} + \sz{\Pi_q} + \sz{\Phi_s} +1 = \sz{\Phi_v} + \sz{\Pi_{p}}  + \underbrace{\sz{\Phi_u}+ \sz{\Pi_q} + \sz{\Phi_s} + 1}_{\sz{\Phi'_u}}  +1 = \sz{\Phi'}$.
  \end{itemize}
The same result holds for $t'=v\esub{\pair{p_1}{p_2}}{\appctx{\slist}{\pair{u_1}{u_2}}} \Rew{\pattern} \appctx{\slist}{v\esub{p_1}{u_1}\esub{p_2}{u_2}}=t$, where $v\notRew{\head}$.
\item $t'= \l p. u' \Rew{\pattern/\head} \l p. u =t$, where
      $u' \Rew{\pattern/\head} u$. This case is straightforward by the \ih\
    \item $t'= v'u \Rew{\pattern} vu =t$, where
      $v' \Rew{\pattern} v$ and $t'= v'u \Rew{\head} vu =t$, where
      $v' \Rew{\head} v$ and $\isnotabs{v'}$. This case is
      straightforward by the \ih\
    \item $t'= vu' \Rew{\pattern} vu =t$, where
      $u' \Rew{\pattern} u$.  This case is
      straightforward by the \ih\  In particular, when $v$ has type
      $\emm \rew \sig$, $u$ (and then $u'$) would be typed by a
      $(\many)$ rule with no premiss, holding $\sz{\Phi} = \sz{\Phi'}$. 
    \item $t'= v'\esub{p}{u} \Rew{\pattern/\head} v\esub{p}{u} =t$, where
      $v' \Rew{\pattern/\head} v$. This case is straightforward by the \ih\
    \item $t'= v\esub{p}{u'} \Rew{\pattern} v\esub{p}{u} =t$, where
      $u' \Rew{\pattern} u$. The proof also holds here by the \ih.
      In particular, when $p=x$ and $x \notin \fv{v}$, then
      by relevance we have $x$ of type $\emul$ as well as $u$ of type $\emul$.
      This means that $u,u'$ are typed by a $(\many)$ rule with no premiss,
      and in that case we get $\sz{\Phi} = \sz{\Phi'}$.
      
      Now, let us consider $t'= v\esub{p}{u'} \Rew{\head}
      v\esub{p}{u} =t$, where $v \notRew{\head}$ and
       $p\neq x$ and $u' \Rew{\head} u$. By construction there are
       subderivations $\Phi_v \tri \seq{\Gam_v}{v:\sig}$,
       $\Pi_p \tri \Gam_v|_{p} \pder p:\A$ and
       $\Phi_{u} \tri \seq{\Gam_{u}}{u:\A}$ for some multi-set
       $\A$ and $\Gam = (\cmin{\Gam_v}{\var{p}}) \inter \Gam_{u}$.
       Since $p$ is not a variable then $ \Pi_p$ ends with
       rule $(\trpairpat)$, in which case $\A$ contains
      only one type, let us say $\A = \mult{\sig_{u}}$. Then $\Phi_{u}$
      has the following form:
      \[ \Phi_{u} \tri \inferrule{\seq{\Gam_{u}}{u:\sig_{u}}}
         {\seq{\Gam_{u}}{u':\mult{\sig_{u}}}}\]
         The \ih\ applied to the premise of $\Phi_{u}$ gives a derivation 
         $ \seq{\Gam_{u}}{u':\sig_{u}}$.
         Therefore, we construct the following derivation $\Phi'$:
         \[
         \begin{prooftree}
           \seq{\Gam_v}{v:\sig} \sep \Gam_v|_{p} \pder p : \mult{\sig_{u}} \sep \seq{\Gam_{u}}{u':\sig_{u}}
           \justifies{ \seq{\cmin{\Gam_v}{\var{p}} \inter \Gam_{u}}{ v\esub{p}{u'}}: \sig}
         \end{prooftree}
         \]
         Furthermore, $\sz{\Phi} = \sz{\Phi_v} + \sz{\Pi_p} + \sz{\Phi_{u}} +1$  and $\sz{\Phi'} = \sz{\Phi_v} + \sz{\Pi_p} + \sz{\Phi_{u'}} +1$. Since by the \ih\ $\sz{\Phi_{u'}} > \sz{\Phi_{u}}$, then  $\sz{\Phi'} >  \sz{\Phi}$.
             
\end{itemize}

  \end{enumerate}
\end{proof}

\begin{lemma}[Substitution for System \Upper]
  \label{l:substitution-upper}
  If $\Phi_t \tri  \seq{\Gam;x:\A}{t:\sig}$, and
  $\Phi_u \tri  \seq{\Del}{u:\A}$, then there
  exists $\Phi_{t\isub{x}{u}} \tri  \seq{\Gam \inter \Del}{t\isub{x}{u}:\sig}$
  such that $\sz{\Phi_{t\isub{x}{u}}} = \sz{\Phi_t} + \sz{\Phi_u } - \size{\A}$.
\end{lemma}
\begin{proof} We generalise the statement as follows:
  Let $\Phi_u \tri  \seq{\Del}{u:\A}$.
  \begin{itemize}
  \item If $\Phi_t \tri  \seq{\Gam; x:\A}{t:\sig}$, then there
    exists $\Phi_{t\isub{x}{u}} \tri  \seq{\Gam \inter \Del}{t\isub{x}{u}:\sig}$.
  \item If
    $\Phi_t \tri \seq{\Gam; x:\A}{t:\B}$, then
    there exists
    $\Phi_{t\isub{x}{u}} \tri
    \seq{\Gam \inter
      \Del}{t\isub{x}{u}:\B}$.
  \end{itemize}
 In both cases $\sz{\Phi_{t\isub{x}{u}}} = \sz{\Phi_t} + \sz{\Phi_u } - \size{\A}$.
  
The proof then follows by induction on $\Phi_t$.
\begin{itemize}
  \item If $\Phi_t$ is $(\ax)$, then we consider two cases:
    \begin{itemize}
    \item $t=x$: then $\Phi_x \tri  \seq{x:\mult{\sig}}{x:\sig}$ and $\Phi_u \tri  \seq{\Del}{u:\mult{\sig}}$, which is a consequence of $\seq{\Del}{u:\sig}$. Then $x\isub{x}{u} = u$, and we trivially obtain  $\Phi_{t\isub{x}{u}} \tri  \seq{\Del}{u:\sig}$.
     We have $\sz{\Phi_{t\isub{x}{u}}} = 1 + \sz{\Phi_u } - 1$ as expected.
    \item $t=y$: then $\Phi_y \tri  \seq{y:\mult{\sig}; x:\emul}{y:\sig}$ and   $\Phi_u \tri \seq{\es}{u:\emul}$ by the $(\many)$ rule. Then $y\isub{x}{u} = y$, and we trivially obtain  $\Phi_{t\isub{x}{u}} \tri  \seq{y:\mult{\sig}}{y:\sig}$. We have $\sz{\Phi_{t\isub{x}{u}}} = 1 + 0 - 0$ as expected.
    \end{itemize}
  \item If $\Phi_t$ ends with $(\many)$, then it has premisses of the form
    $(\Phi^i_t \tri \seq{\Gam^i; x:\A^i}{t:\sig^i})_{\iI}$,
    where $\Gam = \inter_{\iI} \Gam^i$, $\A = \inter_{\iI} \A^i$ and $\B = \mult{\sig^i}_{\iI}$.
    The derivation $\Phi_u$ can also be decomposed into subderivations
    $(\Phi^i_u \tri  \seq{\Del^i}{u:\A^i})_{\iI}$ where $\Del = \inter_{\iI} \Del^i$. The \ih\
    gives $(\Phi^i_{t\isub{x}{u}} \tri \seq{\Gam^i \inter \Del^i}{t\isub{x}{u}:\sig^i})_{\iI}$. Then we apply rule $(\many)$ to get
      $\Phi_{t\isub{x}{u}}\tri \seq{\Gam \inter \Del}{t\isub{x}{u}:\B} $. 
   The statement about $\sz{\_}$ works as expected by the \ih. 
  \item If $\Phi_t$ ends with $(\introarrow)$, so that $t=\lambda p.t'$ then, without loss of generality, one can always assume that $(\fv{u} \cup \{x\}) \cap \var{p} = \es$. The result will follow easily by induction and relevance of the typing system. The statement about $\sz{\_}$ works as expected by the \ih. 
  \item If $\Phi_t$ ends with $(\app)$, so that $t=t'u'$,  then $\Phi_{t'u'}$ is of the form
    \[\inferrule{\seq{\Phi_{t'} \tri \Gam_{t'}; x:\A_{t'} }{t': \B \rew \sig} \\ \Phi_{u'} \tri \seq{\Gam_{u'}; x:\A_{u'} }{u': \B}}{  \seq{\Gam_{t'} \inter \Gam_{u'}; x:\A_{t'} \inter \A_{u'}}{t'u':\sig}}\] Also, $\Phi_u \tri  \amuju{b_u}{e_u}{m_u}{f_u}{\Del}{u:\A}$ is a consequence of $(\amuju{b_u^k}{e_u^k}{m_u^k}{f_u^k}{\Delk}{u: \sigk})_{\kK}$, with $\A=\mult{\sigk}_\kK$ and $\Del = \inter_{\kK}\Delk$. Note that $\A = \A_{t'} \inter \A_{u'} = \mult{\sig_i}_{i\in K_{t'}} \inter \mult{\sig_i}_{i\in K_{u'}}$, with $K = K_{t'} \uplus K_{u'}$, from which one can obtain both $\seq{\Del_{t'}}{u:\A_{t'}}$ and $\seq{\Del_{u'}}{u:\A_{u'}}$, through the $(\many)$ rule.
By the \ih we then have $\seq{\Gam_{t'}\inter \Del_{t'}}{t'\isub{x}{u}: \B \rew \sig}$ and $\seq{\Gam_{u'}\inter \Del_{u'}}{u'\isub{x}{u}: \B}$.
Finally, $\seq{\Gam_{t'}\inter \Gam_{u'}\inter \Del_{t'} \inter \Del_{u'}}{(t'\isub{x}{u})(u'\isub{x}{u}): \sig}$ by the $\app$ rule.  The statement about $\sz{\_}$ works as expected by the \ih.
\item If $\Phi_t$ ends with $(\trpair)$ or $(\tnormalpair)$, so that $t=\pair{t'}{u'}$, then the result is obtained by induction following the same reasoning used in rule $\app$.
   The statement about $\sz{\_}$ works as expected by the \ih.
\item If $\Phi_t$ ends with  $(\trsub)$, so that $t=t'[p/u']$, then the result follows a similar reasoning to the case for application, since $t'[p/u']\isub{x}{u} = (t'\isub{x}{u})[p/u'\isub{x}{u}]$ and we can assume that $(\fv{u} \cup \{x\}) \cap \var{p} = \es$.  The statement about $\sz{\_}$ works as expected by the \ih.
  
\end{itemize}

\end{proof}

\begin{lemma}[Anti-Substitution for System \Upper]
  \label{l:anti-substitution-upper}
  Let $\Phi \tri  \seq{\Gam}{t\isub{x}{u}:\sig}$.
  Then, there exist derivations $\Phi_t$, $\Phi_u$,
  contexts $\Gam_t, \Gam_u$, and multi-type $\A$ such that
  $\Phi_t \tri  \seq{\Gam_t; x:\A}{t:\sig}$, 
  $\Phi_u \tri  \seq{\Gam_u}{u:\A}$ and 
   $\Gam = \Gam_t \inter \Gam_u$.
   Moreover,  $\sz{\Phi} = \sz{\Phi_t} + \sz{\Phi_u } - \size{\A}$.
 \end{lemma}

\begin{proof} As in the case of the substitution lemma, the proof follows by generalising the property for the two cases where the type derivation $\Phi$ assigns a type or a multi-set type:
\begin{itemize}
\item Let $\Phi \tri  \seq{\Gam}{t\isub{x}{u}:\sig}$. Then, there exist derivations $\Phi_t$, $\Phi_u$,
  contexts $\Gam_t, \Gam_u$, and multi-type $\A$ such that
  $\Phi_t \tri  \seq{\Gam_t; x:\A}{t:\sig}$, 
  $\Phi_u \tri  \seq{\Gam_u}{u:\A}$ and 
  $\Gam = \Gam_t \inter \Gam_u$.
\item  Let $\Phi \tri  \seq{\Gam}{t\isub{x}{u}:\B}$. Then, there exist derivations $\Phi_t$, $\Phi_u$,
  contexts $\Gam_t, \Gam_u$, and multi-type $\A$ such that
  $\Phi_t \tri  \seq{\Gam_t; x:\A}{t:\B}$, 
  $\Phi_u \tri  \seq{\Gam_u}{u:\A}$ and 
   $\Gam = \Gam_t \inter \Gam_u$.
\end{itemize}
In both cases $\sz{\Phi} = \sz{\Phi_t} + \sz{\Phi_u } - \size{\A}$ holds.

We will reason by induction on $\Phi$ and cases analysis on $t$. For all the rules (except $\many$), we will have the trivial case $t\isub{x}{u}$, where $t=x$, in which case $t\isub{x}{u} = u$, for which we have a derivation $\Phi \tri \seq{\Gam}{u:\sig}$. Therefore
$\Phi_t \tri \seq{x:\mult{\sig}}{x:\sig}$ and $\Phi_u \tri \seq{\Gam}{u:\mult{\sig}}$ is obtained from $\Phi$ using the $(\many)$ rule. We conclude since $\sz{\Phi} = 1  + \sz{\Phi_u } - 1$.
We now reason on the different cases assuming that $t\neq x$.
 \begin{itemize}
 \item If $\Phi$ is $(\ax)$ then $\Phi \tri \amuju{0}{0}{0}{0}{y:\mult{\sig}}{y:\sig}$ and, since $t\neq x$, $t = y \neq x$.
   Then we take $\A=\ems$, $\Phi_t \tri \seq{y:\mult{\sig};x:\ems}{y:\sig}$, and $\Phi_u \tri \seq{\es}{u:\ems}$ from rule $(\many)$.
 We conclude since $\sz{\Phi} = 1  + 0 - 0$.
 \item  If $\Phi$ ends with $(\many)$, then $\Phi \tri \seq{\inter_{\kK}\Gamk}{t\isub{x}{u}: \mult{\sigk}_{\kK}}$ follows from $\Phi^k \tri \seq{\Gamk}{t\isub{x}{u}: \sigk}$, for each $\kK$.
By the \ih, there exist $\Phi_t^k$, $\Phi_u^k$, contexts $\Gam_t^k$,  $\Gam_u^k$ and multi-type $\A_k$, such that  $\Phi_t^k \tri \seq{\Gam_t^k;x:\A_k}{t:\sigk}$, $\Phi_u^k \tri \seq{\Gam_u^k}{u:\A_k}$, $\Gam_k = \Gam_t^k \inter \Gam_u^k$. Taking $\A = \inter_{\kK}\A_k $ and using rule $\many$ we get $\seq{\inter_{\kK}\Gam_t^k;x:\A}{t: \mult{\sigk}_{\kK}}$. From the premisses of $\Phi_u^k$ for $\kK$, applying the $\many$ rule, we get  $\seq{\inter_{\kK}\Gam_u^k}{u: \A}$.
Note that $\Gam = \inter_{\kK}\Gamk = (\inter_{\kK}\Gam_t^k)\inter (\inter_{\kK}\Gam_u^k)$.
 The statement about $\sz{\_}$ works as expected by the \ih. 
\item If $\Phi$ ends with $(\introarrow)$, then $t=\l p.t'$, therefore $\Phi \tri \seq{\cmin{\Gam}{\var{p}}}{\l p.(t'\isub{x}{u}): \B \rew  \sig}$ follows from $\Phi' \tri \seq{\Gam}{t'\isub{x}{u} : \sig}$ and $\Pi_p \tri \Gam|_{p} \pder p:\B$. Note that, one can always assume that $\var{p} \cap \fv{u} = \es$ and $x\notin \var{p}$. By the \ih, $\Phi_{t'} \tri \seq{\Gam_{t'};x:\A}{t': \sig}$,  $\Phi_{u} \tri \seq{\Gam_u}{u: \A}$, with $\Gam = \Gam_{t'} \inter \Gam_u$. Then using $\introarrow$ we get $\Phi_t \tri \seq{\cmin{\Gam_{t'}}{\var{p}}; x:\A}{\l p.t': \B \rew  \sig}$. Note that $\cmin{(\Gam_{t'}; x:\A)}{\var{p}} = \cmin{\Gam_{t'}}{\var{p}}; x:\A$ and $\cmin{\Gam}{\var{p}} = (\cmin{\Gam_{t'}}{\var{p}}) \inter \Gam_u$.
 The statement about $\sz{\_}$ works as expected by the \ih.
\item If $\Phi$ ends with $(\app)$ then $t=t'u'$ and $\Phi$ is of the form
  \[\inferrule{\Phi' \tri \seq{\Gam}{t'\isub{x}{u}: \B \rew \sig} \\ \Phi'' \tri \seq{\Del}{u'\isub{x}{u} : \B}}{\seq{\Gam \inter \Del}{ t'\isub{x}{u}\,u'\isub{x}{u}: \sig}}\]
 By the \ih on the right premise there exist $\Phi_{t'},\Phi_{t'}^{u}$, contexts $\Gam_{t'}$, $\Gam_{t'}^u$ and multi-type $\A_{t'}$, such that $\Phi_{t'} \tri \seq{\Gam_{t'};x:\A_{t'} }{t': \B \rew \sig}$ and $\Phi_{t'}^{u} \tri \seq{\Gam_{t'}^u}{u: \A_{t'}}$ where $\Gam = \Gam_{t'} \inter \Gam_{t'}^u$. Also by the \ih on the left premise there exist $\Phi_{u'},\Phi_{u'}^{u}$, contexts $\Gam_{u'}$, $\Gam_{u'}^u$ and multi-type $\A_{u'}$, such that $\Phi_{u'} \tri \seq{\Del_{u'};x:\A_{u'} }{u': \B }$ and
 $\Phi_{u'}^{u} \tri \seq{\Del_{u'}^u}{u: \A_{u'}}$ where $\Del = \Del_{u'} \inter \Del_{u'}^u$.
 Taking $\A = \A_{t'} \inter\A_{u'}$, by the $\app$ rule, one gets $\Phi_{t'u'} \tri \seq{\Gam_{t'}\inter \Del_{u'};x:\A }{t': \B \rew \sig}$ and applying the $\many$ rule to the premisses of $\Phi_{t'}^{u}$ and $\Phi_{u'}^{u}$ one gets $\Phi_{u} \tri \seq{\Gam_{t'}^u \inter \Del_{u'}^u}{u: \A}$. Note that $\Gam \inter \Del = (\Gam_{t'} \inter \Gam_{t'}^u) \inter (\Del_{u'} \inter \Del_{u'}^u) = (\Gam_{t'} \inter \Del_{u'}) \inter (\Gam_{t'}^u \inter \Del_{u'}^u)$.
 The statement about $\sz{\_}$ works as expected by the \ih.
\item If $\Phi$ ends with $(\trpair)$, so that $t=\pair{t'}{u'}$, then the result follows by induction using the same reasoning as in rule $\app$. The statement about $\sz{\_}$ works as expected by the \ih.
\item If $\Phi_t$ ends with  $(\trsub)$, so that $t=t'[p/u']$, then the result follows a similar reasoning to the case for application, since $t'[p/u']\isub{x}{u} = (t'\isub{x}{u})[p/u'\isub{x}{u}]$ and we can assume that $(\fv{u} \cup \{x\}) \cap \var{p} = \es$ by $\alpha$-conversion.
 The statement about $\sz{\_}$ works as expected by the \ih. 
 \end{itemize}

\end{proof}

\section{Correctness for System \Exact}\label{a:correctness}

\gettoappendix{l:tight-spreading}
\begin{proof}
First note that, since $t \in \N$, then $t$ is not an abstraction nor a pair, therefore one cannot apply any of the rules $\{\introarrow$, $\ruletight, \trpair, \tnormalpair, \trempty\}$. We now examine the remaining rules.
  \begin{itemize}
  \item $t = x$. Then $\tderiv$ is an axiom $\amuju{0}{0}{0}{0}{x: \mult{\tight}}{x:\tight}$ so the property trivially holds.
  \item $t = uv$, with $u\in \N$. Then $\Phi$ has a (left)
      subderivation
    $\tderiv_{u} \rhd \amuju{b}{e}{m}{f}{\Gam_{u}}{u: \sig_u}$, and
    since $\Gam_{u}\subseteq \Gam$, then $\Gam_{u}$ is
    necessarily tight. Therefore, by the \ih, $\sig_u = \typendn$,
    from which follows that $\sig = \typendn$ by applying rule
    $(\appn)$. Note that one cannot apply rule $(\app)$ to type $uv$,
    since $t$ would have to be an arrow type, which contradicts the
    \ih\
  \item $t = u[p/v]$, with $u\in \N,v \in \N$. Then
    $\Phi$
    follows from
    $\Phi_{u} \tri \amuju{b_{u}}{e_{u}}{m_{u}}{f_{u}}{\Gam_{u}}{u:\sig}$,
    $\Phi_p \tri \Gam_{u}|_{p} \pder^{(e_p,m_p,f_p)} p:\A$ and
    $\Phi_v \tri \amuJu{b_v}{e_v}{m_v}{f_v}{\Del}{v:\A}$,
    where $\Gam =( \cmin{\Gam_{u}}{\var{p}}) \inter\Del$. Since
    $(\cmin{\Gam_{u}}{\var{p}})\inter \Del$ is tight, then $\Delta$ is
    tight. By the \ih\ on $v$ one gets
    $\amuJu{b_v}{e_v}{m_v}{f_v}{\Del}{v:\typendn}$
    so that $\amuJu{b_v}{e_v}{m_v}{f_v}{\Del}{v:\mult{\typendn}}$ follows from $\many$ and
    $\Gam_{u}|_{p} \pder^{(0,0,1)} p:\mult{\typendn}$ necessarily follows from
    rule $(\pattn)$. This implies $\Gam_{u}|_{p}$ is tight, therefore $\Gam_{u}$
    is tight. Since $u \in \N$ the \ih gives $\sig \in \tight$ as expected.
\end{itemize}

\end{proof}

\gettoappendix{l:M-tight}
\begin{proof}
By induction on $\tderiv \rhd \amuju{b}{e}{m}{f}{\Gam}{t: \sig}$,
where $\tderiv$ is tight (right-to-left implication),  and
by induction on $t
\in \M$ (left-to-right implication). The latter is
  presented below.
\begin{itemize}
  \item $t = \l p. u$, with $u \in \M$. Then $\Phi$ cannot end with
    rule $(\introarrow)$ because $\sig$ is tight. The last rule of
    $\Phi$ is necessarily $(\ruletight)$. The \ih\ then applies
    and gives $b=e=m=0$ and $f-1 = \size{u}$. We conclude
    since $f = \size{u} +1 = \size{t}$.
  \item $t = \pair{t_1}{t_2}$. Then $\Phi$
    necessarily ends with rule $(\tnormalpair)$
    and the counters are as required. 
  \item $t = u\esub{p}{v}$, with $u \in \M, v \in \N$.  Then $\Phi$
    ends with rule $(\trsub)$, so that $u$ (resp. $v$) is typable with
    some context $\Gam_{u}$ (resp. $\Gam_v$), where
    $\Gam = (\cmin{\Gam_{u}}{\var{p}}) \inter \Gam_v$.
    Let us consider the type $\A$ of $u$ in the premise of rule $(\trsub)$.
    Since $\Gam_v$ is tight and
    $v \in \N$, then Lem.~\ref{l:tight-spreading} guarantees that
    every type of $v$ in $\A$ is tight, and every counter typing $v$
    is of the form $(0,0,0,\size{v})$. This same multi-type $\A$ types
    the pattern $p$, so that there are in principle two cases:
    \begin{itemize}
    \item Either $p$ is a variable typable with rule $(\trvarpat)$,
      but then $t \notin \M$ since $t$ is still reducible. Contradiction. 
    \item Or $p$ is typable with rule $(\pattn)$,
      so that its  counter is $(0,0,1)$, its type is
    $\mult{\typendn}$ and its context is $\Gam_{u}|_{p}$  necessarily tight by
    definition of rule $(\pattn)$.
  \end{itemize}
  Since $\cmin{\Gam_{u}}{ \var{p}}$ is tight by hypothesis, then  the whole context
    $\Gam_{u}$ is tight. We can then apply the \ih\ to $u$ and obtain
    counters for $u$ of the form $b_{u}=e_{u}=m_{u} = 0$ and
    $f_{u} = \size{u}$. On the other side, since the type of $p$ is
    $\mult{\typendn}$ (rule $\pattn$), there is only one premise to type $v$, which
    is necessarily of the form $\amuju{0}{0}{0}{\size{v}}{\Del}{v:\typendn}$.
    We then conclude that the counters typing $u\esub{p}{v}$
    are $b=e=m=0$ and $f = f_{u} + f_v + 1 =
    \size{u} + \size{v} + 1 = \size{t}$, as required. 
    
  \item $t \in \N$. We have three different cases. 
    \begin{itemize}
    \item $t = x$. This case is straightforward. 
    \item $t = u v$, with $u \in \N$.
      Since $\Phi$ is tight, then $\Gam$ is tight and
      we can apply Lem.~\ref{l:tight-spreading}.
      Then $\Phi$ necessarily ends with rule $(\appn)$.
      The \ih\ then applies to the premise typing $u$,
      thus giving  counters $b=e=m=0$ and $f-1= \size{u}$.
      We conclude since $f = \size{u}+1 = \size{t}$.      
    \item $t = u\esub{p}{v}$, with $u \in \N, v \in \N$.
      This case is similar to the third case. 
    \end{itemize}
      
  \end{itemize}

\end{proof}

\gettoappendix{l:substitution}
\begin{proof} We generalise the statement as follows:
  Let $\Phi_u \tri  \amuju{b_u}{e_u}{m_u}{f_u}{\Del}{u:\A}$.
  \begin{itemize}
  \item If $\Phi_t \tri  \amuju{b_t}{e_t}{m_t}{f_t}{\Gam; x:\A}{t:\sig}$, then there
    exists \\ $\Phi_{t\isub{x}{u}} \tri  \amuju{b_t+b_u}{e_t+e_u}{m_t+m_u}{f_t+f_u}{\Gam \inter \Del}{t\isub{x}{u}:\sig}$.
  \item If
    $\Phi_t \tri \amuju{b_t}{e_t}{m_t}{f_t}{\Gam; x:\A}{t:\B}$, then
    there exists \\
    $\Phi_{t\isub{x}{u}} \tri
    \amuju{b_t+b_u}{e_t+e_u}{m_t+m_u}{f_t+f_u}{\Gam \inter
      \Del}{t\isub{x}{u}:\B}$.
  \end{itemize}
The proof then follows by induction on $\Phi_t$.
\begin{itemize}
  \item If $\Phi_t$ is $(\ax)$, then we consider two cases:
    \begin{itemize}
    \item $t=x$: then $\Phi_x \tri  \amuju{0}{0}{0}{0}{x:\mult{\sig'}}{x:\sig' }$ and $\Phi_u \tri  \amuju{b_u}{e_u}{m_u}{f_u}{\Del}{u:\mult{\sig'}}$, which is a consequence of $\amuju{b_u}{e_u}{m_u}{f_u}{\Del}{u:\sig'}$. Then $x\isub{x}{u} = u$, and we trivially obtain  $\Phi_{t\isub{x}{u}} \tri  \amuju{0+b_u}{0+e_u}{0+m_u}{0+f_u}{\Del}{u:\sig'}$.
    \item $t=y$: then $\Phi_y \tri  \amuju{0}{0}{0}{0}{y:\mult{\sig}; x:\mult{\ }}{y:\sig}$ and   $\Phi_u \tri \amuju{0}{0}{0}{0}{\es}{u:\mult{\ }}$. Then $y\isub{x}{u} = y$, and we trivially obtain  $\Phi_{t\isub{x}{u}} \tri  \amuju{0+0}{0+0}{0+0}{0+0}{y:\mult{\sig}}{y:\sig}$. 
    \end{itemize}
  \item If $\Phi_t$ ends with $(\many)$, then it has premisses of the form
    $(\Phi^i_t \tri \amuju{b^i_t}{e^i_t}{m^i_t}{f^i_t}{\Gam^i; x:\A^i}{t:\sig^i})_{\iI}$,
    where $\Gam = \inter_{\iI} \Gam^i$, $\A = \inter_{\iI} \A^i$,
    $b_t = +_{\iI} b^i_t$, $e_t = +_{\iI} e^i_t$, $m_t = +_{\iI} m^i_t$, $f_t = +_{\iI} f^i_t$ and $\B = \mult{\sig^i}_{\iI}$.
    The derivation $\Phi_u$ can also be decomposed into subderivations
    $(\Phi^i_u \tri  \amuju{b^i_u}{e^i_u}{m^i_u}{f^i_u}{\Del^i}{u:\A^i})_{\iI}$,
    where $b_u = +_{\iI} b^i_u$, $e_u = +_{\iI} e^i_u$, $m_u = +_{\iI} m^i_u$, $f_u = +_{\iI} f^i_u$, $\Del = \inter_{\iI} \Del^i$. The \ih\
    gives $(\Phi^i_{t\isub{x}{u}} \tri \amuju{b^i_t+b^i_u}{e^i_t+e^i_u}{m^i_t+m^i_u}{f^i_t+f^i_u}{\Gam^i \inter \Del^i}{t\isub{x}{u}:\sig^i})_{\iI}$. Then we apply rule $(\many)$ to get
      $\Phi_{t\isub{x}{u}}\tri \amuju{b_t+b_u}{e_t+e_u}{m_t+m_u}{f_t+f_u}{\Gam \inter \Del}{t\isub{x}{u}:\B} $. 
    
  \item If $\Phi_t$ ends with $(\introarrow)$ or $(\ruletight)$, so that $t=\lambda p.t'$: then, without loss of generality, one can always assume that $x \notin \var{p}$. The result will follow easily by induction.
  \item If $\Phi_t$ ends with $(\app)$, so that $t=t'u'$,  then $$\Phi_{t'u'} \tri  \amuju{b_{t'}+b_{u'}}{e_{t'}+e_{u'}}{m_{t'}+m_{u'}}{f_{t'}+f_{u'}}{\Gam_{t'} \inter \Gam_{u'}; x:\A_{t'} \inter \A_{u'}}{t'u':\sig},$$ which follows from $\amuju{b_{t'}}{e_{t'}}{m_{t'}}{f_{t'}}{\Gam_{t'}; x:\A_{t'} }{t': \B \rew \sig}$ and $\amuju{b_{u'}}{e_{u'}}{m_{u'}}{f_{u'}}{\Gam_{u'}; x:\A_{u'} }{u': \B}$. Also, $\Phi_u \tri  \amuju{b_u}{e_u}{m_u}{f_u}{\Del}{u:\A}$ is a consequence of $(\amuju{b_u^k}{e_u^k}{m_u^k}{f_u^k}{\Delk}{u: \sigk})_{\kK}$, with $\A=\mult{\sigk}_\kK$, $\Del = \inter_{\kK}\Delk$ and $b_u = +_{\kK}b_u^k$, $e_u = +_{\kK}e_u^k$, $m_u = +_{\kK}m_u^k$ and $f_u = +_{\kK}f_u^k$. Note that $\A = \A_{t'} \inter \A_{u'} = \mult{\sig_i}_{i\in K_{t'}} \inter \mult{\sig_i}_{i\in K_{u'}}$, with $K = K_{t'} \uplus K_{u'}$, from which one can obtain (using the $\many$ rule):

    \begin{itemize}
    \item $\amuju{B_{t'}}{E_{t'}}{M_{t'}}{F_{t'}}{\Del_{t'}}{u:\A_{t'}}$
    \item  $\amuju{B_{u'}}{E_{u'}}{M_{u'}}{F_{u'}}{\Del_{u'}}{u:\A_{u'}}$
    \end{itemize}
    where    $b_u = B_{t'} + B_{u'} = (+_{i\in K_{t'}}b_u^i)+(+_{i\in K_{u'}}b_u^i)$, $e_u = E_{t'} + E_{u'} = (+_{i\in K_{t'}}e_u^i)+(+_{i\in K_{u'}}e_u^i)$, $m_u = M_{t'} + M_{u'}= (+_{i\in K_{t'}}m_u^i)+(+_{i\in K_{u'}}m_u^i)$, $f_u = F_{t'} + F_{u'} =(+_{i\in K_{t'}}f_u^i)+(+_{i\in K_{u'}}f_u^i)$.
    By the \ih we have:
    $$\amuju{b_{t'}+B_{t'}}{e_{t'}+E_{t'}}{m_{t'}+M_{t'}}{f_{t'}+F_{t'}}{\Gam_{t'}\inter \Del_{t'}}{t'\isub{x}{u}: \B \rew \sig}$$
  $$\amuju{b_{u'}+B_{u'}}{e_{u'}+E_{u'}}{m_{u'}+M_{u'}}{f_{u'}+F_{u'}}{\Gam_{u'}\inter \Del_{u'}}{u'\isub{x}{u}: \B} $$
    Finally, applying the $\app$ rule we obtain:
    $$\amuju{b}{e}{m}{f}{\Gam_{t'}\inter \Gam_{u'}\inter \Del_{t'} \inter \Del_{u'}}{(t'\isub{x}{u})(u'\isub{x}{u}): \sig}$$
    with $b = b_{t'}+b_{u'}+b_u$, $e=e_{t'}+e_{u'}+e_u$, $m=m_{t'}+m_{u'}+m_u$ and $f=f_{t'}+f_{u'}+f_u$, as expected.
    \item If $\Phi_t$ ends with $(\appn)$, the result follows from a straightforward induction. 
  \item If $\Phi_t$ ends with  $(\trsub)$, so that $t=t'[p/u']$, then the result follows a similar reasoning to the case for application, since $t'[p/u']\isub{x}{u} = (t'\isub{x}{u})[p/u'\isub{x}{u}]$. 
  \item If $\Phi_t$ ends with $(\trpair)$ or $(\tnormalpair)$, so that $t=\pair{t'}{u'}$, then we have two cases. The case for $\tnormalpair$, follows from $\Phi_t$ being of the form $\amuju{0}{0}{0}{1}{x:\mult{\ }}{\pair{t'}{u'}:\typendm}$, which implies $\Phi_u \tri \amuju{0}{0}{0}{0}{\empty}{u:\mult{\ }}$. Therefore $\amuju{0+0}{0+0}{0+0}{1+0}{\emptyset}{\pair{t'\isub{x}{u}}{u'\isub{x}{u}}:\typendm}$ holds. The case for $\trpair$ follows by induction following the same reasoning used in rule $\app$.    
\end{itemize}

\end{proof}

\gettoappendix{l:subject-reduction} 

\begin{proof} By induction on $\Rew{\head}$.
\begin{itemize}
    \item $t= \appctx{\slist}{\l p. v} u \Rew{\head} \appctx{\slist}{v \esub{p}{u}} =t'$.
      The proof is by induction on the list context $\slist$. We only show
      the case of the empty list as the other one is straightforward.
      The typing derivation $\Phi$ is necessarily of the form
      \[ \begin{prooftree}
          \begin{prooftree}
            \amuju{b_v}{e_v}{m_v}{f_v}{\Gam_v}{v:\sigma} \sep
            \Gam_v|_{p} \pder^{(e_p,m_p,f_p)} p: \A
          \justifies{\amuju{b_v+1}{e_v+e_p}{m_v+m_p}{f_v+f_p}{\cmin{\Gam_v}{\var{p}}}{\l p.v:\A \rew \sigma} }
        \end{prooftree} \sep
        \amuju{b_u}{e_u}{m_u}{f_u}{\Gam_u}{u: \A}     
          \justifies{\amuju{b_v+1+b_u}{e_v+e_p+e_u}{m_v+m_p+m_u}{f_v+f_p+f_u}{\cmin{\Gam_v}{\var{p} \inter  \Gam_u}}{(\l p. v)  u :\sigma}}
         \end{prooftree} \] 
       We then construct the following derivation $\Phi'$:
\[ \begin{prooftree}
          \amuju{b_v}{e_v}{m_v}{f_v}{\Gam_v}{v:\sigma} \sep
          \Gam_v|_{p} \pder^{(e_p,m_p,f_p)} p: \A \sep
          \amuju{b_u}{e_u}{m_u}{f_u}{\Gam_u}{u: \A}      
          \justifies{\amuju{b_v+b_u}{e_v+e_p+e_u}{m_v+m_u+m_p}{f_v+f_p+f_u}{\cmin{\Gam_v}{\var{p} \inter  \Gam_u}}{v\esub{p}{u} :\sigma}}
        \end{prooftree} \]
     The counters are as expected because the first one has decremented by $1$.
      
    \item $t= v\esub{x}{u}   \Rew{\head}  v\isub{x}{u}=t'$, where
      $v \notRew{\head}$. Then $\Phi$ has two premisses
      $\amuju{b_v}{e_v}{m_v}{f_v}{\Gam_v; x:\A}{v:\sigma}$
      and $\amuju{b_u}{e_u}{m_u}{f_u}{\Gam_u}{u: \A}$, where $\Gam = \Gam_v \inter  \Gam_u$,
      $b=b_v+b_u$, $e=e_v+e_u+1$, $m= m_v+m_u+0$, and $f=f_v+f_u+0$.
      
      Lem.~\ref{l:substitution} then gives a derivation
      ending with $\amuju{b_v+b_u}{e_v+e_u}{m_v+m_u}{f_v+f_u}{\Gam_v \inter \Gam_u}{v\isub{x}{u}:\sigma}$. The context, type, and counters are as expected.
    
      \item $t= v\esub{\pair{p_1}{p_2}}{\appctx{\slist}{\pair{u_1}{u_2}}}
      \Rew{\head} \appctx{\slist}{
        v\esub{p_1}{u_1}\esub{p_2}{u_2}}=t'$, where $v
      \notRew{\head}$.

      Let $p = \pair{p_1}{p_2}$ and $u = \pair{u_1}{u_2}$.
      The typing derivation $\Phi$ is necessarily of the form
     
     \[ \begin{prooftree}
        \amuju{b_v}{e_v}{m_v}{f_v}{\Gam_v}{v:\sigma} \sep
        \Gam_v|_{p} \pder^{(e_p,m_p,f_p)} p:\A \sep
        \amuju{b_u}{e_u}{m_u}{f_u}{\Gam_u}{\appctx{\slist}{u}:\A} 
         \justifies{\amuju{b_v+b_u}{e_v+e_u+e_p}{m_v+m_u+m_p}{f_v+f_u+f_p}{\cmin{\Gam_v}{\var{p}} \inter \Gam_u}{v\esub{\pair{p_1}{p_2}}{\appctx{\slist}{u}}}:\sigma }
        \end{prooftree} \]
Then $\A = \mult{\prodt{\A_1}{\A_2}}$, for some multi-types
  $\A_1$ and $\A_2$, and so the pattern $\pair{p_1}{p_2}$ is typable as follows:
      \[ \begin{prooftree}
          \Gam_v|_{p_1} \pder^{(e_1,m_1,f_1)} p_1:\A_1 \sep
          \Gam_v|_{p_2} \pder^{(e_2,m_2,f_2)} p_2:\A_2 
          \justifies{\Gam_v|_{p} \pder^{(e_1+e_2,1+m_1+m_2,f_1+f_2)} \pair{p_1}{p_2}:\mult{\prodt{\A_1}{\A_2}}}
        \end{prooftree} \]
      where $e_p= e_1+e_2$, $m_p = 1+m_1+m_2$ and $f_p = f_1+f_2$.
      The proof then follows by induction on list $\slist$:

      \begin{itemize}
        \item For $\slist = \Box$ we have
        term $u$  typable as follows:

      \[ \begin{prooftree}
                    \infer{\amuju{b'_1}{e'_1}{m'_1}{f'_1}{\Gam_1}{u_1: \A_1 } \sep \amuju{b'_2}{e'_2}{m'_2}{f'_2}{\Gam_2}{u_2:\A_2 }}{\amuju{b_u}{e_u}{m_u}{f_u}{\Gam_u}{\pair{u_1}{u_2}:\prodt{\A_1}{\A_2}}}
           \justifies{\amuju{b_u}{e_u}{m_u}{f_u}{\Gam_u}{\pair{u_1}{u_2}:\A}}
       \end{prooftree} \]
   
      where $\Gam_u = \Gam_1 \inter
      \Gam_2$ and $(b_u,e_u,m_u,f_u) =
      (b'_1+b'_2,e'_1+e'_2,m'_1+m'_2,f'_1+f'_2)$.
     
      We first construct the following derivation:
     \[  \begin{prooftree}
            \amuju{b_v}{e_v}{m_v}{f_v}{\Gam_v}{v:\sigma} \sep
           \Gam_v|_{p_1} \pder^{(e_1,m_1,f_1)} p_1:\A_1 \sep
           \amuju{b'_1}{e'_1}{m'_1}{f'_1}{\Gam_1}{u_1: \A_1 } 
          \justifies{\amuju{b_v+b'_1}{e_v+e'_1+e_1}{m_v+m'_1+m_1}{f_v+f'_1+f_1}{\cmin{\Gam_v}{\var{p_1}}
                \inter \Gam_1}{v\esub{p_1}{u_1}: \sigma}}
        \end{prooftree} \]

      By using relevance and $\alpha$-conversion to assume freshness of bound variables,
      we can construct a derivation with conclusion
      $$\amuju{b_v+b_u}{e_v+ e_u + e_1+e_2}{m_v+m_u
        +m_1+m_2}{f_v+f_u +f_1+ f_2}{
        \cmin{\cmin{\Gam_v}{\var{p_1}}}{\var{p_2}} \inter
        \Gam_u}{v\esub{p_1}{u_1}\esub{p_2}{u_2}: \sigma}$$
   In order to conclude we remark the following facts:
      \begin{itemize}
      \item $\cmin{\Gam_v}{\var{\pair{p_1}{p_2}}} = \cmin{\cmin{\Gam_v}{\var{p_1}}}{\var{p_2}}$
      \item $b_v+b_u = b_v+b'_1 + b'_2$
      \item $e_v+e_u+e_p = e_v+e'_1 + e'_2+e_1+e_2$
      \item $m_v+m_u+m_p = m_v+m'_1 + m'_2 + 1 + m_1 + m_2$
      \item $f_v+f_u+f_p= f_v + f'_1 + f'_2 + f_1 + f_2$
      \end{itemize}

      Then the context, type and counters are as expected.

    \item Let $\slist = \slist'\esub{q}{s}$. Then $\Phi_u$ is necessarily
        of the following form:
        \[ \begin{prooftree}
            \begin{prooftree}
              \amuju{b'_u}{e'_u}{m'_u}{f'_u}{\Del_u}{\appctx{\slist'}{u}:\prodt{\A_1}{\A_2} } \sep
              \Del_u|_{q} \pder^{(e_q,m_q,f_q)} q:\B \sep
             \amuju{b_s}{e_s}{m_s}{f_s}{\Del_s}{s: \B}
              \justifies{ \amuju{b'_u+b_s}{e'_u+e_s+e_q}{m'_u+m_s+m_q}{f'_u+f_s+f_q}{\Gam_u}{\appctx{\slist'}{u}\esub{q}{s}:\prodt{\A_1}{\A_2}}}
             \end{prooftree}
            \justifies{\amuju{b'_u+b_s}{e'_u+e_s+e_q}{m'_u+m_s+m_q}{f'_u+f_s+f_q}{\Gam_u}{\appctx{\slist'}{u}\esub{q}{s}:\A}}
            \end{prooftree} \]
          where $\Gam_u = \cmin{\Del_u}{\var{q}} \inter
          \Del_s$, $b_u = b'_u+b_s$, $e_u = e'_u+e_s+e_q$, $m_u =
          m'_u+m_s+m_q$ and $f_u = f'_u+f_s+f_q$.

          We will apply the \ih\ on the reduction step
          $v\esub{p}{\appctx{\slist'}{u}}
      \Rew{\pattern} \appctx{\slist'}{
        v\esub{p_1}{u_1}\esub{p_2}{u_2}}$, in particular we type
      the left-hand side term with the following derivation $\Psi_1$:
      \[ \begin{prooftree}
          \amuju{b_v}{e_v}{m_v}{f_v}{\Gam_v}{v:\sigma} \sep
          \Gam_v|_{p} \pder^{(e_p,m_p,f_p)} p:\A  \sep
          \begin{prooftree}
            \amuju{b'_u}{e'_u}{m'_u}{f'_u}{\Del_u}{\appctx{\slist'}{u}:\prodt{\A_1}{\A_2}}
            \justifies{\amuju{b'_u}{e'_u}{m'_u}{f'_u}{\Del_u}{\appctx{\slist'}{u}:\A}} 
          \end{prooftree}
          \justifies{\amuju{b_v + b'_u}{e_v + e'_u + e_p}{m_v + m'_u +
              m_p}{f_v + f'_u + f_p}{\cmin{\Gam_v}{\var{p} \inter \Del_u}}{v\esub{p}{\appctx{\slist'}{u}}}: \sig}
          \end{prooftree}\] 
        where $b_1 = b_v + b'_u$, $e_1 = e_v + e'_u + e_p$, $m_1 = m_v + m'_u +
              m_p$ and $f_1 = f_v + f'_u + f_p$. The \ih\ gives a derivation $\Psi_2 \tri
        \amuju{b_2}{e_2}{m_2}{f_2}{\cmin{\Gam_v}{\var{p} \inter
            \Del_u}}{\appctx{\slist'}{v\esub{p_1}{u_1}\esub{p_2}{u_2}}}:
        \sig$ where $b_2 = b_1$, $e_2 = e_1$, $m_2 = m_1 - 1$ and $f_2
        = f_1$. Let $\Lambda = \cmin{\Gam_v}{\var{p} \inter \Del_u}$.
        We conclude with the following derivation
        $\Phi'$:
         \[ \begin{prooftree}
          \Psi_2 \sep
          \Del_u|_{q} \pder^{(e_q,m_q,f_q)} q:\B \sep
             \amuju{b_s}{e_s}{m_s}{f_s}{\Del_s}{s: \B}
          \justifies{\amuju{b_2 + b_s}{e_2 + e_s + e_q}{m_2 + m_s +
              m_q}{f_2 + f_s + f_q}{\cmin{\Lambda}{\var{q}} \inter \Del_s}{\appctx{\slist'}{v\esub{p_1}{u_1}\esub{p_2}{u_2}}\esub{q}{s}: \sig}}
          \end{prooftree}\] 
        Indeed, we first remark that $\Lambda|_{q} = \Del_u|_{q}$ holds by
        relevance and $\alpha$-conversion. Secondly,
        $\cmin{\Gam_v}{\var{p}} \inter \Gam_u =\cmin{\Gam_v}{\var{p}} \inter
            (\cmin{\Del_u}{\var{q}}) \inter \Del_s =
                      \cmin{\Lambda}{\var{q}} \inter \Del_s$
       also holds  by
       relevance and $\alpha$-conversion. Last, we conclude with the following remarks:
      \begin{itemize}
          \item $b' = b_2 + b_s = b_1 + b_s = b_v + b_u = b$
          \item $e' = e_2 + e_s + e_q = e_1 + e_s + e_q = e_v + e_u +
            e_p = e$
          \item $m' = m_2 + m_s + m_q = m_1 - 1 + m_s + m_q = m_v +
            m_u + m_p - 1 = m - 1$
          \item $f' = f_2 + f_s + f_q = f_1 + f_s + f_q = f_v + f_u +
            f_p = f$
      \end{itemize}  

     \end{itemize}

    \item $t= \l p. u \Rew{\head} \l p. u' =t'$, where
      $u \Rew{\head} u'$. This case is straightforward by the \ih\
    \item $t= vu \Rew{\head} v'u =t'$, where
      $v \Rew{\head} v'$ and $\isnotabs{v}$. This case is straightforward by the \ih\
    \item $t= v\esub{p}{u} \Rew{\head} v'\esub{p}{u} =t'$, where
      $v \Rew{\head} v'$. This case is straightforward by the \ih\
    \item $t= v\esub{p}{u} \Rew{\head} v\esub{p}{u'} =t'$, where
      $v \notRew{\head}$ and $p\neq x$ and $u \Rew{\head} u'$.
      By construction there are typing subderivations 
      $\Phi_v \tri \amuju{b_v}{e_v}{m_v}{f_v}{\Gam_v}{v:\sig}$,
      $ \Phi_p \tri \Gam_v|_{p} \pder^{(e_p,m_p,f_p)} p:\A$ and
      $\Phi_u \tri \amuJu{b_u}{e_u}{m_u}{f_u}{\Gam_u}{u:\A}$.
      Since $p$ is not a variable then $ \Phi_p$ ends with
      rule $\pattn$ or $\trpairpat$. In both cases $\A$ contains
      only one type, let us say $\A = \mult{\sig_u}$. Then $\Phi_u$ has the following form
      \[ \infer{\amuJu{b_u}{e_u}{m_u}{f_u}{\Gam_u}{u:\sig_u}}
               {\Phi_u \tri \amuJu{b_u}{e_u}{m_u}{f_u}{\Gam_u}{u:\mult{\sig_u}}}\]
             The \ih\ applied to the premise of $\Phi_u$ gives
             a derivation 
             $\amuJu{b'_u}{e'_u}{m'_u}{f_u}{\Gam_u}{u':\sig_u}$
             and having the expected counters. 
             To conclude we build a type derivation  $\Phi'$ for
             $v\esub{p}{u'}$ having the expected counters. 
\end{itemize}

\end{proof}

\section{Completeness for System \Exact}\label{a:completeness}

\gettoappendix{l:M-typable-tight} 
\begin{proof}
  We generalise the property to the two following statements, proved by
structural induction on $t\in \N$, $t \in \M$, respectively, using
relevance (Lem.~\ref{l:relevance+clash-free}:\ref{relevance}).
  \begin{itemize}
  \item If $t \in \N$, then there exists a tight derivation $\tderiv
    \rhd \amuju{0}{0}{0}{\size{t}}{\Gam}{t: \typendn}$:
   \begin{itemize}
    \item If $t = x$, then $\amuju{0}{0}{0}{0}{x{:}\mult{\typendn}}{x:
        \typendn}$ by $(\ax)$, where $|x| = 0$.

     \item If $t = u\,v$ where $u \in \N$, then $\size{t} = \size{u} +
       1$ and by \ih there is a tight derivation $\tderiv_u
       \rhd \amuju{0}{0}{0}{\size{u}}{\Gam_u}{u: \typendn}$. Then
       \[\inferrule*{\amuju{0}{0}{0}{\size{u}}{\Gam_u}{u: \typendn} \\
         }{\amuju{0}{0}{0}{\size{u} + 1}{\Gam_u}{u\,v: \typendn}}\]
       The result then holds for $\Gam := \Gam_u$.
       
       \item  If $t = u\esub{\pair{p_1}{p_2}}{v}$ where $u,v \in \N$,
         then $\size{t} = \size{u} + \size{v} +
       1$ and by \ih there are tight derivations $\tderiv_u
       \rhd \amuju{0}{0}{0}{\size{u}}{\Gam_u}{u: \typendn}$, $\tderiv_v
       \rhd \amuju{0}{0}{0}{\size{v}}{\Gam_v}{v: \typendn}$. Then, $\tderiv'_v
       \rhd \amuju{0}{0}{0}{\size{v}}{\Gam_v}{v: \mult{\typendn}}$ and
       \[ \inferrule*{\tderiv_u
           \\   {\Gam_u}|_{\pair{p_1}{p_2}} \pder^{(0,0,1)}  \pair{p_1}{p_2}: \mult{\typendn} \\ \tderiv'_v}{\amuju{0}{0}{0}{\size{u} + \size{v} +
       1}{(\cmin{\Gam_u}{\var{\pair{p_1}{p_2}}}) \inter
       \Gam_v}{u\esub{\pair{p_1}{p_2}}{v}: \typendn}}\]
    The result then holds for $\Gam := (\cmin{\Gam_u}{\var{\pair{p_1}{p_2}}}) \inter
       \Gam_v$, since by \ih\
       $\istight{\Gam_u}$ and $\istight{\Gam_v}$ thus $\istight{\Gam}$.
     \end{itemize}
     
  \item If $t \in \M$, then there exists a tight derivation $\tderiv
    \rhd \amuju{0}{0}{0}{\size{t}}{\Gam}{t: \tight}$. 
    \begin{itemize}
    \item If $t \in \N$ then by the previous item the result holds
      for $\tight := \typendn$.
      
    \item If $t = \pair{u}{v}$ then $\size{t} = 1$ and $\amuju{0}{0}{0}{1}{}{\pair{u}{v}
        : \typendm}$ by $(\tnormalpair)$. The result then holds for
      $\Gam := \es$.
      
      \item If $t = \l{p}.u$ where $u \in \M$ then $\size{t} =
        \size{u} + 1$ and, by \ih, there is a tight derivation $\tderiv_u \rhd
        \amuju{0}{0}{0}{\size{u}}{\Gam_u}{u:
          \tight}$. Then \[\inferrule*{\tderiv_u \rhd
        \amuju{0}{0}{0}{\size{u}}{\Gam_u}{u:
          \tight}}{\amuju{0}{0}{0}{\size{u} +
          1}{\cmin{\Gam_u}{\var{p}}}{\l p.u:\typendm}}\] 
     The result then holds for $\Gam :=
     \cmin{\Gam_u}{\var{p}}$. Observe that, since by \ih\
     $\istight{\Gam_u}$ holds,
     and $\cmin{\Gam_u}{\var{p}} \subseteq \Gam_u$ then
     $\istight{\cmin{\Gam_u}{\var{p}}}$ trivially holds.

     \item If $t = u\esub{\pair{p_1}{p_2}}{v}$ where $u\in \M$ and $v
       \in \N$ then $\size{t} = \size{u} + \size{v} +
       1$. Moreover, by the previous item there is a tight $\tderiv_v
    \rhd \amuju{0}{0}{0}{\size{v}}{\Gam_v}{v: \typendn}$ (so that $\istight{\Gam_v}$) and, by \ih,
    there is a tight derivation $\tderiv_u
    \rhd \amuju{0}{0}{0}{\size{u}}{\Gam_u}{u: \tight}$. Then, $\tderiv'_v
    \rhd \amuju{0}{0}{0}{\size{v}}{\Gam_v}{v: \mult{\typendn}}$ and
    \[ \inferrule*{\tderiv_u \\   (\Gam_u)|_{\pair{p}{q}} \pder^{(0,0,1)}  \pair{p}{q}: \mult{\typendn} \\ \tderiv'_v}{\amuju{0}{0}{0}{\size{u} + \size{v} +
       1}{(\cmin{\Gam_u}{\var{\pair{p_1}{p_2}}}) \inter
       \Gam_v}{u\esub{\pair{p_1}{p_2}}{v}: \tight}}\]
    The result then holds for $\Gam := (\cmin{\Gam_u}{\var{\pair{p_1}{p_2}}}) \inter
    \Gam_v$, since $\istight{\Gam_v}$ as remarked,
       and by \ih\   $\istight{\Gam_u}$, thus $\istight{\Gam}$.

     \end{itemize} 
    
  \end{itemize}
\end{proof} 

\gettoappendix{l:anti-substitution}

\begin{proof} As in the case of the substitution lemma, the proof follows by generalising the property for the two cases where the type derivation $\Phi$ assigns a type or a multi-set type. 
\begin{itemize}
\item Let $\Phi \tri  \amuju{b}{e}{m}{f}{\Gam}{t\isub{x}{u}:\sig}$. Then, there exist derivations $\Phi_t$, $\Phi_u$,
  integers $b_t, b_u, e_t, e_u, m_t, m_u, f_t, f_u$,
  contexts $\Gam_t, \Gam_u$, and multi-type $\A$ such that
  $\Phi_t \tri  \amuju{b_t}{e_t}{m_t}{f_t}{\Gam_t; x:\A}{t:\sig}$, 
  $\Phi_u \tri  \amuju{b_u}{e_u}{m_u}{f_u}{\Gam_u}{u:\A}$,
  $b=b_t+ b_u$,
  $e = e_t+ e_u$,
  $m= m_t+ m_u$,
  $f=f_t+ f_u$, and 
  $\Gam = \Gam_t \inter \Gam_u$.
\item  Let $\Phi \tri  \amuju{b}{e}{m}{f}{\Gam}{t\isub{x}{u}:\B}$. Then, there exist derivations $\Phi_t$, $\Phi_u$,
  integers $b_t, b_u, e_t, e_u, m_t, m_u, f_t, f_u$,
  contexts $\Gam_t, \Gam_u$, and multi-type $\A$ such that
  $\Phi_t \tri  \amuju{b_t}{e_t}{m_t}{f_t}{\Gam_t; x:\A}{t:\B}$, 
  $\Phi_u \tri  \amuju{b_u}{e_u}{m_u}{f_u}{\Gam_u}{u:\A}$,
  $b=b_t+ b_u$,
  $e = e_t+ e_u$,
  $m= m_t+ m_u$,
  $f=f_t+ f_u$, and 
   $\Gam = \Gam_t \inter \Gam_u$.
\end{itemize}
We will reason by induction on $\Phi$. For all the rules (except $\many$), we will have the trivial case $t\isub{x}{u}$, where $t=x$, in which case $t\isub{x}{u} = u$, for which we have a derivation $\Phi \tri \amuju{b}{e}{m}{f}{\Gam}{u:\sig}$. Therefore
$\Phi_t \tri \amuju{0}{0}{0}{0}{x:\mult{\sig}}{x:\sig}$ and $\Phi_u \tri \amuju{b}{e}{m}{f}{\Gam}{u:\mult{\sig}}$ is obtained from $\Phi$ using the $(\many)$ rule. The conditions on the counters hold trivially. We now reason on the different cases assuming that $t\neq x$.
 \begin{itemize}
 \item If $\Phi$ is $(\ax)$, therefore $\Phi \tri \amuju{0}{0}{0}{0}{y:\mult{\sig}}{y:\sig}$. We only consider the case where 
$t=y$ and $y\neq x$. Then we take $\A=\ems$, $\Phi_t \tri \amuju{0}{0}{0}{0}{y:\mult{\sig};x:\ems}{y:\sig}$ and $\Phi_u \tri \amuju{0}{0}{0}{0}{}{u:\ems}$,  using rule $(\many)$. The conditions on the counters follow trivially.
 \item  If $\Phi$ ends with $(\many)$, then $\Phi \tri \amuJu{+_{\kK}b_k}{+_{\kK}e_k}{+_{\kK} m_k}{+_{\kK}f_k}{\inter_{\kK}\Gamk}{t\isub{x}{u}: \mult{\sigk}_{\kK}}$ follows from $\Phi^k \tri \amuJu{b_k}{e_k}{m_k}{f_k}{\Gamk}{t\isub{x}{u}: \sigk}$, for each $\kK$.
By the \ih, there exist $\Phi_t^k$, $\Phi_u^k$, $b_t^k$, $b_u^k$, $e_t^k$, $e_u^k$, $m_t^k$, $m_u^k$, $f_t^k$, $f_u^k$, contexts $\Gam_t^k$,  $\Gam_u^k$ and multi-type $\A_k$, such that  $\Phi_t^k \tri \amuJu{b_t^k}{e_t^k}{m_t^k}{f_t^k}{\Gam_t^k;x:\A_k}{t:\sigk}$, $\Phi_u^k \tri \amuJu{b_u^k}{e_u^k}{m_u^k}{f_u^k}{\Gam_u^k}{u:\A_k}$, $\Gam_k = \Gam_t^k \inter \Gam_u^k$, $b_k = b_t^k + b_u^k$, $e_k = e_t^k + e_u^k$, $m_k = m_t^k + m_u^k$, $f_k = f_t^k + f_u^k$. Taking $\A = \inter_{\kK}\A_k $ and using rule $(\many)$ we get $\amuJu{+_{\kK}b_t^k}{+_{\kK}e_t^k}{+_{\kK} m_t^k}{+_{\kK}f_t^k}{\inter_{\kK}\Gam_t^k; x:\A}{t: \mult{\sigk}_{\kK}}$. From the premisses of $\Phi_u^k$ for $\kK$, applying the $(\many)$ rule, we get  
$\amuJu{+_{\kK}b_u^k}{+_{\kK}e_u^k}{+_{\kK} m_u^k}{+_{\kK}f_u^k}{\inter_{\kK}\Gam_u^k}{u: \A}$.
Note that $\Gam = \inter_{\kK}\Gamk = (\inter_{\kK}\Gam_t^k)\inter (\inter_{\kK}\Gam_u^k)$, $b=+_{\kK}b_k = (+_{\kK}b_t^k) + (+_{\kK}b_u^k)$,  $e=+_{\kK}e_k = (+_{\kK}e_t^k) + (+_{\kK}e_u^k)$,  $m=+_{\kK}m_k = (+_{\kK}m_t^k) + (+_{\kK}m_u^k)$ and  $f=+_{\kK}f_k = (+_{\kK}f_t^k) + (+_{\kK}f_u^k)$, as expected.
\item If $\Phi$ ends with $(\introarrow)$, then $t=\l p.t'$, therefore $$\Phi \tri \amuju{b_t+1}{e_t+e_p}{m_t + m_p}{f_t + f_p}{\cmin{\Gam}{\var{p}}}{\l p.(t'\isub{x}{u}): \B \rew  \sig}$$ follows from $\amuju{b_t}{e_t}{m_t}{f_t}{\Gam}{t'\isub{x}{u} : \sig}$ and $\Gam|_{p} \pder^{(e_p,m_p,f_p)}  p:\B$. Note that, one can always assume that $\var{p} \cap \fv{u} = \es$ and $x\notin \var{p}$. By the \ih, $\Phi_{t'} \tri \amuju{b_t'}{e_t'}{m_t'}{f_t'}{\Gam_{t'};x:\A}{t': \sig}$,  $\Phi_{u} \tri \amuju{b_u}{e_u}{m_u}{f_u}{\Gam_u}{u: \A}$, with $\Gam = \Gam_{t'} \inter \Gam_u$, $b_t = b_{t'} + b_u$, $e_t = e_{t'} + e_u$, $m_t = m_{t'} + m_u$, $f_t = f_{t'} + f_u$. Then using rule $(\introarrow)$ we get $\Phi_t \tri \amuju{b_{t'}+1}{e_{t'}+e_p}{m_{t'} + m_p}{f_{t'} + f_p}{\cmin{\Gam_{t'}}{\var{p}} ; x:\A}{\l p.t': \B \rew  \sig}$. And $\cmin{\Gam}{\var{p}} = (\cmin{\Gam_{t'}}{\var{p}}) \inter \Gam_u$, $b_t+1 = b_t'+1 + b_u$, $e_t = e_{t'} + e_u$, $m_t = m_{t'} + m_u$ and $f_t = f_{t'} + f_u$, as expected.
\item If $\Phi$ ends with $(\ruletight)$, the result follows by induction, similarly to the previous case.
\item If $\Phi$ ends with $(\app)$ then $t=t'u'$, then $$\Phi \tri \amuju{b_{t'}+b_{u'}}{e_{t'}+e_{u'}}{m_{t'}+m_{u'}}{f_{t'}+f_{u'}}{\Gam \inter \Del}{ t'\isub{x}{u}\,u'\isub{x}{u}: \sig}$$
  follows from $\amuju{b_{t'}}{e_{t'}}{m_{t'}}{f_{t'}}{\Gam}{t'\isub{x}{u}: \B \rew \sig}$ and 
  $\amuJu{b_{u'}}{e_{u'}}{m_{u'}}{f_{u'}}{\Del}{u'\isub{x}{u} : \B}$. By the \ih, there exist $\Phi_{t'},\Phi_{t'}^{u}$, $b_{t''}$, $b_{t'}^u$, $e_{t''}$, $e_{t'}^u$, $m_{t''}$, $m_{t'}^u$, $f_{t''}$, $f_{t'}^u$, contexts $\Gam_{t'}$, $\Gam_{t'}^u$ and multi-type $\A_{t'}$, such that
  \[
  \begin{array}{l}
    \Phi_{t'} \tri \amuju{b_{t''}}{e_{t''}}{m_{t''}}{f_{t''}}{\Gam_{t'};x:\A_{t'} }{t': \B \rew \sig}\\
    \Phi_{t'}^{u} \tri \amuju{b_{t'}^u}{e_{t'}^u}{m_{t'}^u}{f_{t'}^u}{\Gam_{t'}^u}{u: \A_{t'}}
  \end{array}
  \]
  where $b_{t'} = b_{t''} +b_{t'}^u$, $e_{t'} = e_{t''} +e_{t'}^u$, $m_{t'} = m_{t''} +m_{t'}^u$, $f_{t'} = f_{t''} +f_{t'}^u$ and $\Gam = \Gam_{t'} \inter \Gam_{t'}^u$.
  And by the \ih applied to the second premise of $\Phi$, there exist $\Phi_{u'},\Phi_{u'}^{u}$, $b_{u''}$, $b_{u'}^u$, $e_{u''}$, $e_{u'}^u$, $m_{u''}$, $m_{u'}^u$, $f_{u''}$, $f_{u'}^u$, contexts $\Gam_{u'}$, $\Gam_{u'}^u$ and multi-type $\A_{u'}$, such that
  \[
  \begin{array}{l}
     \Phi_{u'} \tri \amuju{b_{u''}}{e_{u''}}{m_{u''}}{f_{u''}}{\Del_{u'};x:\A_{u'} }{u': \B }\\
    \Phi_{u'}^{u} \tri \amuju{b_{u'}^u}{e_{u'}^u}{m_{u'}^u}{f_{u'}^u}{\Del_{u'}^u}{u: \A_{u'}}
  \end{array}
  \]
  where $b_{u'} = b_{u''} +b_{u'}^u$, $e_{u'} = e_{u''} +e_{u'}^u$, $m_{u'} = m_{u''} +m_{u'}^u$, $f_{u'} = f_{u''} +f_{u'}^u$ and $\Del = \Del_{u'} \inter \Del_{u'}^u$.
  Taking $\A = \A_{t'} \inter\A_{u'}$, by the $(\app)$ rule, one gets $\Phi_{t'u'} \tri \amuju{b_{t''}+b_{u''}}{e_{t''}+e_{u''}}{m_{t''}+m_{u''}}{f_{t''}+f_{u''}}{\Gam_{t'}\inter \Del_{u'};x:\A_{t'} \inter\A_{u'} }{t': \B \rew \sig}$ and applying the $(\many)$ rule to the premisses of $\Phi_{t'}^{u}$ and $\Phi_{u'}^{u}$ one gets $\Phi_{u} \tri \amuju{b_{t'}^u+b_{u'}^u}{e_{t'}^u+e_{u'}^u}{m_{t'}^u+m_{u'}^u}{f_{t'}^u+f_{u'}^u}{\Gam_{t'}^u \inter \Del_{u'}^u}{u: \A}$. Note that $\Gam \inter \Del = (\Gam_{t'} \inter \Gam_{t'}^u) \inter (\Del_{u'} \inter \Del_{u'}^u) = (\Gam_{t'} \inter \Del_{u'}) \inter (\Gam_{t'}^u \inter \Del_{u'}^u)$ and $b = b_{t'} + b_{u'} = (b_{t''} +b_{t'}^u)+(b_{u''} +b_{u'}^u) = (b_{t''} +b_{u''})+(b_{t'}^u +b_{u'}^u)$ as expected (the same happens for the remaining counters).
\item If $\Phi$ ends with $(\appn)$ the result follows from a straightforward induction\item If $\Phi_t$ ends with  $(\trsub)$, so that $t=t'[p/u']$, then the result follows a similar reasoning to the case for application, since $t'[p/u']\isub{x}{u} = (t'\isub{x}{u})[p/u'\isub{x}{u}]$.
  \item If $\Phi$ ends with $(\trpair)$ or $(\tnormalpair)$, so that $t=\pair{t'}{u'}$, then we have two cases. The case for $\tnormalpair$, follows from $\Phi$ being of the form $\amuju{0}{0}{0}{1}{}{\pair{t'\isub{x}{u}}{u'\isub{x}{u}}:\typendm}$. We then take $\A =\ems$, $\Phi_{\pair{t'}{u'}} \tri \amuju{0}{0}{0}{1}{x:\ems}{\pair{t'}{u'}:\typendm}$ and $\Phi_u \tri \amuju{0}{0}{0}{0}{\empty}{u:\mult{\ }}$ follows trivially from the $(\many)$ rule. Then conditions on counters and contexts hold trivially. The case for $(\trpair)$ follows by induction using the same reasoning as in rule $(\app)$. 
 \end{itemize}

\end{proof}

\gettoappendix{l:subject-expansion}
\begin{proof} By induction on $\Rewbis{\head}$, using Lem.~\ref{l:anti-substitution}.
\begin{itemize}
\item $t=\appctx{\slist}{\l p. v}  u \Rew{\head} \appctx{\slist}{v \esub{p}{u}} =t'$.
  The proof is by induction on the list context $\slist$.\\
  If $\slist = \Box$ then by construction there are $\Phi_v \rhd \amuju{b_v}{e_v}{m_v}{f_v}{\Gam_v}{v :\sig}$ and  $\Phi_u \rhd \amuju{b_u}{e_u}{m_u}{f_u}{\Gam_u}{u :\A}$ for some multi-type $\A$, and $\Phi'$ is of the form
  \[\inferrule*{\Phi_v \\ (\Gam_v)|_{p} \pder^{(e_p,m_p,f_p)} p:\A \\ \Phi_u}{\amuju{b'}{e'}{m'}{f}{(\cmin{\Gam_v}{\var{p}}) \inter \Gam_u}{v \esub{p}{u}:\sig}}\]
  where $\Gam = (\cmin{\Gam_v}{\var{p}}) \inter \Gam_u$, $b' = b_v + b_u$, $e' = e_v + e_u + e_p$, $m' = m_v + m_u + m_p$ and $f = f_v + f_u + f_p$. Then
  \[\Phi \tri \inferrule{ \inferrule*{ \Phi_v \\  (\Gam_v)|_{p} \pder^{(e_p,m_p,f_p)} p:\A}{\amuju{b_v + 1}{e_v + e_p}{m_v + m_p}{f_v + f_p}{\cmin{\Gam_v}{\var{p}}}{\l{p}.v :\A \rew \sig}} \\ \Phi_u }{\amuju{(b_v + 1) + b_u}{(e_v + e_p) + e_u}{(m_v + m_p) + m_u}{(f_v + f_p) + f_u}{ (\cmin{\Gam_v}{\var{p}}) \inter \Gam_u}{(\l p. v) u :\sig}}\] where $b = (b_v + 1) + b_u = b' + 1$, $e = (e_v + e_p) + e_u = e'$ and $m = (m_v + m_p) + m_u = m'$.

  If $\slist \neq \Box$ the proof from the \ih is straightforward.

\item  $t= v\esub{x}{u}   \Rew{\head}  v\isub{x}{u}=t'$, where $v \notRew{\head}$. Then by Lem.~\ref{l:anti-substitution}  there exist derivations $\Phi_v$, $\Phi_u$, integers $b_v, b_u, e_v, e_u, m_v, m_u, f_v, f_u$, contexts $\Gam_v, \Gam_u$, and multi-type $\A$ such that $\Phi_v \tri  \amuju{b_v}{e_v}{m_v}{f_v}{\Gam_v; x:\A}{v:\sig}$, $\Phi_u \tri  \amuju{b_u}{e_u}{m_u}{f_u}{\Gam_u}{u:\A}$,
  $b'=b_v+ b_u$, $e' = e_v+ e_u$, $m'= m_v+ m_u$,
  $f=f_v+ f_u$, and $\Gam = \Gam_v \inter \Gam_u$. Then,
  \[\Phi \tri \inferrule{ \Phi_v \\  (\Gam_v; x:\A)|_{x} \pder^{(1,0,0)} x:\A\\ \Phi_u}{\amuju{b_v + b_u}{e_v + e_u + 1}{m_v + m_u + 0}{f_v + f_u + 0}{\Gam_v \inter \Gam_u}{v\esub{x}{u}: \sig}}\]
where $b = b_v + b_u = b'$, $e = e_v + e_u + 1 = e' + 1$ and $m = m_v + m_u = m'$. Note that $\cmin{(\Gam_v;x:\A)}{\var{x}} = \Gam_v$.  

\item  $t= v\esub{\pair{p_1}{p_2}}{\appctx{\slist}{\pair{u_1}{u_2}}}
      \Rew{\head} \appctx{\slist}{
        v\esub{p_1}{u_1}\esub{p_2}{u_2}}=t'$, where $v \notRew{\head}$.
    Let abbreviate $p = \pair{p_1}{p_2}$ and $u = \pair{u_1}{u_2}$. 
      The proof is by induction on the list $\slist$.
\begin{itemize}
    \item $\slist = \Box$, then by construction there are $\Phi_v \tri \amuju{b_v}{e_v}{m_v}{f_v}{\Gam_v}{v:\sig}$, $\Phi_1 \tri \amuju{b_u^1}{e_u^1}{m_u^1}{f_u^1}{\Gam_1}{u_1:\A_1}$ and $\Phi_{2} \tri \amuju{b_u^2}{e_u^2}{m_u^2}{f_u^2}{\Gam_2}{u_2:\A_2}$ where
\[\Phi_{v\esub{p_1}{u_1}} \tri \inferrule{\Phi_v \\  \Gam_v|_{p_1} \pder^{(e_1,m_1,f_1)} p_1:\A_1 \\ \Phi_1 }{\amuju{b_v + b_u^1}{e_v + e_u^1 + e_1}{m_v + m_u^1 + m_1}{f_v + f_u^1 + f_1}{(\cmin{\Gam_v}{\var{p_1}}) \inter  \Gam_1}{v\esub{p_1}{u_1}:\sig}}\]
and
  \[\Phi' \tri \inferrule{ \Phi_{v\esub{p_1}{u_1}} \\ ((\cmin{\Gam_v}{\var{p_1}}) \inter  \Gam_1)|_{p_2} \pder^{(e_2,m_2,f_2)} p_2:\A_2 \\ \Phi_2}{\amuju{b'}{e'}{m'}{f}{\Gam}{v\esub{p_1}{u_1}\esub{p_2}{u_2}:\sig}}\]
  where $b' = b_v +_{i=1,2} b_u^i$, $e' = e_v +_{i=1,2} e_u^i + e_i$, $m' = m_v  +_{i=1,2} m_u^i + m_i$, $f = f_v + _{i=1,2} f_u^i + f_i$ and $\Gam = (\cmin{((\cmin{\Gam_v}{\var{p_1}}) \inter  \Gam_1)}{\var{p_2}}) \inter \Gam_2$. Moreover, $(\cmin{((\cmin{\Gam_v}{\var{p_1}}) \inter  \Gam_1)}{\var{p_2}}) = (\cmin{(\cmin{\Gam_v}{\var{p_1}})}{\var{p_2}} \inter  \cmin{\Gam_1}{\var{p_2}}
  )$, where $\cmin{(\cmin{\Gam_v}{\var{p_1}})}{\var{p_2}} = \cmin{\Gam_v}{\var{\pair{p_1}{p_2}}}$ and $\cmin{\Gam_1}{\var{p_2}} =^{Lem. \ref{l:relevance+clash-free}:\ref{relevance}} \Gam_1$. Similarly, $ ((\cmin{\Gam_v}{\var{p_1}}) \inter  \Gam_1)|_{p_2} =^{Lem. \ref{l:relevance+clash-free}:\ref{relevance}}
(\cmin{\Gam_v}{\var{p_1}})|_{p_2}$ and, by linearity of patterns, $ (\cmin{\Gam_v}{\var{p_1}})|_{p_2} =\Gam_v|_{p_2}$. Hence,
 \[\Phi_{\pair{p_1}{p_2}} \tri \inferrule{ \Gam_v|_{p_1} \pder^{(e_1,m_1,f_1)} p_1:\A_1 \\ \Gam_v|_{p_2} \pder^{(e_2,m_2,f_2)} p_2:\A_2 }{\Gam_v|_{p_1} \inter  \Gam_v|_{p_2} \pder^{(e_1+e_2,1+m_1+m_2,f_1+f_2)} \pair{p_1}{p_2} :
     \mult{\prodt{\A_1}{\A_2}} }\] where $\Gam_v|_{p_1} \inter  \Gam_v|_{p_2} = \Gam_v|_{\pair{p_1}{p_2}}$, and $\Phi_{\pair{u_1}{u_2}} \tri \amuju{b_u^1+b_u^2}{e_u^1+e_u^2}{m_u^1+m_u^2}{f_u^1+f_u^2}{\Gam_1 \inter \Gam_2}{\pair{u_1}{u_2} : \mult{\prodt{\A_1}{\A_2}}}$. Therefore,
 \[\Phi \tri \inferrule{ \Phi_v \\ \Phi_{\pair{p_1}{p_2}} \\ \Phi_{\pair{u_1}{u_2}} }{\amuju{b}{e}{m}{f}{(\cmin{\Gam_v}{\var{\pair{p_1}{p_2}}}) \inter  (\Gam_1 \inter \Gam_2)}{v\esub{\pair{p_1}{p_2}}{\pair{u_1}{u_2}}:\sig}}\] where $b = b_v +_{i=1,2} b_u^i = b'$, $e' = e_v +_{i=1,2} e_u^i + e_i = e'$ and $m = 1 + m_v  +_{i=1,2} m_u^i + m_i = m' + 1$.

  \item If $\slist = \slist'\esub{q}{s}$, then $\Phi'$ is of the form:
  \[
  \begin{prooftree}
    \amuju{b'_l}{e'_l}{m'_l}{f'_l}{\Gam_{\slist'}}{\appctx{\slist'}{v\esub{p_1}{u_1}\esub{p_2}{u_2}}: \sig} \sep \Gam_{\slist'}|_{q} \pder^{(e_q,m_q,f_q)} q :\A \sep \amuju{b_s}{e_s}{m_s}{f_s}{\Gam_s}{s:\A}  
    \justifies{\amuju{b'_l + b_s}{e'_l + e_s + e_q}{m'_l + m_s + m_q}{f'_l + f_s + f_q}{\cmin{\Gam_{\slist'}}{\var{q}} \inter \Gam_s}{\appctx{\slist'}{v\esub{p_1}{u_1}\esub{p_2}{u_2}}\esub{q}{s} : \sig} }
  \end{prooftree}
  \]
  where $b' = b'_l + b_s$, $e' = e'_l + e_s + e_q$, $m' = m'_l + m_s + m_q$, $f' = f'_l + f_s + f_q$ and $\Gam = \cmin{\Gam_{\slist'}}{\var{q}} \inter \Gam_s$. From $v\esub{\pair{p_1}{p_2}}{\appctx{\slist'}{\pair{u_1}{u_2}}}
  \Rew{\head} \appctx{\slist'}{v\esub{p_1}{u_1}\esub{p_2}{u_2}}$
  and derivation $\Phi'_{\slist'}$ for the leftmost premise by the \ih one gets $\Phi_{\slist'} \tri
  \amuju{b_l}{e_l}{m_l}{f_l}{\Gam_{\slist'}}{v\esub{\pair{p_1}{p_2}}{\appctx{\slist'}{\pair{u_1}{u_2}}}:
    \sig}$ where $b_l = b'_l$, $e_l = e'_l$, $m_l = m'_l + 1$ and $f_l = f'_l$. Furthermore $\Phi_{\slist'}$ is necessarily of the form:
{\small
  \[
  \begin{prooftree}
    \amuju{b_v}{e_v}{m_v}{f_v}{\Gam_v}{v:\sig} \sep \Gam_v|_{p} \pder^{(e_p,m_p,f_p)} p : \mult{\prodt{\A_1}{\A_2}}\sep
    \begin{prooftree}
      \amuju{b'_u}{e'_u}{m'_u}{f'_u}{\Gam_u}{\appctx{\slist'}{u}: \prodt{\A_1}{\A_2}}
      \justifies{\amuju{b'_u}{e'_u}{m'_u}{f'_u}{\Gam_u}{\appctx{\slist'}{u}: \mult{\prodt{\A_1}{\A_2}}}}
    \end{prooftree}
    \justifies{\amuju{b_v + b'_u}{e_v + e'_u + e_p}{m_v + m'_u + m_p}{f_v + f'_u + f_p}{\cmin{\Gam_{v}}{\var{p}} \inter \Gam_u}{v\esub{p}{\appctx{\slist'}{u}}: \sig}}
  \end{prooftree}
  \]}
  where $b_l = b_v + b'_u$, $e_l = e_v + e'_u + e_p$, $m_l = m_v + m'_u + m_p$, $f_l = f_v + f'_u + f_p$ and $\Gam_{\slist'} = \cmin{\Gam_{v}}{\var{p}} \inter \Gam_u$. Note that, by relevance and $\alpha$-conversion we have that $\Gam_{\slist'} |_{q} = \Gam_{u}|_{q}$. Then one can construct the following derivation $\Phi_u$:
  \[
  \begin{prooftree}
    \begin{prooftree}
      \amuju{b'_u}{e'_u}{m'_u}{f'_u}{\Gam_u}{\appctx{\slist'}{u}: \prodt{\A_1}{\A_2}} \sep
      \Gam_u|_{q} \pder^{(e_q,m_q,f_q)} q: \A \sep \Phi_s
      \justifies{
        \amuju{b'_u + b_s}{e'_u + e_s + e_q}{m'_u + m_s + m_q}{f'_u + f_s + f_q}{\cmin{\Gam_u}{\var{q}} \inter \Gam_s}{\appctx{\slist'\esub{q}{s}}{u} : \prodt{\A_1}{\A_2}}}
    \end{prooftree}
      \justifies{\amuju{b'_u + b_s}{e'_u + e_s + e_q}{m'_u + m_s + m_q}{f'_u + f_s + f_q}{\cmin{\Gam_u}{\var{q}} \inter \Gam_s}{\appctx{\slist'\esub{q}{s}}{u} : \mult{\prodt{\A_1}{\A_2}}}}
    \end{prooftree}
  \]
  where $b_u = b'_u + b_s$, $e_u = e'_u + e_s + e_q$, $m_u = m'_u + m_s + m_q$ and $f_u = f'_u + f_s + f_q$. From which we build $\Phi$:
  {\small
  \[
  \begin{prooftree}
    \amuju{b_v}{e_v}{m_v}{f_v}{\Gam_v}{v:\sig} \sep \Gam_v|_{p} \pder^{(e_p,m_p,f_p)} p : \mult{\prodt{\A_1}{\A_2}}\sep
    \amuju{b_u}{e_u}{m_u}{f_u}{\cmin{\Gam_u}{\var{q}} \inter \Gam_s }{\appctx{\slist'\esub{q}{s}}{u} : \mult{\prodt{\A_1}{\A_2}}} 
    \justifies{\amuju{b_v + b_u}{e_v + e_u + e_p}{m_v + m_u + m_p}{f_v + f_u + f_p}{\cmin{\Gam_v}{\var{p}} \inter \cmin{\Gam_u}{\var{q}} \inter \Gam_s }{v\esub{p}{\appctx{\slist'\esub{q}{s}}{u}}:\sig}}
  \end{prooftree}
  \]}
  With $\cmin{\Gam_v}{\var{p}} \inter \cmin{\Gam_u}{\var{q}} \inter \Gam_s = \cmin{(\cmin{\Gam_v}{\var{p}} \inter \Gam_u)}{\var{q}} \inter \Gam_s = \Gam$.
  Furthermore,
  \begin{itemize}
    \item $b = b_v + b_u = b_v + b'_u + b_s = b_l + b_s = b'$
    \item $e = e_v + e_u + e_p = e_v + e'_u + e_s + e_q + e_p = e_l + e_s + e_q = e'$
    \item $m = m_v + m_u + m_p = m_v + m'_u + m_s + m_q + m_p = m_l + m_s + m_q = m' + 1$
    \item $f = f_v + f_u + f_p =  f_v + f'_u + f_s + f_q + f_p = f_l + f_s + f_q = f'$  
  \end{itemize} 
\end{itemize}  
   
 \item $t= \l p. u \Rew{\head} \l p. u' =t'$, where
      $u \Rew{\head} u'$. This case is straightforward by the \ih\
    \item $t= vu \Rew{\head} v'u =t'$, where
      $v \Rew{\head} v'$ and $\isnotabs{v}$. This case is straightforward by the \ih\
    \item $t= v\esub{p}{u} \Rew{\head} v'\esub{p}{u} =t'$, where
      $v \Rew{\head} v'$. This case is straightforward by the \ih\
    \item $t= v\esub{p}{u} \Rew{\head} v\esub{p}{u'} =t'$, where
      $v \notRew{\head}$ and $p\neq x$ and $u \Rew{\head} u'$. By construction there are subderivations $\Phi_v \tri \amuju{b_v}{e_v}{m_v}{f_v}{\Gam_v}{v:\sig}$,  $\Gam_v|_{p} \pder^{(e_p,m_p,f_p)} p:\A$ and $\Phi_{u'} \tri \amuju{b_{u'}}{e_{u'}}{m_{u'}}{f_{u'}}{\Gam_{u'}}{u':\A}$ for some multi-set $\A$ and $\Gam = (\cmin{\Gam_v}{\var{p}}) \inter \Gam_{u'}$. Since $p$ is not a variable then $ \Phi_p$ ends with rule $(\pattn)$ or $(\trpairpat)$. In both cases $\A$ contains
      only one type, let us say $\A = \mult{\sig_{u'}}$. Then $\Phi_{u'}$ has the following form
      \[ \Phi_{u'} \tri \inferrule{\amuJu{b_{u'}}{e_{u'}}{m_{u'}}{f_{u'}}{\Gam_{u'}}{u':\sig_{u'}}}
               {\amuJu{b_{u'}}{e_{u'}}{m_{u'}}{f_{u'}}{\Gam_{u'}}{u':\mult{\sig_{u'}}}}\]
             The \ih\ applied to the premise of $\Phi_{u'}$ gives
             a derivation 
             $\amuJu{b_u}{e_u}{m_u}{f_u}{\Gam_{u'}}{u:\sig_{u'}}$
             and having the expected counters. 
             To conclude we build a type derivation  $\Phi'$ for
             $v\esub{p}{u'}$ having the expected counters. 
\end{itemize} 

\end{proof}

\end{document}